\documentclass[11pt,onecolumn,english]{IEEEtran}
\usepackage[T1]{fontenc}
\usepackage[utf8]{inputenc}
\usepackage{color}
\usepackage{babel}
\usepackage{mathrsfs}
\usepackage{amsmath}
\usepackage{amssymb}
\usepackage{graphicx}
\usepackage{esint}

\usepackage{multicol,pdfpages,enumerate}
\usepackage{arydshln}

\usepackage[unicode=true,
 bookmarks=true,bookmarksnumbered=true,bookmarksopen=true,bookmarksopenlevel=1,
 breaklinks=false,pdfborder={0 0 0},backref=false,colorlinks=false]
 {hyperref}
\hypersetup{pdftitle={Your Title},
 pdfauthor={Your Name},
 pdfpagelayout=OneColumn,pdfnewwindow=true,pdfstartview=XYZ,plainpages=false}

\makeatletter

\providecommand{\tabularnewline}{\\}

\usepackage{comment}\usepackage{amsthm}\usepackage{latexsym}\usepackage{cite}\usepackage{pstricks}\usepackage{float}\usepackage{pdfpages}\usepackage{babel}\usepackage{amsthm}

\theoremstyle{plain}

\theoremstyle{plain}

\providecommand{\corollaryname}{Corollary}
\providecommand{\theoremname}{Theorem}

\theoremstyle{remark}
\newcounter{rem}
\newtheorem{remark}[rem]{\bf Remark}

\theoremstyle{plain}
\ifCLASSOPTIONcompsoc
\usepackage[caption=false,font=normalsize,labelfont=sf,textfont=sf]{subfig}\else
\usepackage[caption=false,font=footnotesize]{subfig}\fi

\@ifundefined{showcaptionsetup}{}{%
 \PassOptionsToPackage{caption=false}{subfig}}
\usepackage{subfig}

\allowdisplaybreaks

\newcommand{\suppress}[1]{}
\IEEEoverridecommandlockouts

\newtheorem{theorem}{Theorem}

\newtheorem{lemma}{Lemma}

\usepackage{algorithmic}\usepackage{algorithm2e}\usepackage{multicol}\usepackage{pdfpages}\usepackage{enumerate}

\@ifundefined{definecolor}
 {\usepackage{color}}{}
\DeclareSymbolFont{lettersA}{U}{txmia}{m}{it}
 \DeclareMathSymbol{\bbr}{\mathord}{lettersA}{"92}
 \DeclareMathSymbol{\bbc}{\mathord}{lettersA}{"83}
 \DeclareMathSymbol{\bbn}{\mathord}{lettersA}{"8E}
 \DeclareMathSymbol{\bbg}{\mathord}{lettersA}{"87}

\DeclareMathAlphabet{\mathscr}{OT1}{pzc}{m}{it}
\newcommand{\proofof}[1]{\phantom{aaaaaaaa}{\em Proof of #1:}}
\def\QEDmark{\blacksquare}
\def\endproof{\hfill\QEDmark}

\newcommand{\bbR}{\mathbb{R}}
\newcommand{\bbF}{\mathbb{F}}
\newcommand{\graph}{\mathcal{G}}
\newcommand{\preci}{P}
\newcommand{\bA}{{{A}}}

\newcommand{\samp}{{ n}}
\newcommand{\spar}{{ k}}
\newcommand{\degr}{{ 3}}
\newcommand{\dgr}{{ d}}
\newcommand{\meas}{{m}}
\newcommand{\cnst}{{ c}}

\newcommand{\bet}{{\beta}}

\newcommand{\cO}{{\mathcal{O}}}
\newcommand{\bx}{{\mathbf{x}}}
\newcommand{\by}{{\mathbf{y}}}
\newcommand{\bz}{{\mathbf{z}}}
\newcommand{\be}{{\mathbf{e}}}
\newcommand{\cSx}{{\mathcal{S}}(\bx)}
\newcommand{\cS}{{\mathcal{S}}}
\newcommand{\cSpx}{{\mathcal{S}}'(\bx)}
\newcommand{\steps}{{\Gamma}}
\newcommand{\step}{{\gamma}}
\newcommand{\digit}{{ g}}
\newcommand{\cD}{{\mathcal{D}}}
\newcommand{\iter}{{t}}
\newcommand{\ofiter}{{(\iter)}}
\newcommand{\ofiterplusone}{{(\iter+1)}}
\newcommand{\xj}{{j}}
\newcommand{\yi}{{i}}

\newcommand{\el}{r}
\newcommand{\eL}{R}
\newcommand{\set}[1]{\left\{#1\right\}}
\newcommand{\cE}{\mathcal{E}}
\makeatother

\begin{document}

\title{SHO-FA: Robust compressive sensing with order-optimal complexity,
measurements, and bits}

\author{\IEEEauthorblockN{Mayank Bakshi, \and Sidharth Jaggi, \and Sheng
Cai, \and Minghua Chen}\\
 \texttt{\small \{mayank,jaggi,cs010,minghua\}@ie.cuhk.edu.hk}{\small }\\
{\small{} \IEEEauthorblockA{The Chinese University of Hong Kong,
Hong Kong SAR China}}}

\maketitle
\global\long\def\hbx{{\hat{\bx}}}
 \global\long\def\bxs{{{\bx_{\spar}}^{\ast}}}
 \global\long\def\sz{{\color{red}{\sigma_{z}^{2}}}}
 \global\long\def\se{{\color{red}{\sigma_{e}^{2}}}}

\begin{abstract}
Suppose $\bx$ is any exactly $\spar$-sparse vector in $\bbr^{\samp}$.
We present a class of {}``sparse'' matrices $\bA$, and a corresponding
algorithm that we call SHO-FA (for Short and Fast%
\footnote{Also, SHO-FA sho good! In fact, it's all $\cO(\spar)$!%
}) that, with high probability over $\bA$, can reconstruct $\bx$
from $\bA\bx$. The SHO-FA algorithm is related to the Invertible
Bloom Lookup Tables (IBLTs) recently introduced by Goodrich \textit{et
al.}, with two important distinctions -- SHO-FA relies on linear measurements,
and is robust to noise. The SHO-FA algorithm is the first to simultaneously
have the following properties: (a) it requires only $\cO(\spar)$
measurements, (b) the bit-precision of each measurement and each arithmetic
operation is $\cO\left(\log(\samp)+\preci\right)$ (here $2^{-\preci}$
corresponds to the desired relative error in the reconstruction of
$\bx$), (c) the computational complexity of decoding is $\cO(\spar)$
arithmetic operations and that of encoding is $\cO(\samp)$
arithmetic operations, and (d) if the reconstruction goal is simply
to recover a single component of $\bx$ instead of all of $\bx$,
with significant probability over $\bA$ this can be done in constant time.
All constants above are independent of all problem parameters other
than the desired probability of success. For a wide range of parameters
these properties are information-theoretically order-optimal. In addition, our SHO-FA algorithm works over fairly general ensembles of ``sparse random matrices'', is robust to random noise, and (random) approximate sparsity for a large range of $\spar$. In particular, suppose the
measured vector equals $\bA(\bx+\bz)+\be$, where $\bz$ and $\be$
correspond respectively to the \textit{source tail} and \textit{measurement
noise}. Under reasonable statistical assumptions on $\bz$ and $\be$
our decoding algorithm reconstructs $\bx$ with an estimation error
of $\cO(||\bz||_{2}+||\be||_{2})$. The SHO-FA algorithm
works with high probability over $\bA$, $\bz$, and $\be$, and still
requires only $\cO(\spar)$ steps and $\cO(\spar)$ measurements over
$\cO(\log(\samp))$-bit numbers.  This is in contrast to most existing
algorithms which focus on the {}``worst-case'' $\bz$
model, where it is known $\Omega(\spar\log(\samp/\spar))$ measurements
over $\cO(\log(\samp))$-bit numbers are necessary. Our algorithm
has good empirical performance, as validated by simulations.\footnote{A preliminary version of this work was presented in~\cite{BakJCC:12}. In parallel and independently of this work, an algorithm with very similar design and performance was proposed and presented at the same venue in~\cite{PawR:12}.}
\end{abstract}

\section{Introduction}

\label{sec:intr} In recent years, spurred by the seminal work on
\textit{compressive sensing} of~\cite{CanRT:06,Don:06}, much attention
has focused on the problem of reconstructing a length-$\samp$ {}``compressible''
vector $\bx$ over $\bbr$ with fewer than $\samp$ linear measurements.
In particular, it is known (\textit{e.g.}~\cite{Can:08,BarDDW:08})
that with $\meas=\cO(\spar\log(\samp/\spar))$ linear measurements
one can computationally efficiently%
\footnote{The caveat is that the reconstruction techniques require one to solve
an LP. Though polynomial-time algorithms to solve LPs are known, they
are generally considered to be impractical for large problem instances.%
} obtain a vector $\hbx$ such that the \textit{reconstruction error}
$||\bx-\hbx||_{1}$ is $\cO(||\bx-\bxs||_{1})$,%
\footnote{\label{foot:recon}In fact this is the so-called ${\mathscr{l}}_{1}<C{\mathscr{l}}_{1}$
guarantee. One can also prove stronger ${\mathscr{l}}_{2}<C{\mathscr{l}}_{1}/\sqrt{\spar}$
reconstruction guarantees for algorithms with similar computational
performance, and it is known that a ${\mathscr{l}}_{2}<C{\mathscr{l}}_{2}$
reconstruction guarantee is not possible if the algorithm is required
to be zero-error~\cite{CohDD:09}, but is possible if some (small)
probability of error is allowed~\cite{GilLPS:10,PriW:11}.%
} where $\bxs$ is the best possible $\spar$-sparse approximation
to $\bx$ (specifically, the $\spar$ non-zero terms of $\bxs$ correspond
to the $\spar$ largest components of $\bx$ in magnitude, hence $\bx-\bxs$
corresponds to the {}``tail'' of $\bx$). A number of different
classes of algorithms are able to give such performance, such as those
based on ${\mathscr{l}}_{1}$-optimization (\textit{e.g.}~\cite{CanRT:06,Don:06}),
and those based on iterative {}``matching pursuit'' (\textit{e.g.}~\cite{TroG:07,DonTDS:12}).
Similar results, with an additional additive term in the reconstruction
error hold even if the linear measurements themselves also have noise
added to them (\textit{e.g.}~\cite{Can:08,BarDDW:08}). The fastest
of these algorithms use ideas from the theory of expander graphs,
and have running time $\cO(\samp\log(\samp/\spar))$~\cite{BerIR:08,BerI:09,GilI:10}.

The class of results summarized above are indeed very strong -- they
hold for \textit{all} $\bx$ vectors, including those with {}``worst-case
tails'', \textit{i.e.} even vectors where the components of $\bx$
smaller than the $\spar$ largest coefficients (which can be thought
of as {}``source tail'') are chosen in a maximally worst-case manner.
In fact~\cite{BaIPW:10} proves that to obtain a reconstruction error
that scales linearly with the ${\mathscr{l}}_{1}$-norm of the $\bz$,
the tail of $\bx$, requires $\Omega(\spar\log(\samp/\spar))$ linear
measurements.

\noindent \textbf{Number of measurements:} However, depending on the
application, such a lower bound based on {}``worst-case $\bz$''
may be unduly pessimistic. For instance, it is known that if $\bx$
is exactly $\spar$-sparse (has exactly exactly $\spar$ non-zero
components, and hence $\bz=\mathbf{0}$), then based on Reed-Solomon
codes~\cite{ReeS:60} one can efficiently reconstruct $\bx$ with
$\cO(\spar)$ noiseless measurements (\textit{e.g.}~\cite{ParH:08})
via algorithms with decoding time-complexity $\cO(\samp\log(\samp))$,
or via codes such as in~\cite{KudP:10,MitG:12} with $\cO(\spar)$
noiseless measurements with decoding time-complexity $\cO(\samp)$.%
\footnote{In general the linear systems produced by Reed-Solomon codes are ill-conditioned,
which causes problems for large $\samp$.%
} In the regime where $\spar=\theta(\samp)$~\cite{JafWHC:09} use
the {}``sparse-matrix'' techniques of~\cite{BerIR:08,BerI:09,GilI:10}
to demonstrate that $\cO(\spar)=\cO(\samp)$ measurements suffice
to reconstruct $\bx$.

\noindent \textbf{Noise:} Even if the source is not exactly $\spar$-sparse,
a spate of recent work has taken a more information-theoretic view
than the coding-theoretic/worst-case point-of-view espoused by much
of the compressive sensing work thus far. Specifically, suppose the
length-$\samp$ source vector is the sum of \textit{any} exactly $\spar$-sparse
vector $\bx$ and a {}``random'' source noise vector $\bz$ (and
possibly the linear measurement vector $\bA(\bx+\bz)$ also has a
{}``random'' measurement noise vector $\be$ added to it). Then
as long as the noise variances are not {}``too much larger''
than the signal power, the work of~\cite{AkcT:10} demonstrates that
$\cO(\spar)$ measurements suffice (though the proofs in~\cite{AkcT:10}
are information-theoretic and existential -- the corresponding {}``typical-set
decoding'' algorithms require time exponential in $\samp$). Indeed,
even the work of~\cite{BaIPW:10}, whose primary focus was to prove
that $\Omega(\spar\log(\samp/\spar))$ linear measurements are necessary
to reconstruct $\bx$ in the worst case, also notes as an aside that
if $\bx$ corresponds to an exactly sparse vector plus random noise,
then in fact $\cO(\spar)$ measurements suffice. 
The work in~\cite{WuV:10,WuV:11} examines this phenomenon information-theoretically
by drawing a nice connection with the \textit{R\'enyi information dimension
$\bar{d}(X)$} of the signal/noise. 
The heuristic algorithms in~\cite{KrzMSSZ:12} indicate that {\it approximate message passing}
algorithms achieve this performance computationally efficiently (in time $\cO(\samp^2)$), and~\cite{DonJM:11} prove this rigorously. Corresponding lower bounds showing $\Omega(\spar\log(\samp/\spar))$
samples are required in the higher noise regime are provided in~\cite{FleRG:09,Wai:09}.

\noindent \textbf{Number of measurement bits:} However, most of the
works above focus on minimizing the number of linear measurements
in $\bA\bx$, rather than the more information-theoretic view of trying
to minimize the number of bits in $\bA\bx$ over all measurements.
Some recent work attempts to fill this gap -- notably {}``Counting
Braids''~\cite{LuMPDK:08,YiMP:08} (this work uses {}``multi-layered
non-linear measurements''), and {}``one-bit compressive sensing''~\cite{PlaV:11,JacLBB:11}
(the corresponding decoding complexity is somewhat high (though still
polynomial-time) since it involves solving an LP).

\noindent \textbf{Decoding time-complexity:} The emphasis of the discussion
thus far has been on the number of linear measurements/bits required
to reconstruct $\bx$. The decoding algorithms in most of the works
above have decoding time-complexities%
\footnote{For ease of presentation, in accordance with common practice in the
literature, in this discussion we assume that the time-complexity
of performing a single arithmetic operation is constant. Explicitly
taking the complexity of performing finite-precision arithmetic into
account adds a multiplicative factor (corresponding to the precision
with which arithmetic operations are performed) in the time-complexity
of most of the works, including ours.%
} that scale at least linearly with $\samp$. In regimes where $\spar$
is significantly smaller than $\samp$, it is natural to wonder whether
one can do better. Indeed, algorithms based on iterative techniques
answer this in the affirmative. These include Chaining Pursuit~\cite{GilSTV06},
group-testing based algorithms~\cite{CorM:06}, and Sudocodes~\cite{SarBB:06}
-- each of these have decoding time-complexity that can be sub-linear
in $\samp$ (but at least $\cO(\spar\log(\spar)\log(\samp))$), but
each requires at least $\cO(\spar\log(\samp))$ linear measurements.

\noindent \textbf{Database query:} Finally, we consider a \textit{database
query} property that is not often of primary concern in the compressive
sensing literature. That is, suppose one is given a compressive sensing
algorithm that is capable of reconstructing $\bx$ with the desired
reconstruction guarantee. Now suppose that one instead wishes to reconstruct,
with reasonably high probability, just {}``a few'' (constant number)
\textit{specific} components of $\bx$, rather than all of it. Is
it possible to do so even faster (say in constant time) -- for instance,
if the measurements are in a database, and one wishes to query it
in a computationally efficient manner? If the matrix $\bA$ is {}``dense''
(most of its entries are non-zero) then one can directly see that
this is impossible. However, several compressive sensing algorithms
(for instance~\cite{JafWHC:09}) are based on {}``sparse'' matrices
$\bA$, and it can be shown that in fact these algorithms do indeed
have this property {}``for free'' (as indeed does our algorithm),
even though the authors do not analyze this. As can be inferred from
the name, this database query property is more often considered in
the database community, for instance in the work on IBLTs~\cite{GooM:11}.

\subsection{Our contributions}

Conceptually, the {}``iterative decoding'' technique we use is not
new. Similar ideas have been used in various settings in, for instance
~\cite{Spi:95,Pri:11,GooM:11,KudP:10}. However, to the best of our
knowledge, no prior work has the same performance as our work -- namely
-- information-theoretically order-optimal number of measurements,
bits in those measurements, and time-complexity, for the problem of
reconstructing a sparse signal (or sparse signal with a noisy tail
and noisy measurements) via linear measurements (along with the database
query property).%
The key to this performance is our novel design of {}``sparse random''
linear measurements, as described in Section~\ref{sec:noiseless}.

To summarize, the desirable properties of SHO-FA are that with high
probability%
\footnote{For most of the properties, we show that this probability is at least
$1-1/\spar^{\cO(1)}$, though we explicitly prove only $1-\cO(1/\spar)$.%
}:
\begin{itemize}
\item \textbf{Number of measurements:} For every $\spar$-sparse $\bx$,
with high probability over $\bA$, $\cO(\spar)$ linear measurements
suffice to reconstruct $\bx$. This is information-theoretically order-optimal.

\item \textbf{Number of measurement bits:} The total number of bits in $\bA\bx$
required to reconstruct $\bx$ to a relative error of $2^{-\preci}$
is $\cO(\spar(\log(\samp)+\preci))$. This is information-theoretically
order-optimal for any $\spar=\cO(\samp^{1-\Delta})$ (for any $\Delta>0$).

\item \textbf{Decoding time-complexity:} The total number of arithmetic
operations required is $\cO(\spar)$. This is information-theoretically
order-optimal.

\item \textbf{``Database-type queries'':} With constant probability $1-\epsilon$
any single ``database-type query'' can be answered in constant time. That is, the value of a single component of $\bx$ can be reconstructed in constant time with constant probability. %
\footnote{The constant $\epsilon$ can be made arbitrarily close to zero, at
the cost of a multiplicative factor $\cO(1/\epsilon)$ in the number
of measurements required. In fact, if we allow the number of measurements
to scale as $\cO(\spar\log(\spar))$, we can support any number of
database queries, each in constant time, with probability of every
one being answered correctly at with probability at least $1-\epsilon$.%
}

\item \textbf{Encoding/update complexity:} The computational complexity of generating $\bA\bx$ from $\bA$ and $\bx$ is $\cO(\samp)$, and if $\bx$ changes to some $\bx'$ in $\cO(1)$ locations, the computational complexity of updating $\bA\bx$ to $\bA\bx'$ is $\cO(1)$. Both of these are information-theoretically order-optimal.

\item \textbf{Noise:} Suppose $\bz$ and $\be$ have i.i.d. components%
\footnote{Even if the statistical distribution of the components of $\bz$ and
$\be$ are not i.i.d. Gaussian, statements with a similar flavor can
be made. For instance, pertaining to the effect of the distribution
of $\bz$, it turns out that our analysis is sensitive only on the
distribution of the \textit{sum} of components of $\bz$, rather then
the components themselves. Hence, for example, if the components of
$\bz$ are i.i.d. non-Gaussian, it turns out that via the Berry-Esseen
theorem~\cite{BerE} one can derive similar results to the ones derived
in this work. In another direction, if the components of $\bz$ are
not i.i.d. but do satisfy some {}``regularity constraints'',
then using Bernstein's inequality~\cite{Bernstein} one can again
derive analogous results. However, these arguments are more sensitive
and outside the scope of this paper, where the focus is on simpler
models.%
} drawn respectively from ${\mathcal{N}}(0,\sigma_{z}^{2})$ and ${\mathcal{N}}(0,\sigma_{e}^{2})$.
For every $\epsilon'>0$ and for $\spar=\cO(\samp^{1-\Delta})$ for any $\Delta>0$, a modified
version of SHO-FA (SHO-FA-NO) that with high probability reconstructs
$\bx$ with an estimation error of $(1+\epsilon')(||\bz||_{2}+||\be||_{2})$\footnote{As noted in Footnote~\ref{foot:recon}, this $\ell_2<\ell_2$ reconstruction guarantee implies the weaker $\ell_1<\ell_1$ reconstruction guarantee $||\bx-\hat{\bx}||_1<(1+\epsilon')(||\bz||_1+||\be||_1)$}.

\item \textbf{Practicality:} As validated by simulations (shown in Appendix~\ref{app:sim}),
most of the constant factors involved above are not large. 

\item \textbf{Different bases:} As is common in the compressive sensing
literature, our techniques generalize directly to the setting wherein
$\bx$ is sparse in an alternative basis (say, for example, in a wavelet
basis). 

\item \textbf{Universality:} While we present a specific ensemble of matrices
over which SHO-FA operates, we argue that in fact similar algorithms
work over fairly general ensembles of {}``sparse random matrices'' (see Section~\ref{sec:shofaint}),
and further that such matrices can occur in applications, for instance
in wireless MIMO systems~\cite{Guo:10} (Figure~\ref{fig:base_station} gives such an example) and Network Tomography~\cite{XuMT:11}.
\end{itemize}

\begin{figure}[th]
\centering{}\includegraphics[scale=0.75]{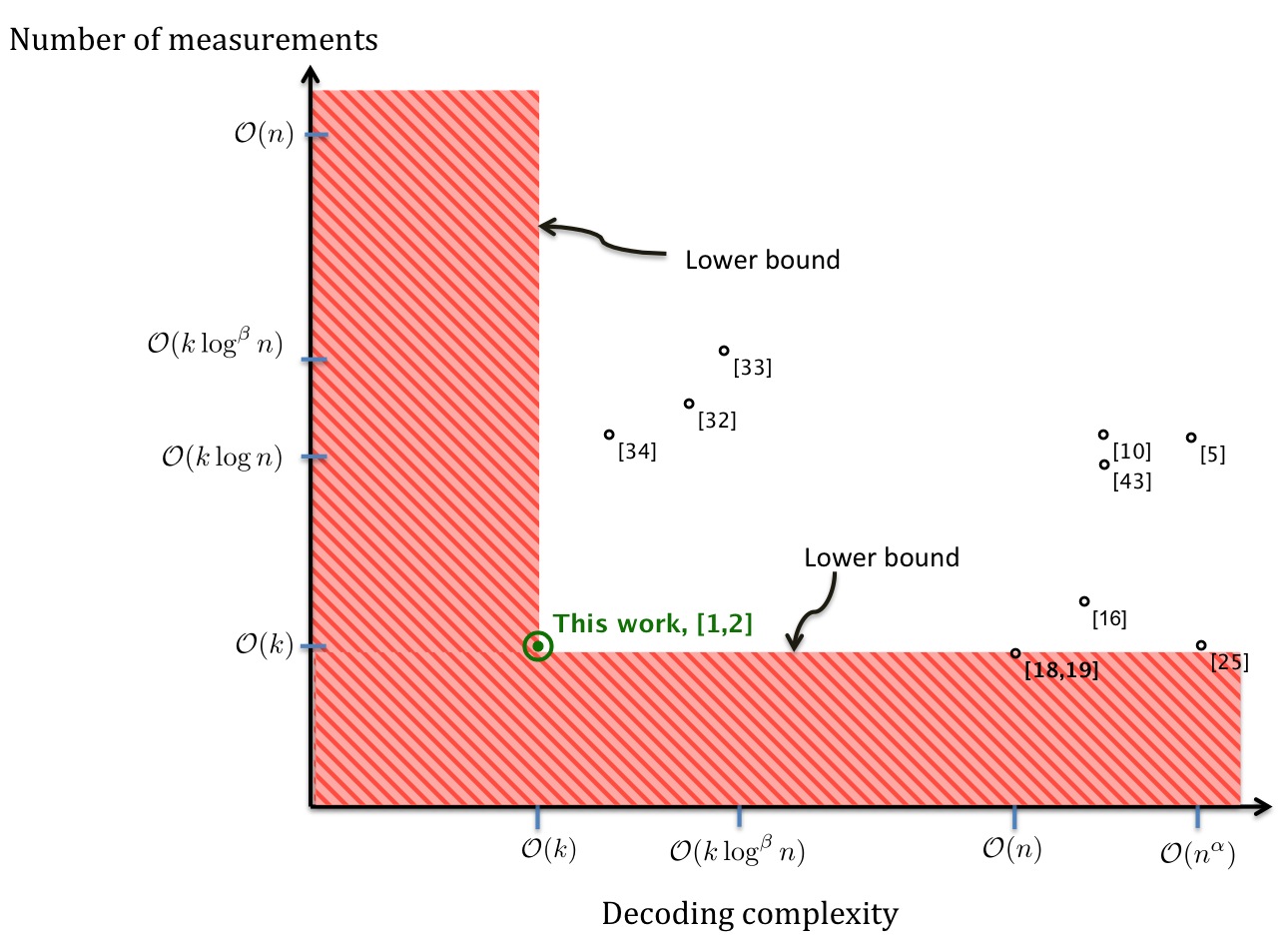} \label{fig:comparison}
\caption{A comparison of prior work with this work (and that of~\cite{PawR:12}) in two parameters -- decoding complexity, and number of measurements. 
Information-theoretic bounds show that the shaded region is infeasible. Only some of the many prior results in the field of Compressive Sensing are plotted here (some other works are referenced in the following table). Since $\spar$ and $\samp$ do not scale together in a fixed manner, the  position of points in this graph should be interpreted in a relative rather than absolute manner. The constants $\alpha$ and $\beta$ are intended to represent the degrees of ``small'' polynomials.
}
\end{figure}

\begin{table}[!htb]\label{tab:comparison}
 {\scriptsize { \hspace{-0.4in} }{\small { }}{\scriptsize }%
\begin{tabular}{|c|c|c|c|c|c|c|c|c|c|}
\hline
{\scriptsize Reference } & {\scriptsize $\bA$ } & {\scriptsize $\bx$ } & {\scriptsize $\bz$ } & {\scriptsize $\be$ } & {\scriptsize Reconstruction } & {\scriptsize $\mathbf{P_{e}}$ } & {\scriptsize \# Measurements } & {\scriptsize \# Decoding steps } & {\scriptsize Precision }\tabularnewline
 &  &  &  &  & {\scriptsize Goal } &  &  &  & \tabularnewline
\hline
\hline
{\scriptsize Reed-Solomon~\cite{ReeS:60} } & {\scriptsize D } & {\scriptsize D } & {\scriptsize $0$ } & {\scriptsize $0$ } & {\scriptsize Exact } & {\scriptsize $0$ } & {\scriptsize $2\spar+1$ } & {\scriptsize $\cO(\samp\log\samp))$~\cite{Ale:05} } & {\scriptsize -- }\tabularnewline
\hline
{\scriptsize Singleton~\cite{Sin:64} } & {\scriptsize D/R } & {\scriptsize D } & {\scriptsize $0$ } & {\scriptsize $0$ } & {\scriptsize Exact } & {\scriptsize $0$ } & {\scriptsize $\geq2\spar$ } & {\scriptsize -- } & {\scriptsize --}\tabularnewline
\hline
{\scriptsize Mitzenmacher-Varghese~\cite{MitG:12} } & {\scriptsize R } & {\scriptsize D } & {\scriptsize $0$ } & {\scriptsize $0$ } & {\scriptsize Exact } & {\scriptsize $\cO(1)$ } & {\scriptsize $\cO(\spar)$ } & {\scriptsize $\cO(\samp)$ } & {\scriptsize --}\tabularnewline
\hline
{\scriptsize Kudekar-Pfister~\cite{KudP:10} } & {\scriptsize R } & {\scriptsize D } & {\scriptsize $0$ } & {\scriptsize $0$ } & {\scriptsize Exact } & {\scriptsize $o(1)$ } & {\scriptsize $\cO(\spar)$ } & {\scriptsize $\cO(\samp)$ } & {\scriptsize -- }\tabularnewline
\hline
{\scriptsize Tropp-Gilbert~\cite{TroG:07} } & {\scriptsize G } & {\scriptsize D } & {\scriptsize $0$ } & {\scriptsize $0$ } & {\scriptsize Exact } & {\scriptsize $o(1)$}  & {\scriptsize $\cO(\spar\log(\samp))$ } & {\scriptsize $\cO(\spar^{2}\samp\log(\samp))$ } & {\scriptsize -- }\tabularnewline
\hline
\hline
{\scriptsize 
Wu-Verd\'u '10~\cite{WuV:10} } & {\scriptsize R } & {\scriptsize R } & {\scriptsize R } & {\scriptsize $0$ } & {\scriptsize Exact } & {\scriptsize $o(1)$ } & {\scriptsize $\samp\bar{d}(\bx+\bz)+o(\samp)$ } & {\scriptsize $\cO(\exp(\samp))$ } & {\scriptsize --}\tabularnewline
\hline
{\scriptsize Donoho }\textit{\scriptsize et al.}{\scriptsize ~\cite{DonJM:11} } & {\scriptsize R } & {\scriptsize R } & {\scriptsize R } & {\scriptsize $0$ } & {\scriptsize Exact } & {\scriptsize o(1) } & {\scriptsize $\samp\bar{d}(\bx+\bz)+o(\samp)$ } & {\scriptsize $\cO(\samp^{3})$ } & {\scriptsize -- }\tabularnewline
 &  &  &  & {\scriptsize R } & {\scriptsize ${\mathscr{l}}_{2}<{C}{\mathscr{l}}_{2}$ } &  &  &  & \tabularnewline
\hline
\hline
{\scriptsize 
Cormode-Muthukrishnan ~\cite{CorM:06} } & {\scriptsize R } & {\scriptsize D } & {\scriptsize $0$ } & {\scriptsize $0$ } & {\scriptsize ${\mathscr{l}}_{2}<{C}{\mathscr{l}}_{2}$ } & {\scriptsize $o(1)$ } & {\scriptsize $\cO(\spar\;\mbox{poly}(\log(\samp)))$ } & {\scriptsize $\cO(\spar\;\mbox{poly}(\log(\samp)))$ } & {\scriptsize --}\tabularnewline
\hline
{\scriptsize Cohen }\textit{\scriptsize et al.}{\scriptsize ~\cite{CohDD:09} } & {\scriptsize D } & {\scriptsize D } & {\scriptsize D } & {\scriptsize $0$ } & {\scriptsize ${\mathscr{l}}_{2}<{C}{\mathscr{l}}_{2}$ } & {\scriptsize $0$ } & {\scriptsize $\Omega(\samp)$ } & {\scriptsize -- } & {\scriptsize -- }\tabularnewline
\hline
{\scriptsize Price-Woodruff~\cite{PriW:11} } & {\scriptsize D } & {\scriptsize D } & {\scriptsize D } & {\scriptsize $0$ } & {\scriptsize ${\mathscr{l}}_{2}<{C}{\mathscr{l}}_{2}$ } & {\scriptsize $o(1)$ } & {\scriptsize $\theta(\spar\log(\samp/\spar))$ } & {\scriptsize -- } & {\scriptsize -- }\tabularnewline
\hline
{\scriptsize Ba }\textit{\scriptsize et al.}{\scriptsize ~\cite{BaIPW:10} } & {\scriptsize D/R } & {\scriptsize D } & {\scriptsize D } & {\scriptsize $0$ } & {\scriptsize ${\mathscr{l}}_{1}<{C}{\mathscr{l}}_{1}$ } & {\scriptsize $\cO(1)$ } & {\scriptsize $\Omega(\spar\log(\samp/\spar))$ } & {\scriptsize -- } & {\scriptsize $\cO(\log(\samp))$ }\tabularnewline
\hline
{\scriptsize Ba }\textit{\scriptsize et al.}{\scriptsize ~\cite{BaIPW:10} } & {\scriptsize R } & {\scriptsize D } & {\scriptsize R } & {\scriptsize $0$ } & {\scriptsize ${\mathscr{l}}_{2}<{C}{\mathscr{l}}_{2}$ } & {\scriptsize $o(1)$ } & {\scriptsize $\cO(\spar)$ } & {\scriptsize $\cO(\exp(\samp))$ } & \tabularnewline
\hline
{\scriptsize Candés~\cite{Can:08}, } & {\scriptsize R } & {\scriptsize D } & {\scriptsize D } & {\scriptsize D } & {\scriptsize ${\mathscr{l}}_{2}<\frac{C}{\sqrt{\spar}}{\mathscr{l}}_{1}$ } & {\scriptsize $o(1)$ } & {\scriptsize $\cO(\spar\log(\samp/\spar))$ } & {\scriptsize LP } & {\scriptsize --}\tabularnewline
{\scriptsize Baraniuk }\textit{\scriptsize et al.}{\scriptsize ~\cite{BarDDW:08} } &  &  &  &  &  &  &  &  & \tabularnewline
\hline
{\scriptsize 
Indyk }\textit{\scriptsize et al.}{\scriptsize ~\cite{IndR:08} } & {\scriptsize D } & {\scriptsize D } & {\scriptsize D } & {\scriptsize D } & {\scriptsize ${\mathscr{l}}_{1}<(1+\epsilon){\mathscr{l}}_{1}$ } & {\scriptsize $0$ } & {\scriptsize $\cO(\spar\log(\samp/\spar))$ } & {\scriptsize $\cO(\samp\log(\samp/\spar))$ } & {\scriptsize -- }\tabularnewline
\hline
\hline
{\scriptsize Ak\c{c}akaya }\textit{\scriptsize et al.}{\scriptsize ~\cite{AkcT:08} } & {\scriptsize R } & {\scriptsize D } & {\scriptsize $0$ } & {\scriptsize R } & {\scriptsize ${\mathscr{l}}_{2}<{C}{\mathscr{l}}_{2}/$ } & {\scriptsize $0$ } & {\scriptsize $\cO(\spar)$ } & {\scriptsize $\cO(\exp(\samp))$ } & {\scriptsize --}\tabularnewline
 &  &  &  &  & {\scriptsize Sup. Rec. } &  & {\scriptsize Cond. on $x_{min}$ } &  & \tabularnewline
\hline
{\scriptsize Wu-Verd\'u '11~\cite{WuV:11} } & {\scriptsize R } & {\scriptsize R } & {\scriptsize R } & {\scriptsize R } & {\scriptsize ${\mathscr{l}}_{2}<{C}{\mathscr{l}}_{2}$ } & {\scriptsize $\cO(1)$ } & {\scriptsize $\bar{d}(\bx+\bz)$ } & {\scriptsize $\cO(\exp(\samp))$ } & {\scriptsize -- }\tabularnewline
\hline
{\scriptsize Wainwright~\cite{Wai:09} } & {\scriptsize ${\mathcal{N}}$ } & {\scriptsize D } & {\scriptsize $0$ } & {\scriptsize R } & {\scriptsize Sup. Rec. } & {\scriptsize $\cO(1)$ } & {\scriptsize $\Omega(\spar\log(\samp/\spar))$ } & {\scriptsize -- } & {\scriptsize --}\tabularnewline
\hline
{\scriptsize Fletcher }\textit{\scriptsize et al.}{\scriptsize ~\cite{FleRG:09} } & {\scriptsize ${\mathcal{N}}$ } & {\scriptsize D } & {\scriptsize $0$ } & {\scriptsize R } & {\scriptsize Sup. Rec. } & {\scriptsize $o(1)$ } & {\scriptsize $\cO(\spar\log(\samp-\spar))$ } & {\scriptsize -- } & {\scriptsize --}\tabularnewline
\hline
{\scriptsize Aeron }\textit{\scriptsize et al.}{\scriptsize ~\cite{AerSM:10} } & {\scriptsize ${\mathcal{N}}$ } & {\scriptsize D } & {\scriptsize $0$ } & {\scriptsize R } & {\scriptsize Sup. Rec. } & {\scriptsize $\cO(1)$ } & {\scriptsize $\Omega(\spar\log(\samp/\spar))$ } &  & {\scriptsize --}\tabularnewline
\hline
\hline
{\scriptsize Plan-Vershynin~\cite{PlaV:11} } & {\scriptsize R } & {\scriptsize D } & {\scriptsize $0$ } & {\scriptsize sgn } & {\scriptsize ${\mathscr{l}}_{2}<f(\bx,\bx_{\spar})$ } & {\scriptsize $\cO(1)$ } & {\scriptsize $\spar^{2}\log(\samp/\spar)$ } & {\scriptsize LP } & {\scriptsize 1}\tabularnewline
\hline
{\scriptsize Jacques }\textit{\scriptsize et al.}{\scriptsize ~\cite{JacLBB:11} } & {\scriptsize R } & {\scriptsize D } & {\scriptsize $0$ } & {\scriptsize sgn } & {\scriptsize ${\mathscr{l}}_{2}<f'(\samp,\spar)$ } & {\scriptsize $\cO(1)$ } & {\scriptsize $\spar\log(\samp)$ } & {\scriptsize $\exp(\samp)$ } & {\scriptsize 1 }\tabularnewline
\hline
{\scriptsize 
Sarvotham }\textit{\scriptsize et al.}{\scriptsize ~\cite{SarBB:06} } & {\scriptsize R } & {\scriptsize D } & {\scriptsize $0$ } & {\scriptsize $0$ } & {\scriptsize Exact } & {\scriptsize $\cO(1)$ } & {\scriptsize $\spar\log(\samp)$ } & {\scriptsize $\spar\log(\spar)\log(\samp)$ } & {\scriptsize -- }\tabularnewline
\hline
{\scriptsize Gilbert }\textit{\scriptsize et al.}{\scriptsize ~\cite{GilSTV06} } & {\scriptsize R } & {\scriptsize P.L. } & {\scriptsize P.L. } & {\scriptsize 0 } & {\scriptsize ${\mathscr{l}}_{1}<{[1+C\log(\samp)]}{\mathscr{l}}_{1}$ } & {\scriptsize $\cO(1)$ } & {\scriptsize $\spar\log^{2}(\samp)$ } & {\scriptsize $\spar\log^{2}(\samp)\log^{2}(\spar)$ } & {\scriptsize --}\tabularnewline
\hline
\hline
{\scriptsize 
This work/Pawar {\it et al.}~\cite{PawR:12} } & {\scriptsize R } & {\scriptsize D } & {\scriptsize $0$ } & {\scriptsize $0$ } & {\scriptsize Exact } & {\scriptsize $o(1)$ } & {\scriptsize $\cO(\spar)$ } & {\scriptsize $\cO(\spar)$ } & {\scriptsize $\cO(\log(\samp)+\preci)$ }\tabularnewline
 & {\scriptsize R } & {\scriptsize D } & {\scriptsize R } & {\scriptsize R } & {\scriptsize ${\mathscr{l}}_{1}<{C}{\mathscr{l}}_{1}$ } & {\scriptsize $o(1)$ } & {\scriptsize $\cO(\spar)$ } & {\scriptsize $\cO(\spar)$ } & {\scriptsize $\cO(\log(\samp))$ }\tabularnewline
\hline
\end{tabular}{\small \vspace*{2mm}
 }{\scriptsize }\\
{\scriptsize{} }{\small {} }{\small \par}
{\scriptsize{} } \centering{}}%
\parbox[c]{0.95\linewidth}{%
{\small {} }\textbf{\small Explanations and discussion}{\small :{
At the risk of missing much of the literature, and also perhaps oversimplifying
nuanced results, we summarize in this table many of the strands of
work preceding this paper and related to it -- not all results from
each work are represented in this table. The second to the fifth columns
respectively reference whether the measurement matrix $\bA$, source
$\spar$-sparse vector $\bx$, source noise $\bz$, and measurement
noise $\be$ are random (R) or deterministic (D) -- a $0$ in a column
corresponding to noise indicates that that work did not consider that
type of noise. An entry {}``P.L.'' stands for {}``Power Law''
decay in columns corresponding to $\bx$ and $\bz$. For achievability
schemes, in general $D$-type results are stronger than $R$-type
results, which in turn are stronger than $0$-type results. This is
because a $D$-type result for the measurement matrix indicates that
there is an explicit construction of a matrix that satisfies the required
goals, whereas the $R$-type results generally indicate that the result
is true with high probability over measurement matrices. Analogously,
a $D$ in the columns corresponding to $\bx$, $\bz$ or $\be$ indicates
that the scheme is true for all vectors, whereas an $R$ indicates
that it is true for random vectors from some suitable ensemble. For
converse results, the the opposite is true $0$ results are stronger
than $R$-type results, which are stronger than $D$-type results.
An entry ${\mathcal{N}}$ indicates the normal distribution -- the
results of~\cite{Wai:09} and~\cite{FleRG:09} are converses for
matrices with i.i.d. Gaussian entries. An entry {}``sgn'' denotes
(in the case of works dealing with one-bit measurements) that the
errors are sign errors. The sixth column corresponds to what the desired
goal is. The strongest possible goal is to have exact reconstruction
of $\bx$ (up to quantization error due to finite-precision artihmetic),
but this is not always possible, especially in the presence of noise.
Other possible goals include {}``Sup. Rec. '' (short
for support recovery) of $\bx$, or that the reconstruction $\hat{\bx}$
of $\bx$ differs from $\bx$ as a {}``small'' function
of $\bz$. It is known that if a deterministic reconstruction algorithm
is desired to work for all $\bx$ and $\bz$, then $||\hat{\bx}-\bx||_{2}<\cO(||\bz||_{2})$
is not possible with less than $\Omega(\samp)$ measurements~\cite{CohDD:09},
and that $||\hat{\bx}-\bx||_{2}<\cO(||\bz||_{1}/\sqrt{\spar})$ implies
$||\hat{\bx}-\bx||_{1}<\cO(||\bz||_{1})$. The reconstruction guarantees
in~\cite{PlaV:11,JacLBB:11} unfortunately do not fall neatly in
these categories. The seventh column indicates what the probability
of error is -- }\textit{\small i.e.}{\small{} the probability over
any randomness in $\bA$, $\bx$, $\bz$ and $\be$ that the reconstruction
goal in the sixth column is not met. In the eighth column, some entries
are marked $\bar{d}(\bx+\bz)$ -- this denotes the (upper) Rényi dimension
of $\bx+\bz$ -- in the case of exactly $\spar$-sparse vectors this
equals $\spar$, but for non-zero $\bz$ it depends on the distribution
of $\bz$. The ninth column considers the computational complexity
of the algorithms -- the entry {}``LP'' denotes the computational
complexity of solving a linear program. The final column notes whether
the particular work referenced considers the precision of arithmetic
operations, and if so, to what level. }}{\scriptsize{} }%
}{\scriptsize{} }

\end{table}

\subsection{Special acknowledgements}
{\bf While writing this paper, we became aware of parallel and independent work by Pawar
and Ramchandran~\cite{PawR:12} that relies on ideas similar
to our work and achieves similar performance guarantees. Both the work of~\cite{PawR:12} and the preliminary version  of this work~\cite{BakJCC:12} were presented at the same venue.
}

{In particular, the bounds on the minimum number of measurements
required for {}``worst-case'' recovery and the corresponding discussion
on recovery of signals with {}``random tails'' in~\cite{BaIPW:10}
led us to consider this problem in the first place. Equally, the class
of compressive sensing codes in~\cite{JafWHC:09}, which in turn
build upon the constructions of expander codes in~\cite{Spi:95},
have been influential in leading us to this work. While the model
in~\cite{Pri:11} differs from the one in this work, the techniques
therein are of significant interest in our work. The analysis in~\cite{Pri:11}
of the number of disjoint components in certain classes of random
graphs, and also the analysis of how noise propagates in iterative
decoding is potentially useful sharpening our results. We elaborate
on these in Section~\ref{sec:noisy}.

The work that is conceptually the closest to SHO-FA is that of the
Invertible Bloom Lookup Tables (IBLTs) introduced by Goodrich-Mitzenmacher~\cite{GooM:11}
(though our results were derived independently, and hence much of
our analysis follows a different line of reasoning). The data structures
and iterative decoding procedure (called {}``peeling'' in~\cite{GooM:11})
used are structurally very similar to the ones used in this work.
However the {}``measurements'' in IBLTs are fundamentally non-linear
in nature -- specifically, each measurement includes within it a {}``counter''
variable -- it is not obvious how to implement this in a linear manner.
Therefore, though the underlying graphical structure of our algorithms
is similar, the details of our implementation require new non-trivial
ideas. Also, IBLTs as described are not robust to either signal tails
or measurement noise. Nonetheless, the ideas in~\cite{GooM:11} have
been influential in this work. In particular, the notion that an individual
component of $\bx$ could be recovered in constant time, a common
feature of Bloom filters, came to our notice due to this work. 

\section{Exactly $\spar$-sparse $\bx$ and noiseless measurements}

\label{sec:noiseless} We first consider the simpler case when the
source signal is exactly $\spar$-sparse and the measurements are
noiseless, \textit{i.e.}, $\by=\bA\bx$, and both $\bz$ and $\be$
are all-zero vectors. The intuition presented here carries over to
the scenario wherein both $\bz$ and $\be$ are non-zero, considered
separately in Section~\ref{sec:noisy}

For $\spar$-sparse input vectors { $\bx\in\bbr^{\samp}$}
let the set $\cSx$ denote its \textit{support}, \textit{i.e.}, its
set of nonzero values $\{\xj:x_{\xj}\neq0\}$. Recall that in our
notation, for some $\meas$, a {\em measurement matrix} {$\bA\in\bbr^{\meas\times\samp}$}
is chosen probabilistically. This matrix operates on $\bx$ to yield
the \textit{measurement vector} {$\by\in\bbr^{\meas}$}
as $\by=\bA\bx$. The decoder takes the vector $\by$ as input and
outputs the reconstruction $\hat{\bx}\in\bbR^{\samp}$ -- it is desired
that $\hat{\bx}$ equal $\bx$ (with upto $P$ bits of precision)
with high probability over the choice of measurement matrices .

In this section, we describe a probabilistic construction of the measurement
matrix $\bA$ and a reconstruction algorithm SHO-FA that achieves
the following guarantees.

\begin{theorem} \label{thm:main}
Let $\spar\leq\samp$. There exists a reconstruction
algorithm SHO-FA for {$\bA\in\bbr^{\meas\times\samp}$
}with the following properties:
\begin{enumerate}
\item For every {$\bx\in\bbR^{\samp}$}, with probability $1-\cO(1/\spar)$
over the choice of $\bA$, SHO-FA produces a reconstruction $\hat{\bx}$
such that $||\bx-\hat{\bx}||_{1}/||\bx||_{1}\leq2^{-P}$
\item The number of measurements $\meas=\cnst\spar$ for some $\cnst>0$
\item The {number} of steps required by SHO-FA is $\cO(\spar)$
\item The {number} of bitwise arithmetic operations required by SHO-FA
is $\cO(\spar(\log{\samp}+P))$.
\end{enumerate}
\end{theorem} We present a ``simple'' proof of the above theorem in Sections~\ref{subsec:intuition} to~\ref{subsec:correctness}. In Section~\ref{2-core} we direct the reader to an alternative, more technically challenging, analysis (based on the work of~\cite{GooM:11}) that leads to a tighter characterization of the constant factors in the parameters of Theorem~\ref{thm:main}.

\subsection{High-level intuition}\label{subsec:intuition}

If $\meas=\Theta(\samp)$, the task of reconstructing $\bx$ from
$\by=\bA\bx$ appears similar to that of \textit{syndrome decoding}
of a channel code of rate $\samp/\meas$~\cite{Roth}. It is well-known~\cite{HooLW:06}
that channel codes based on {\em bipartite expander graphs}, \textit{i.e.},
bipartite graphs with good expansion guarantees for all sets of size
less than or equal to $\spar$, allow for decoding in a number of
steps that is linear in the size of $\bx$. In particular, given such
a bipartite expander graph with $\samp$ nodes on the left and $\meas$
nodes on the right, choosing the matrix $\bA$ as a $\meas\times\samp$
binary matrix with non-zero values in the locations where the corresponding
pair of nodes in the graph has an edge is known to result in codes
with rate and relative minimum distance that is linear in $\samp$.

Motivated by this~\cite{JafWHC:09} explore a measurement design
that is derived from expander graphs and show that $\cO(\spar\log(\samp/\spar))$
measurements suffice, and $\cO(\spar)$ iterations with overall decoding
complexity of $\cO(\samp\log{(\samp/\spar)})$.%
\footnote{The work of~\cite{BerIR:08} is related -- it also relies on bipartite
expander graphs, and has similar performance for exactly $\spar$-sparse
vectors. But~\cite{BerIR:08} can also handle a significantly larger class of approximately
$\spar$-sparse vectors than~\cite{JafWHC:09}. However, our algorithms
are closer in spirit to those of~\cite{JafWHC:09}, and hence we
focus on this work.%
}

It is tempting to think that perhaps an optimized application of expander
graphs could result in a design that require only $\cO(\spar)$ number
of measurements. However, we show that in the compressive sensing
setting, where, typically $\spar=o(\samp)$, it is not possible to
satisfy the desired expansion properties. In particular, if one tries
to mimic the approach of~\cite{JafWHC:09}, one would need bipartite
expanders such that \textit{all} sets of size $\spar$ on one side
of the graph {}``expand'' -- we show in Lemma~\ref{lem:lowerbound}
that this is not possible. As such, this result may be of independent
interest for other work that require similar graphical constructions
(for instance the {}``magical graph{''} constructions of~\cite{CheKL:12},
or the expander code constructions of~\cite{Spi:95} in the high-rate
regime).

Instead, one of our key ideas is that we do not really need {}``true''
expansion. Instead, we rely on a notion of approximate expansion that
guarantees expansion for most $\spar$-sized sets (and their subsets)
of nodes on the left of our bipartite graph. We do so by showing that
any set of size at most $\spar$, with high probability over suitably
chosen measurement matrices, expands to the desired amount. Probabilistic
constructions turn out to exist for our desired property.%
\footnote{In fact similar properties have been considered before in the literature
-- for instance~\cite{CheKL:12} constructed so-called {}``magical
graphs'' with similar properties. Our contribution is the way we
use this property for our needs.%
} Such a construction is shown in Lemma~\ref{lem:expansion}.

Our second key idea is that in order to be able to recover all the
$\spar$ non-zero components of $\bx$ with at most $\cO(\spar)$
steps in the decoding algorithm, it is necessary (and sufficient)
that on average, the decoder reconstructs one previously undecoded
non-zero component of $\bx$, say $x_{\xj}$, in $\cO(1)$ steps in
the decoding algorithm. For $\spar=o(\samp)$ the algorithm does not
even have enough time to write out all of $\bx$, but only its non-zero
values. To achieve such efficient identification of $x_{\xj}$, we
go beyond the $0/1$ matrices used in almost all prior work on compressive
sensing based on expander graphs.%
\footnote{It can be argued that such a choice is a historical artifact, since
error-correcting codes based on expanders were originally designed
to work over the binary field $\bbF_{2}$. There is no reason to stick
to this convention when, as now, computations are done over $\bbr$.%
} Instead, we use distinct values in each row for the non-zero values
in $\bA$, so that if only one non-zero $x_{\xj}$ is involved in
the linear measurement involving a particular $y_{\yi}$ (a situation
that we demonstrate happens in a constant fraction of $y_{\yi}$),
one can identify which $x_{\xj}$ it must be in $\cO(1)$ time. Our
decoding then proceeds iteratively, by identifying such $x_{\xj}$
and canceling their effects on $y_{\yi}$, and terminates after $\cO(\spar)$
steps after all non-zero $x_{\xj}$ and their locations have been
identified (since we require our algorithm to work with high probability
for all $\bx$, we also add {}``verification'' measurements -- this
only increases the total number of measurements by a constant factor).
Our calculations are precise to $\cO(\log(\samp)+\preci)$ bits --
the first term in this comes from requirements necessary for computationally
efficient identification of non-zero $x_{\xj}$, and the last term
from the requirement that we require that the reconstructed vector
be correct up to $\preci$-precision. Hence the total number of bits
over all measurements is $\cO(\spar((\log(\samp)+\preci))$. Note
that this is information-theoretically order-optimal, since even specifying
$\spar$ locations in a length-$\samp$ vector requires $\Omega(\spar(\log(\samp/\spar))$
bits, and specifying the value of the non-zero locations so that the
relative reconstruction error is $\cO(2^{-\preci})$ requires $\Omega(\spar\preci)$
bits.

We now present our SHO-FA algorithm in two stages. We first by use
our first key idea (of {}``approximate{''} expansion) in Section~\ref{subsec:graph_prop}
to describe some properties of bipartite expander graphs with certain
parameters. We then show in Section~\ref{subsec:code_prop} how these
properties, via our second key idea (of efficient identification)
can be used by SHO-FA to obtain desirable performance.

\subsection{{}``Approximate Expander{''} Graph $\graph$}

\label{subsec:graph_prop} We first construct a bipartite graph $\graph$
(see Example $1$ in the following) with some desirable
properties outlined below. We then show in Lemmas~\ref{lem:expansion}
and~\ref{lem:manyleafs} that such graphs exist (Lemma~\ref{lem:lowerbound}
shows the non-existence of graphs with even stronger properties).
In Section~\ref{subsec:code_prop} we then use these graph properties
in the SHO-FA algorithm. To simplify notation in what follows (unless
otherwise specified) we omit rounding numbers resulting from taking
ratios or logarithms, with the understanding that the corresponding
inaccuracy introduced is negligible compared to the result of the
computation. Also, for ease of exposition, we fix various internal
parameters to {}``reasonable{''} values rather than optimizing
them to obtain {}``slightly{''} better performance at the cost
of obfuscating the explanations -- whenever this happens we shall
point it out parenthetically. Lastly, let $\epsilon$ be any {}``small{''}
positive number, corresponding to the probability of a certain {}``bad
event{''}.

\noindent \underline{\bf Properties of $\graph$:}
\begin{enumerate}
\item \label{prop:1} \underline{\it Construction of a left-regular bipartite graph:}
The graph $\graph$ is chosen uniformly at random from the set of
bipartite graphs with $\samp$ nodes on the left and $\meas'$ nodes
on the right, such that each node on the left has degree {$\dgr\geq 7$}.\footnote{For ease of analysis we now consider the case when $\dgr \geq 7$ -- our tighter result in {Theorem~\ref{thm:tighter} relaxes this}, and work for any $\dgr \geq 3$}
In particular, $\meas'$ is chosen to equal $\cnst\spar$ for some
design parameter $\cnst$ to be specified later as part of code design.
\item \underline{\it Edge weights for ``identifiability{''}:} \label{prop:2}
For each node on the right, the weights of the edges attached to it
are required to be distinct. In particular, each edge weight is chosen
as a complex number of unit magnitude, and phase between $0$ and
$\pi/2$. Since there are a total of $\dgr\samp$ edges in $\graph$,
choosing distinct phases for each edge attached to a node on the right
requires at most $\log(\dgr\samp)$ bits of precision (though on average
there {are} about $\dgr\samp/\meas'$ edges attached to a node on
the right, and hence on average one needs about $\log(\dgr\samp/\meas')$
bits of precision).

\item \underline{\it $\cSx$-expansion:} \label{prop:3} With high probability
over $\graph$ defined in Property~\ref{prop:1} above, for any set
$\cSx$ of $\spar$ nodes on the left, the number of nodes neighbouring
those in any $\cSpx\subseteq\cSx$ is required to be at least $2\dgr/3$
times the size of $\cSpx$.%
\footnote{The \textit{expansion factor} $2\dgr/3$ is somewhat arbitrary. In
our proofs, this can be replaced with any number strictly between
half the degree and the degree of the left nodes, and indeed one can
carefully optimize over such choices so as to improve the constant
in front of the expected time-complexity/number of measurements of
SHO-FA. Again, we omit this optimization since this can only improve
the performance of SHO-FA by a constant factor.%
} The proof of this statement is the subject of Lemma~\ref{lem:expansion}.
\item \underline{\it ``Many{''} $\cSx$-leaf nodes:} \label{prop:4} For
any set $\cSx$ of at most $\spar$ nodes on the left of $\graph$,
we call any node on the right of $\graph$ an \textit{$\cSx$-leaf
node} if it has exactly one neighbor in $\cSx$, and we call it a
\textit{$\cSx$-non-leaf node} if it has two or more neighbours in
$\cSx$. (If the node on the right has no neighbours in $\cSx$, we
call it a \textit{$\cSx$-zero node}.) Assuming $\cSx$ satisfies
the expansion condition in Property~\ref{prop:3} above, it can be
shown that at least a fraction $1/2$ of the nodes that are neighbours
of any $\cSpx\subseteq\cSx$ are $\cSpx$-leaf nodes.%
\footnote{Yet again, this choice of $1/2$ is a function of the choices made
for the degree of the left nodes in Property~\ref{prop:1} and the
expansion factor $2$ in Property~\ref{prop:3}. Again, we omit optimizing
it.%
} The proof of this statement is the subject of Lemma~\ref{lem:expansion}.
\end{enumerate}
\noindent \textit{Example $1$:} We now demonstrate via the following
toy example in Figures~\ref{fig:bip_graph} and~\ref{fig:bip_graph2}
a graph $\graph$ satisfying Properties~\ref{prop:1}-\ref{prop:4}.
\begin{figure}[th]
\centering{}\includegraphics[scale=0.55]{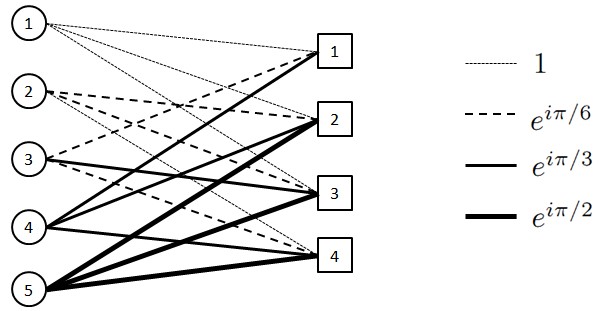} \label{fig:bip_graph}
\caption{\underline{\it Property~\ref{prop:1}}: Bipartite {}``approximate
expander{''} graph with $\samp=5$ nodes on the left, and $\meas'=4$
nodes on the right. Each node on the left has degree $\degr$. \underline{\it Property~\ref{prop:2}}:
The thicknesses of the edges represent the weights assigned to the
edges. In particular, it is required that for each node on the right,
the edges incoming have distinct weights. In this example, the thinnest
edges are assigned a weight of $1$, the next thickest edges have
a weight $e^{\iota\pi/6}$, the next thickest edges have weight $e^{\iota2\pi/6}=e^{\iota\pi/3}$,
and the thickest edges have weight $e^{\iota3\pi/6}=e^{\iota\pi/2}$. }
\end{figure}

\begin{figure}[th]
\begin{centering}
\includegraphics[scale=0.65]{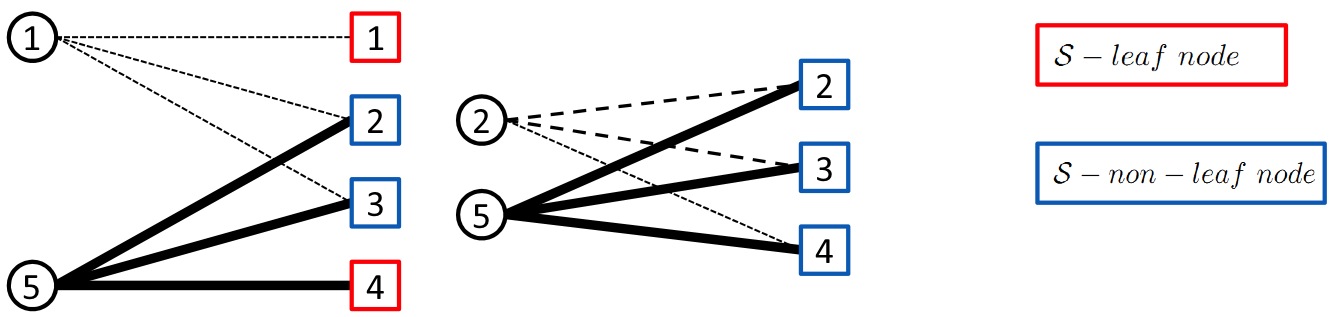} \caption{ \underline{\it Property~\ref{prop:3}}: We require that \textit{most}
sets $\cSpx$ of at most $|\cSx|=\spar=2$ nodes on the left in the
graph $\graph$ in Figure~\ref{fig:bip_graph} have at least $2|\cSpx|$
neighbors on the right. ) In the graph in Figure~\ref{fig:bip_graph}
it can be manually verified that most sets of size $\cSpx$ at most
$2$ have at least $2|\cSpx|$ neighbors. For example, Figure~\ref{fig:bip_graph}(a)
focuses on the subset $\cSpx=\{1,5\}$ of nodes on the left side of
$\graph$ in Figure~\ref{fig:bip_graph}. This particular $\cSpx$
has $4$ neighbours, and all its single-node subsets have $3$ neighbours.
The only $\cSpx$ set of two or fewer nodes that does not satisfy
Property~\ref{prop:3} is $\{2,5\}$, as shown in Figure~\ref{fig:bip_graph}(b),
since it has only $3<2\times2$ neighbours. \underline{\it Property~\ref{prop:4}}:
For sets $\cSpx$ that satisfy Property~\ref{prop:3} it can be manually
verified that {}``many{''} of their neighbours are $\cSpx$-leaf
nodes. For example, for $\cSpx=\{1,5\}$, two out of its four neighbours
(\textit{i.e.}, a fraction $1/2$) are $\cSpx$-leaf nodes -- which
satisfies the constraint that at least a fraction $1/2$ of its neighbours
be $\cSpx$-leaf nodes. On the other hand, for $\cSpx=\{2,5\}$ (which
does \textit{not} satisfy Property~\ref{prop:3}), \textit{none}
of its neighbours are $\cSpx$-leaf nodes.}

\par\end{centering}

\centering{}\label{fig:bip_graph2}
\end{figure}

\hfill{}$\Box$

We now state the Lemmas needed to make our arguments precise. First,
we formalize the $\cSpx$-expansion property defined in Property~\ref{prop:3}.
\begin{lemma}\label{lem:expansion}\textbf{(Property~\ref{prop:3}
($\cSx$-expansion)):} Let
$\spar<\samp\in\bbn$ be arbitrary, and let $\cnst\in\bbn$ be
fixed. Let $\graph$ be chosen uniformly at random from the set of
all bipartite graphs with $\samp$ nodes (each of degree $\dgr$)
on the left and $\meas'=\cnst\spar$ nodes on the right. Then for
any $\cSx$ of size at most $\spar$ and any $\cSpx\subseteq\cSx$,
with probability $1-o(1/\spar)$ (over the random choice $\graph$)
there are at least $2\dgr/3$ times as many nodes neighbouring those
in $\cSpx$, as there are in $\cSpx$. \end{lemma} \textit{Proof:}
Follows from a standard probabilistic method argument. Given for completeness
in Appendix~\ref{apx:proof_lem_expansion}.

\noindent Note here that, in contrast to the {}``usual{''} definition
of {}``vertex expansion{''}~\cite{HooLW:06} (wherein the expansion
property is desired \textit{{}``for all{''}} subsets of left nodes
up to a certain size) Lemma~\ref{lem:expansion} above only gives
a probabilistic expansion guarantee for any subset of $\cSx$ of size
$\spar$. In fact, Lemma~\ref{lem:lowerbound} below shows that for
the parameters of interest, {}``for all{''}-type expanders cannot
exist.

\begin{lemma}\label{lem:lowerbound} Let $\spar=o(\samp)$, and $d>0$
be an arbitrary constant. Let $\graph$ be an arbitrary bipartite
graph with $\samp$ nodes (each of degree $d$) on the left and $\meas'$
nodes on the right. Then for all sufficiently large $\samp$, suppose
each set of of size $\spar$ of $\cSx$ nodes on the left of $\graph$
has strictly more than $d/2$ times as many nodes neighbouring those
in $\cSx$, as there are in $\cSx$. Then $\meas'=\Omega(\spar\log(\samp/\spar))$.
\end{lemma} \textit{Proof:} Follows from the Hamming bound in coding
theory~\cite{Roth} and standard techniques for expander codes~\cite{Spi:95}.
Proof in Appendix~\ref{apx:proof_lowerbound}.

\noindent Another way of thinking about Lemma~\ref{lem:lowerbound}
is that it indicates that if one wants a {}``for all{''} guarantee
on expansion, then one has to return to the regime of $\meas'=\cO(\spar\log(\samp/\spar))$
measurements, as in {}``usual{''} compressive sensing.

Next, we formalize the {}``many $\cSx$-leaf nodes{''} property
defined in Property~\ref{prop:4}. Recall that for any set $\cSx$
of at most $\spar$ nodes on the left of $\graph$, we call any node
on the right of $\graph$ an \textit{$\cSx$-leaf node} if it has
exactly one neighbor in $\cSx$. \begin{lemma}\label{lem:manyleafs}
Let $\cSx$ be a set of $\spar$ nodes on the left of $\graph$ such
that the number of nodes neighbouring those in any $\cSpx\subseteq\cSx$
is at least $2\dgr/3$ times the size of $\cSpx$. Then at least a
fraction $1/2$ of the nodes that are neighbours of any $\cSpx\subseteq\cSx$
are $\cSpx$-leaf nodes. \end{lemma} \textit{Proof:} Based on Lemma~\ref{lem:expansion}.
Follows from a counting argument similar to those used in expander
codes~\cite{Spi:95}. Proof in Appendix~\ref{apx:proof_lem_manyleafs}.

\noindent 

{Given a graph $\graph$ satisfying properties \ref{prop:1}-\ref{prop:4},
we now describe our encoding and decoding procedure.}
\subsection{{Measurement design}}
\label{subsec:code_prop}

 \textit{Matrix structure and entries:} The encoder's \textit{measurement
matrix} $\bA$ is chosen based on the structure of $\graph$ (recall
that $\graph$ has $\samp$ nodes on the left and $\meas'$ nodes
on the right). To begin with, the matrix $\bA$ has $\meas=2\meas'$
rows, and its non-zero values are unit-norm complex numbers. 
\begin{remark} \label{rem:universality}This
choice of using complex numbers rather than real numbers in $\bA$
is for notational convenience only. One equally well choose a
matrix $\bA'$ with $\meas=4\meas'$ rows, and replace each row of
$\bA$ with two consecutive rows in $\bA'$ comprising respectively
of the real and imaginary parts of rows of $\bA$. Since the components
of $\bx$ are real numbers, hence there is a bijection between $\bA\bx$
and $\bA'\bx$ -- indeed, consecutive pairs of elements in $\bA'\bx$
are respectively the real and imaginary parts of the complex components
of $\bA\bx$. Also, as we shall see (in Section~\ref{subsec:bits}), 
the choice of unit-norm complex
numbers ensures that {}``noise{''} due to finite precision arithmetic
does not get {}``amplified{''}. In Section~\ref{subsec:univ}, we argue that this property enables us to apply SHO-FA to other settings such as wireless systems that naturally generate an ensemble of matrices that resemble SHO-FA. \end{remark}

\noindent In particular, corresponding to node $\yi$ on the right-hand
side of $\graph$, the matrix $\bA$ has two rows. The $\xj^{th}$
entries of the {$(2\yi-1)^{th}$} and {$2\yi^{th}$} rows of $\bA$
are respectively denoted $a_{\yi,\xj}^{(I)}$ and $a_{\yi,\xj}^{(V)}$
respectively. (The superscripts $(I)$ and $(V)$ respectively stand
for \textit{Identification} and \textit{Verification}, for reasons
that shall become clearer when we discuss the process to reconstruct
$\bx$.)

\noindent \textit{Identification entries:} If $\graph$ has no edge
connecting node $\xj$ on the left with $\yi$ on the right, then
the \textit{identification entry} $a_{\yi,\xj}^{(I)}$ is set to equal
$0$. Else, if there is indeed such an edge, $a_{\yi,\xj}^{(I)}$
is set to equal
\begin{equation}
a_{\yi,\xj}^{(I)}=e^{\iota\xj\pi/(2\samp)}.\label{eq:a_iden}
\end{equation}

(Here $\iota$ denotes the positive square root of $-1$.) This entry
$a_{\yi,\xj}^{(I)}$ can also be thought of as the weight of the edge
in $\graph$ connecting $\xj$ on the left with $\yi$ on the right.
In particular, the \textit{phase} $\xj\pi/(2\samp)$ of $a_{\yi,\xj}^{(I)}=e^{\iota\xj\pi/(2\samp)}$
will be critical for our algorithm. As in Property~\ref{prop:2}
in Section~\ref{subsec:graph_prop}, our choice above guarantees
distinct weights for all edges connected to a node $\yi$ on the right.\\

\noindent \textit{Verification entries:} Whenever the identification
entry $a_{\yi,\xj}^{(I)}$ equals $0$, we choose to set the corresponding
\textit{verification entry} $a_{\yi,\xj}^{(V)}$ also to be zero.
On the other hand, whenever $a_{\yi,\xj}^{(I)}\neq0$, then we set
$a_{\yi,\xj}^{(V)}$ to equal $e^{\iota\theta_{\yi,\xj}^{(V)}}$ for
$\theta_{\yi,\xj}^{(V)}$ chosen uniformly at random from $[0,\pi/2]$
(with $\cO(\log(\spar))$ bits of precision).%
\footnote{This choice of precision for the verification entries contributes
one term to our expression for the precision of arithmetic required.
As we argue later in Section~\ref{subsec:bits}, this choice of precision
guarantees that if a single identification step returns a value for
$x_{\xj}$, this is indeed correct with probability $1-o(1/\spar)$.
Taking a union bound over $\cO(\spar)$ indices corresponding to non-zero
$x_{\xj}$ gives us an overall $1-o(1)$ probability of success.%
}

\noindent {\it Example $2$:} The matrix $\bA$ corresponding to the graph $\graph$ in Example $1$ is show in Figure~\ref{fig:matrix}.

\begin{figure}[h]\begin{center}
\includegraphics[scale=.45]{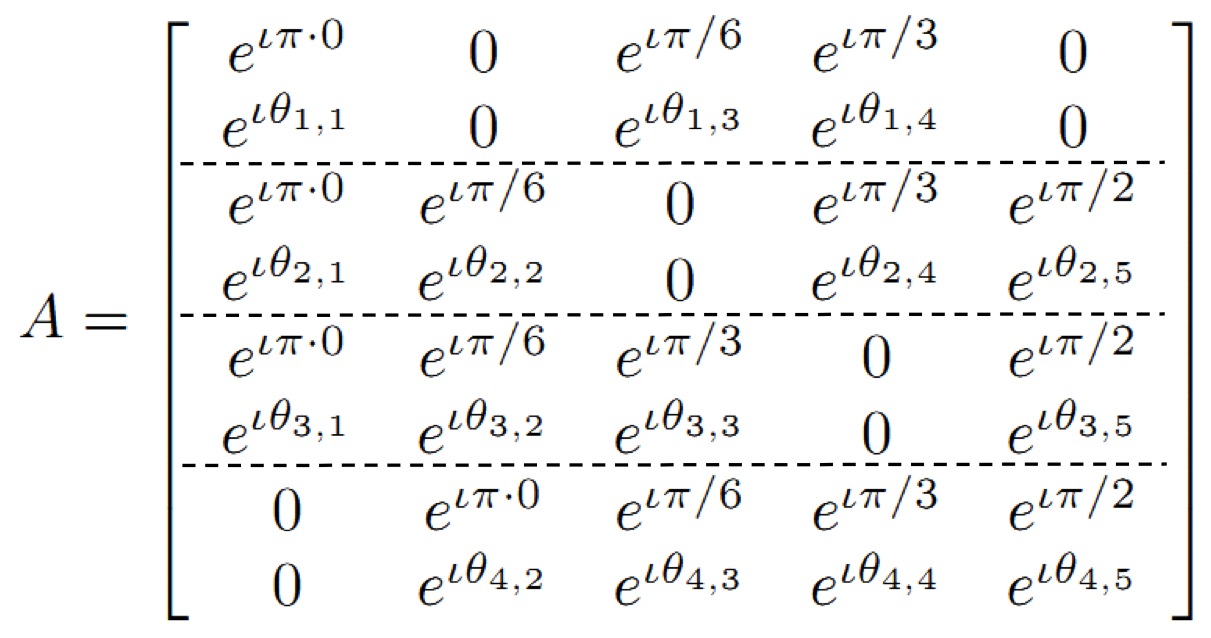}
\caption{
This $8\times 5$ matrix denotes the $\bA$ corresponding to the graph $\graph$. Note that its primary purpose is expository -- clearly, $8$ measurements (or indeed, $16$ measurements over $\bbr$) to reconstruct a $2$-sparse vector of length $5$ is too many! Nonetheless, this is just an artifact of the fact that $\samp$ in this example is small. In fact, according to our proofs, even as $\samp$ scales to infinity, the number of measurements required to reconstruct a $2$-sparse vector (or in general a $\spar$-sparse vector for constant $\spar$) remains constant! Also, note that we do not use the assignment for the identification {entries} $a_{\yi,\xj}^{(I)}$ specified in (\ref{eq:a_iden}), since doing so would result in ugly and not very illuminating calculations in Example $3$ below. However, as noted in Remark~\ref{rem:ident}, this is not critical -- it is sufficient that distinct entries in the identification rows of the matrix be distinct.
}
\label{fig:matrix}\end{center}
\end{figure}

\subsection{Reconstruction}
\subsubsection{{Overview}}
\label{sec:noiselessreconstruction}

We now provide some high-level intuition on the decoding process.

Since the measurement matrix $\bA$ has interspersed identification
and verification rows, this induces corresponding interspersed \textit{identification
observations} $y_{\yi}^{(I)}$ and verification \textit{verifications
observations} $y_{\yi}^{(V)}$ in the \textit{observation vector} $\by=\bA\bx$.
Let $\by^{(I)}=\{y_{\yi}^{(I)}\}$ denote the length-$\meas$ \textit{identification
vector} over $\bbc$, and $\by^{(V)}=\{y_{\yi}^{(V)}\}$ denote the
length-$\meas$ \textit{verification vector} over $\bbc$.

Given the measurement matrix $\bA$ and the observed $(\by^{(I)},\by^{(V)})$
identification and verification vectors, the decoder's task is to
find \textit{any ``consistent''} $\spar$-sparse vector $\hat{\bx}$ such that $\bA\hat{\bx}$
results in the corresponding identification and verification vectors.
We shall argue below that if we succeed, then with high probability
over $\bA$ (specifically, over the verification entries of $\bA$),
this $\hat{\bx}$ must equal $\bx$.

{To find such a consistent $\hat{\bx},$ we design an iterative
decoding scheme. This scheme starts by setting the initial guess for
the reconstruction vector $\hat{\bx}$ to the all-zero vector. It
then initializes, in the manner described in the next paragraph, a $\cSx$-leaf-node list},
$\mathcal{L}(1)$, a set of indices of $\cSx$-leaf nodes.

The
decoder checks to see whether $\yi$ is a $\cSx$-leaf node in the following
way. First, it looks at the entry $y_{\yi}^{(I)}$ and {}``estimates{''}
which node $\xj$ on the left of the graph $\graph$ ``could have generated the identification
observation $y_{\yi}^{(I)}$''. It then uses the verification
entry $a_{\yi,\xj}^{(V)}$ and the verification observation $y_{\yi}^{(V)}$
to verify its estimate. After sequentially examining each entry $y_{\yi}^{(I)}$, the list of {\it all} $\cSx$-leaf nodes is denoted $\mathcal{L}(1)$.

{In the $\iter^{th}$ iteration of the decoding process,
the decoder picks a leaf node in $\yi\in\mathcal{L}(t)$. Using this, it then reconstructs the non-zero component $x_{\xj}$ of $\bx$ that ``generated'' $y_{\yi}^{(I)}$.
If this reconstructed value $x_{\xj}$ is successfully ``verified'' using the verification
entry $a_{\yi,\xj}^{(V)}$ and the verification observation $y_{\yi}^{(V)}$)%
\footnote{As Ronald W. Reagan liked to remind us, {}``\textit{doveryai,
no proveryai''}.}%
}, {then the algorithm performs the following steps in this iteration:
\begin{itemize}
\item It updates the observation vectors
by subtracting the {}``contribution{''} of the coordinate $x_{\xj}$
to the measurements it influences (there are exactly $7$ of them
since the degree of the nodes on the left side of $\graph$ is $7$).
\item
It updates the $\cSx$-leaf-node list, $\mathcal{L}(t)$ by removing $\yi$ from $\mathcal{L}(t)$
and checking the change of status (zero, leaf, or non-leaf) of other indices influenced by $x_{\xj}$ (there at
most $6$), 
\item Finally the algorithm picks a new index $\yi$ from the updated list, $\mathcal{L}(t+1)$
for the next iteration. 
\end{itemize}
The decoder performs the above operations
repeatedly until $\hat{\bx}$ has been completely recovered. We also
show that (with high probability over $\bA$) in at most $\spar$
steps this process does indeed terminate.}

\noindent \textit{Example $3$:} Figures~\ref{fig:bip_fig1}--\ref{fig:bip_fig5}
show a sample decoding process for the matrix $\bA$ as in Example
$2$, and the observed vector $\by$ shown in the figures. The example
also demonstrates each of several possible scenarios the algorithm
can find itself in, and how it deals with them.

\begin{figure*}[p]
\includegraphics[scale=0.55]{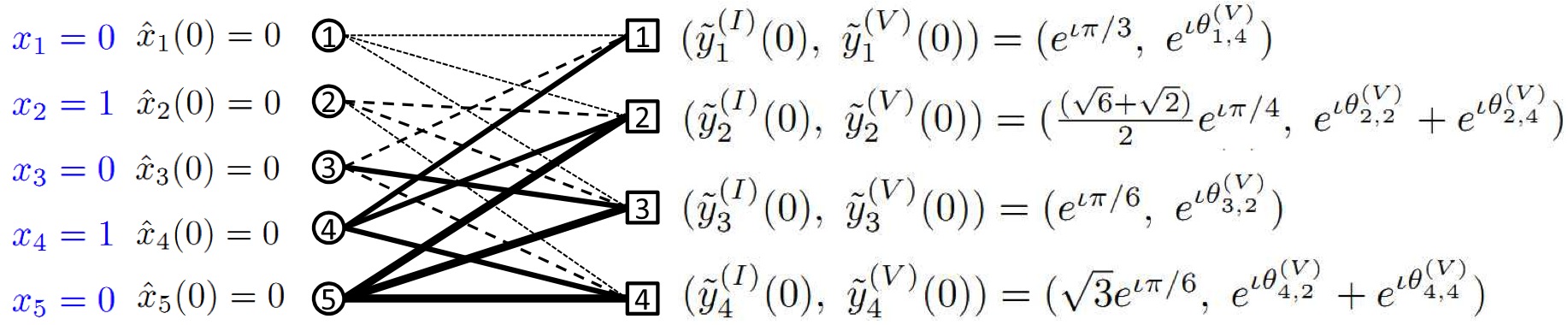} \caption{ \underline{\it Initialization}: The (true) $\bx$
equals $(0,1,0,1,0)$ (and hence $\cSx=\{2,4\}$). Also note that
nodes $1$ and $3$ on the right of $\graph$ are $\cSx$-leaf nodes,
as defined in Property{~\ref{prop:4}}. However, all of this is
unknown to the decoder \textit{a priori}.
The decoder sets the (starting) estimate $\hat{\bx}(0)$ of the reconstruction
vector $\hat{\bx}$ to the all-zeros vector. The (starting) \textit{gap
vector} $\tilde{\by}$ is set to equal $\by$, which
in turn equals the corresponding $4$ pairs of identification and
verification observations on the right-hand side of $\graph$. The
specific values of $\theta_{i,j}^{(V)}$ in the verification observations
do not matter currently -- all that matters is that given $\bx$,
each of the four verification observations are non-zero (with high
probability over the choices of $\theta_{i,j}^{(V)}$). Hence the
(starting) value of the \textit{neighbourly}
set equals $\{1,2,3,4\}$. This step takes $\cO(\spar)$ number of
steps, just to initialize the neighbourly set. By the end of the decoding
algorithm (if it runs successfully), the tables will be turned --
all the entries on the right of $\graph$ will equal zero, and (at
most) $\spar$ entries on the left of $\graph$ will be non-zero. }

\label{fig:bip_fig1}
\end{figure*}

\begin{figure*}[p]
 \includegraphics[scale=0.55]{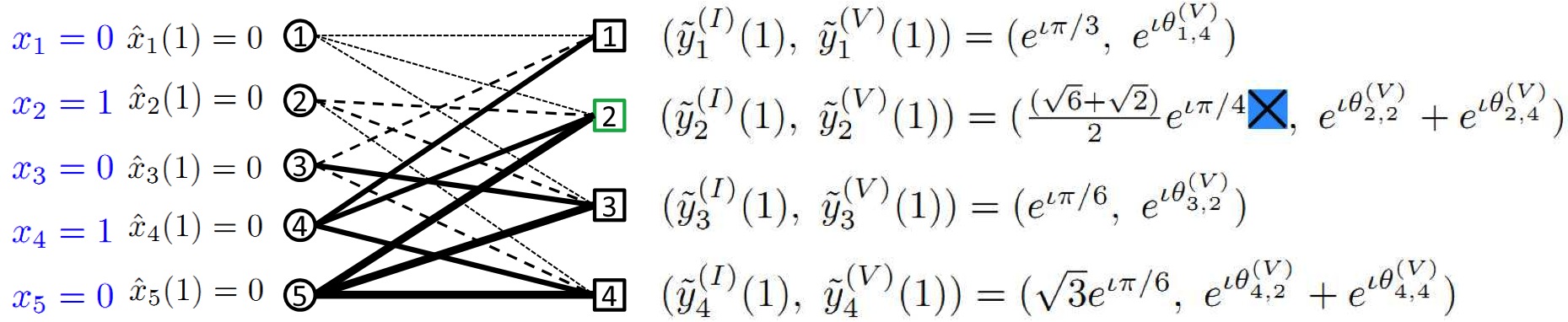} \caption{\underline{\it Leaf-Node List 1 (Failed identification)}:
The decoder picks the index $\yi=2$ from the neighbourly set $\{1,2,3,4\}$,
and checks the phase of the corresponding gap vector identification
observation $\tilde{y}_{2}^{(I)}$. Since this equals $\pi/4$, which
is \textit{not} in the set of possible
phases in the $2^{nd}$ identification row of $\bA$ (which are all
multiples of $\pi/6$), the decoder declares $\yi=2$ is not a leaf
node. This entire process takes a constant number of steps. }

\label{fig:bip_fig2}
\end{figure*}

\begin{figure*}[p]
 \includegraphics[scale=0.55]{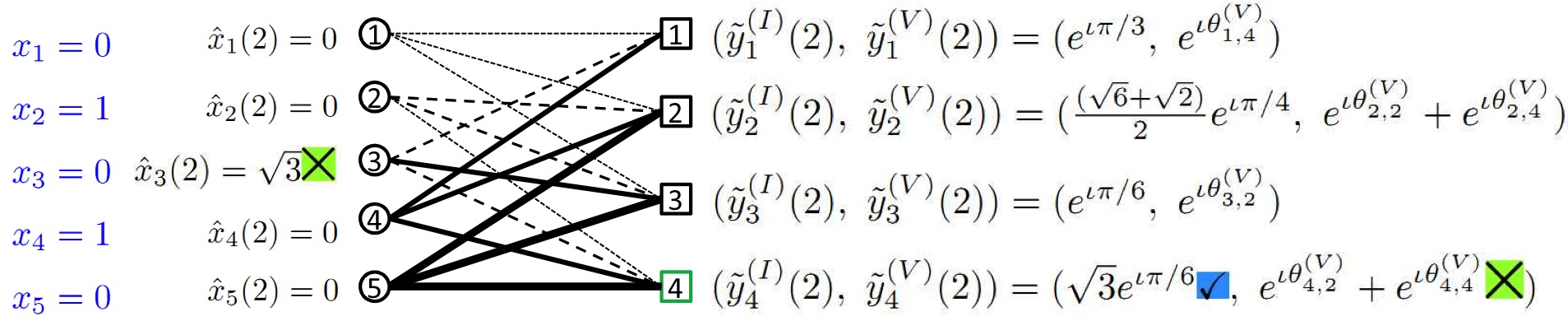} \caption{\underline{\it Leaf-Node List 2 (Passed identification, failed verification)}:
In this step, a potentially more serious failure could happen. In
particular, suppose the decoder picks the index $\yi=4$ from the
neighbourly set $\{1,2,3,4\}$ (note that $4$ is \textit{also}
not a $\cSx$-leaf node), and checks the phase of the corresponding
gap vector identification observation $\tilde{y}_{4}^{(I)}$, it just
so happens that the value of $\bx$ is such that this corresponds
to a phase of $\pi/6$. But as can be seen from the matrix in Figure~\ref{fig:matrix},
for $\yi=4$ this corresponds to $a_{\yi,\xj}^{(I)}$ for $\xj=3$.
Hence the decoder would make a {}``false identification''
of $\xj=3$, and estimate that $\hat{x}_{3}$ equals the magnitude
of $\tilde{y}_{4}^{(I)}$, which would equal $\sqrt{3}$. This is
where the verification entries and verification observations save
the day. Recall that the phase of each verification entry is chosen
uniformly at random (with sufficient bit precision) from $[0,\pi/2)$,
independently of both $\bx$ and the other entries of $\bA$. Hence
the probability that $\sqrt{3}$ (the misdirected value of $\hat{x}_{3}$)
times the corresponding verification entry $a_{4,3}^{(V)}$ equals
$\tilde{y}_{4}^{(I)}$ is {}``small''. Hence the decoder in this
case too declares $\yi=4$ is not a leaf node. This entire process
takes a constant number of steps. }

\label{fig:bip_fig3}
\end{figure*}

\begin{figure*}[p]
 \includegraphics[scale=0.55]{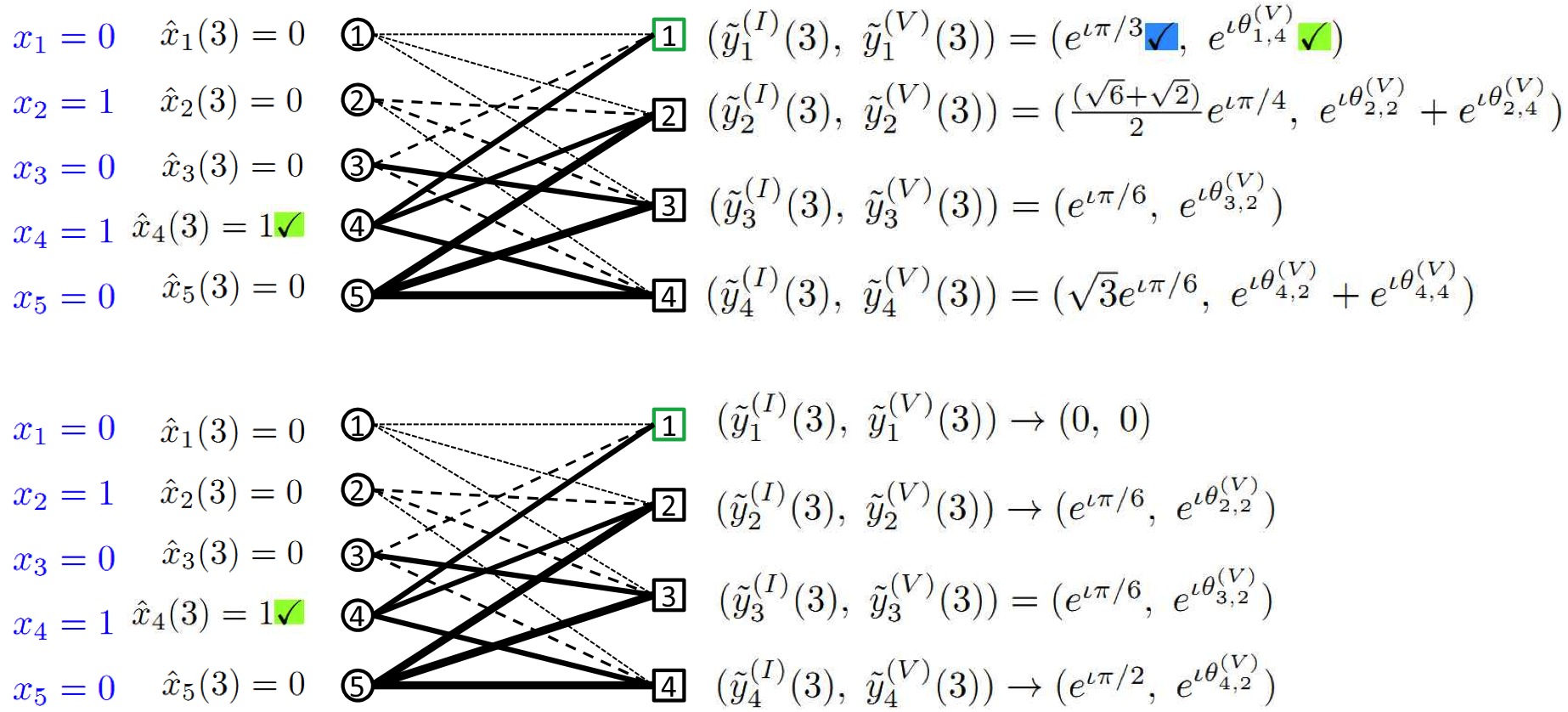} \caption{\underline{\it Leaf-Node List 3 (Passed identification, passed verification) and the first iteration}:
Now, suppose the decoder randomly picks the index $\yi=1$ from the
neighbourly set $\{1,2,3,4\}$ (note that $1$ \textit{is}
a $\cSx$-leaf node). In this case, the phase of the corresponding
gap vector identification observation $\tilde{y}_{1}^{(I)}$ equals
$\pi/3$. As can be seen from the matrix in Figure~\ref{fig:matrix},
for $\yi=1$ this corresponds to $a_{\yi,\xj}^{(I)}$ for $\xj=4$.
Hence the decoder makes a {}``correct identification'' of $\xj=4$,
and estimates (also correctly) that $\hat{x}_{4}$ equals the magnitude
of $\tilde{y}_{1}^{(I)}$, which equals $1$. On checking with the
verification entry, the decoder observes also that $1$ (the detected
value of $\hat{x}_{4}$) times the corresponding verification entry
$a_{1,4}^{(V)}$ equals $\tilde{y}_{1}^{(V)}$. Hence it declares that
$\yi=1$ is a leaf node. Similarly, we know that $3$ is a leaf node
too. Therefore, the leaf node set equals $\{1,3\}$. The entire process
of making a list of leaf node takes $\cO(\spar)$ number of steps.
Suppose in the first iteration, the decoder picks $\yi=1$. Hence
it updates the value of $\hat{x}_{4}$ to $1$, the neighbourly set
to $\{2,3,4\}$, the leaf node set to $\{2,3,4\}$ and $\tilde{\by}$
to the values shown (only the three indices $1$, $3$ and $4$ on
the right need to be changed). At this point, note that $\cSpx$ also
changes from $\{2,4\}$ to the singleton set $\{4\}$. This entire
iteration takes a constant number of steps.}

\label{fig:bip_fig4}
\end{figure*}

\begin{figure*}[p]
 \includegraphics[scale=0.55]{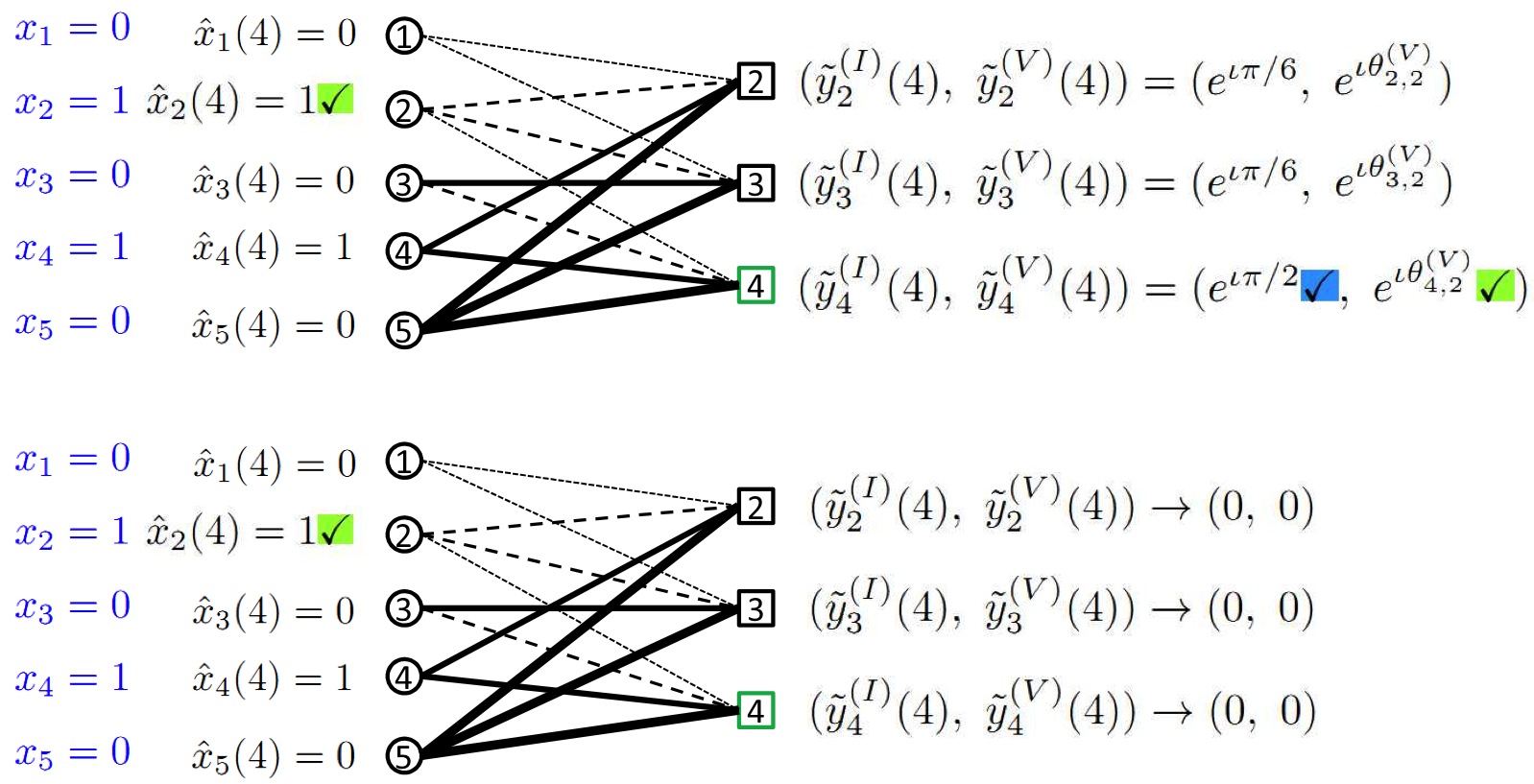} \caption{\underline{\it Second iteration and Termination}: In the second iteration,
the decoder randomly picks $\yi=4$ from the leaf node set $\{2,3,4\}$.
Recall that in the first iteration this choice of $\yi$ did not aid
in decoding. However, now that node $4$ on the right of $\graph$
has been {}``cleaned up'', it is now a leaf node for
$\cSpx$. This demonstrates the importance of not {}``throwing away''
information which seems useless at some point in time. Hence, analogously
to the process in Figure~\ref{fig:bip_fig4}, the decoder estimates
the value of $\hat{x}_{2}$ to $1$, updates the leaf node set to
the empty set, and $\tilde{\by}$ to the all-zero vector (all in a
constant number of steps). Since the gap vector is zero, this indicates
to the decoder that it should output $\hat{\bx}$ as its estimate
of $\bx$, and terminate. }

\label{fig:bip_fig5}
\end{figure*}

\subsubsection{Formal description of SHO-FA's reconstruction process}
{\label{subset:shofa_decode}}
Our algorithm proceeds iteratively, and has at
most $\spar$ overall number of iterations, with $\iter$ being the
variable indexing the iteration number.
\begin{enumerate}
\item \underline{\it Initialization}: We initialize by
setting the \textit{signal estimate vector}
$\hat{\bx}{(1)}$ to the all-zeros vector $0^{\samp}$, and the \textit{residual
measurement identification/verification vectors}
$\tilde{\by}^{(I)}(1)$ and $\tilde{\by}^{(V)}(1)$ to the decoder's
observations $\by^{(I)}$ and $\by^{(V)}$.
\item \underline{\it Leaf-Node List}: Let $\mathcal{L}(1)$,
the initial $\cSx$-leaf node set, be the set of indices $\yi$ which
are $\cSx$-leaf nodes. We generate the list via the following steps:

\begin{enumerate}
\item {\em Compute angles $\theta^{(I)}(\yi)$ and $\theta^{(V)}(\yi)$:}\label{step:computephase}
Let the \textit{identification and verification angles}
be defined respectively as the phases of the identification and verification
entries being considered for index $i$ (starting from 1), as follows:
\begin{eqnarray*}
\theta^{(I)}(\yi) & \triangleq & \angle\left(\tilde{y}_{i}^{(I)}(1)\right),\\
\theta^{(V)}(\yi) & \triangleq & \angle\left(\tilde{y}_{i}^{(V)}(1)\right).
\end{eqnarray*}
 Here $\angle(.)$ computes the phase of a complex number (up to $\cO(\max\{\log{\samp/\spar},\log(\spar))\})$
bits of precision)%
\footnote{Roughly, the former term guarantees that the identification
angle is calculated precisely enough, and the latter that the verification
angle is calculated precisely enough. }%
.
\item {\em Check if the current identification and verification
angles correspond to a valid and unique $x_{\xj}$:} For this, we
check at most two things (both calculations are done up to the precision
specified in the previous step).

\begin{enumerate}
\item First, we check if $\xj\triangleq\theta^{(I)}(i)(2\samp/\pi)$
is an integer, and the corresponding $\xj^{th}$ element of the $\yi^{th}$
row is non-zero. If so, we have {}``tentatively identified{''}
that the $\yi^{th}$ component of $\tilde{\by}$ is a leaf-node of
the currently unidentified non-zero components of $\bx$, and in particular
is connected to the $\xj^{th}$ node on the left, and the algorithm
proceeds to the next step below. If not, we simply increment $\yi$
by $1$ and return to Step (\ref{step:computephase}).
\item Next, we verify our estimate from the previous step.
If $a_{\yi,\xj}^{(V)}\tilde{y}_{i}^{(I)}/a_{\yi,\xj}^{(I)}=\tilde{y}_{i}^{(V)}$,
the verification test passes, and include $\yi$ in $\mathcal{L}(1)$.
If not, we simply increment $\yi$ by $1$ and return to Step (\ref{step:computephase}).
\end{enumerate}
\end{enumerate}
\item \underline{\it Operations in $\iter^{th}$ iteration:} {The $\iter^{th}$ decoding iteration accepts as its
input the $\iter^{th}$ signal estimate vector {$\hat{\bx}\ofiter$},
the $\iter^{th}$ leaf node set $\mathcal{L}(\iter)$, and the $\iter^{th}$
residual measurement identification/verification vectors $(\tilde{\by}^{(I)}\ofiter,\tilde{\by}^{(V)}\ofiter)$.
In $\cO(1)$ steps it outputs the $(\iter+1)^{th}$ signal estimate
vector {$\hat{\bx}\ofiterplusone$}, the $(\iter+1)^{th}$ leaf node
set $\mathcal{L}(\iter+1)$, and the $(\iter+1)^{th}$ residual measurement
identification/verification vectors $(\tilde{\by}^{(I)}\ofiterplusone,\tilde{\by}^{(V)}\ofiterplusone)$
after the performing the following steps sequentially (each of which
takes at most a constant number of atomic steps):}

\begin{enumerate}
\item {\em Pick a random $\yi\ofiter\in\mathcal{L}(\iter)$:}
The decoder picks an element $\yi\ofiter$ uniformly at random from the $\iter^{th}$
leaf-node list $\mathcal{L}(\iter)$.
\item {\em Compute angles $\theta^{(I)}({\iter})$ and
$\theta^{(V)}({\iter})$:}\label{step:computephase-1} Let the \textit{current
identification and verification angles} be defined
respectively as the phases of the residual identification and verification
entries being considered in that step, as follows:
\begin{eqnarray*}
\theta^{(I)}(\iter) & \triangleq & \angle\left(\tilde{y}_{i\ofiter}^{(I)}\ofiter\right),\\
\theta^{(V)}(\iter) & \triangleq & \angle\left(\tilde{y}_{i\ofiter}^{(V)}\ofiter\right).
\end{eqnarray*}

\item \noindent {\em Locate non-zero entry $\xj$ and
derive the value of $\hat{x}_{\xj\ofiter}\ofiter$:} For this, we do at most two things
(both calculations are done up to the precision specified in the previous
step).

\begin{enumerate}
\item First, we calculate $\xj\ofiter\triangleq\theta^{(I)}(\iter)(2\samp/\pi)$.
We have identified that the $\yi^{th}$ component of $\tilde{\by}$
is a leaf-node of the currently unidentified non-zero components of
$\bx$, and in particular is connected to the $\xj\ofiter^{th}$ node
on the left, and the algorithm proceeds to the next step below.
\item Next, we assign the value, $\tilde{y}_{i\ofiter}^{(I)}\ofiter/a_{\yi\ofiter,\xj\ofiter}^{(I)}=\tilde{y}_{i\ofiter}^{(V)}\ofiter/a_{\yi\ofiter,\xj\ofiter}^{(V)}$,
to $\hat{x}_{\xj\ofiter}\ofiter$ and proceeds the algorithm to the
next step below.
\end{enumerate}
\item {\em Update $\hat{\bx}\ofiterplusone$, $\mathcal{L}(t+1)$
,$\tilde{\by}^{(I)}\ofiterplusone$, and $\tilde{\by}^{(V)}\ofiterplusone$:}
In particular, at most $7$ components of each of these vectors need
to be updated. Specifically, $\hat{x}_{\xj\ofiter}\ofiterplusone$
equals $\tilde{y}_{i\ofiter}^{(I)}\ofiter/a_{\yi\ofiter,\xj\ofiter}^{(I)}$.
$\yi\ofiter$ is removed from the leaf node set $\mathcal{L}(t)$
and check whether the (at most six) neighbours of $\hat{x}_{\xj\ofiter}\ofiter$
become leaf node to get the leaf-node list $\mathcal{L}(t+1)$. And
finally (seven) values each of $\tilde{\by}^{(I)}\ofiterplusone$
and $\tilde{\by}^{(V)}\ofiterplusone$ are updated from those of $\tilde{\by}^{(I)}\ofiter$
and $\tilde{\by}^{(V)}\ofiter$ (those corresponding to the neighbours
of $\hat{x}_{\xj\ofiter}\ofiter$) by subtracting out $\hat{x}_{\xj\ofiter}\ofiter$
multiplied by the appropriate coefficients of $\bA$.
\end{enumerate}
\item {\underline{Termination:}} The algorithm stops when the
leaf node set is empty, and outputs the last $\hat{\bx}(\iter)$.
\end{enumerate}

\subsection{Decoding complexity}
\label{sec:noiselessiterations}
 We start by generating $\mathcal{L}(1)$, the initial list  of leaf nodes. 
For each node $i$, we calculate the identification and verification angles
(which takes $2$ operations), and then check if the identification and verification angles
correspond to a valid and unique $x_{j}$ (which takes $2$ operations). 
Therefore generating the initial list of leaf nodes takes $\cO(\spar)$ (to be precise
$4c\spar$) operations .



In iteration $\iter$, we decode a new non-zero entry $x_\xj$ of $\bx$ by picking a leaf node from $\mathcal{L}(\iter)$,
identifying the corresponding index $\xj$ and value $x_\xj$ (via $2$ arithmetic operations corresponding to the identification and verification steps respectively), and updating
$\mathcal{L}(\iter+1)$ (since $x_\xj$ is connected to $3$ nodes on the right, out of which one has already been decoded, this takes at most $4$ operations -- $2$ for identification and $2$ for verification), $\tilde{\by}^{(I)}\ofiterplusone$,
and $\tilde{\by}^{(V)}\ofiterplusone$ (similarly, this takes at most $8$ operations -- $4$ additions and $4$ multiplications).

Next we note that each iteration results in recovering a new non-zero coordinate of
$\bx$ (assuming no decoding errors, which it true with high probability as demonstrated in the next section). Hence the total number of iterations is at most $\spar$.

Hence the overall number of operations over all iterations is $\cO(\spar)$ (to be precise, at most $4c\spar+14\spar$).

\subsection{Correctness}\label{subsec:correctness}

Next, we show that $\hat{\bx}=\bx$ with high probability over $\bA$. To show
this, it suffices to show that each non-zero update to the estimate
$\hat{\bx}\ofiter$ sets a previously zero coordinate to the
correct value with sufficiently high probability.

Note that if $\yi\ofiter$ is a leaf node for {$\cS\ofiter$}, and
if all non-zero coordinates of $\hat{\bx}\ofiter$ are equal to the
corresponding coordinates in $\bx$, then the decoder correctly identifies
the parent node {$j\ofiter\in\cS\ofiter$} for the leaf node $\yi\ofiter$
as the unique coordinate that passes the phase identification and
verification checks.

{ Thus, the {$\iter^{th}$} iteration ends with an erroneous update
only if {
\[
\angle \left (\sum_{p\in N(\{\yi\ofiter\})}x_{p}e^{\iota\theta_{i(t),p}^{(I)}} \right )=\theta_{i(t),j(t)}^{(I)}
\]
} for some $\xj$ such that there are more than one non-zero terms
in the summation on the left.

{
\[
\angle \left (\sum_{p\in N(\{\yi\ofiter\})}x_{p}e^{\iota\theta_{i(t),p}^{(V)}} \right )=\theta_{i(t),j(t)}^{(V)}
\]
} {Since $\theta_{i(t),j(t)}^{(V)}$ is drawn uniformly
at random from $[1,2,\ldots,\pi/2]$ (with $\Omega(\log(\samp)+\preci) = \cO(\log(\spar))$ (say) bits
of precision), the probability that the second equality holds with
more than one non-zero term in the summation on the left is at most
$o(1/poly(\spar))$.} { The above analysis gives an upper bound on the
probability of incorrect update for a single iteration to be $o(1/poly(\spar))$.
Finally, as the total number of updates is at most $\spar$, by applying
a union bound over the updates, the probability of incorrect decoding
is bounded from above by $o(1/poly(\spar))$. }

\subsection{Remarks on the Reconstruction process for exactly $\spar$-sparse
signals}

We elaborate on these choices of entries of $\bA$ in the remarks
below, which also give intuition about the reconstruction process
outlined in Section~\ref{subset:shofa_decode}.

\begin{remark}\label{rem:ident} In fact, it is not critical that (\ref{eq:a_iden})
be used to assign the identification entries. As long as $\xj$ can
be {}``quickly{''} (computationally efficiently) identified
from the phases of $a_{\yi,\xj}^{(I)}$ (as outlined in Remark~\ref{rem:quickident}
below, and specified in more detail in Section~\ref{subset:shofa_decode}),
this suffices for our purpose. This is the primary reason we call
these entries identification entries.
\end{remark}
\begin{remark}\label{rem:quickident} The reason for the choice of phases
specified in (\ref{eq:a_iden}) is as follows. Suppose $\cSx$ corresponds
to the support (set of non-zero values) of $\bx$. Suppose $y_{\yi}$
corresponds to a $\cSx$-leaf node, then by definition {$y_{\yi}^{(I)}$}
equals $a_{\yi,\xj}^{(I)}x_{\xj}$ for some $\xj$ in $\{1,\ldots,\samp\}$
(if $y_{\yi}$ corresponds to a $\cSx$-\textit{non}-leaf node, then
in general {$y_{\yi}^{(I)}$} depends on two or more $x_{\xj}$).
But $x_{\xj}$ is a real number. Hence examining the phase of $y_{\yi}$
enables one to efficiently compute $\xj\pi/(2\samp)$, and hence $\xj$.
It \textit{also} allows one to recover the magnitude of $x_{\xj}$,
simply by computing the magnitude of $y_{\yi}$.\end{remark}

\begin{remark} The choice of phases specified in
(\ref{eq:a_iden}) divides the set of allowed phases (the interval
$[0,\pi/2]$) into $\samp$ distinct values. Two things are worth
noting about this choice.
\begin{enumerate}
\item We consider the interval $[0,\pi/2]$ rather than the full range $[0,2\pi)$
of possible phases since we wish to use the phase measurements to
\textit{also} recover the sign of $x_{\xj}$s. If the phase of $y_{\yi}$
falls within the interval $[0,\pi/2]$, then (still assuming that
$y_{\yi}$ corresponds to a $\cSx$-leaf node) $x_{\xj}$ must have
been positive. On the other hand, if the phase of $y_{\yi}$ falls
within the interval $[\pi,3\pi/2]$, then $x_{\xj}$ must have been
negative. (It can be directly verified that the phase of a $\cSx$-leaf
node $y_{\yi}$ can never be outside these two intervals -- this wastes
{roughly} half of the set of possible phases we could have used
for identification purposes, but it makes notation easier.
\item The choice in (\ref{eq:a_iden}) divides the interval $[0,\pi/2]$
into $\samp$ distinct values. However, in expectation over $\graph$
the actual number of non-zero entries in a row of $\bA$ is $\cO(\samp/\spar)$,
so on average one only needs to choose $\cO(\samp/\spar)$ distinct
phases in (\ref{eq:a_iden}), rather than the worst case $\samp$
number of values. This has the advantage that one only needs $\cO(\log(\samp/\spar))$
bits of precision to specify distinct phase values (and in fact we
claim that this is the level of precision required by our algorithm).
However, since we analyze only left-regular $\graph$, the degrees
of nodes on the right will in general vary stochastically around this
expected value. If $\spar$ is {}``somewhat large{''} (for instance
$\spar=\Omega(\samp)$), then the degrees will not be very tightly
concentrated around their mean. One way around this is to choose $\graph$
uniformly at random from the set of bipartite graphs with $\samp$
nodes (each of degree $\dgr$) on the left and $\meas$ nodes (each
of degree $\dgr\samp/\meas$) on the right. This would require a more
intricate proof of the $\cSpx$-expansion property defined in Property~\ref{prop:3}
and proved in Lemma~\ref{lem:expansion}. For the sake of brevity,
we omit this proof here.
\end{enumerate}\end{remark}
\begin{remark} In fact, the recent work of~\cite{GooM:11}
demonstrates an alternative analytical technique (bypassing the expansion
arguments outlined in this work), involving analysis of properties
of the {}``$2$-core'' of random hyper-graphs, that allows for a
tight characterization of the number of measurements required by SHO-FA
to reconstruct $\bx$ from $\by$ and $\bA$, rather than the somewhat
loose (though order-optimal) bounds presented in this work. Since
our focus in this work is a simple proof of order-optimality (rather
than the somewhat more intricate analysis required for the tight characterization)
we again omit this proof here.%
\footnote{We thank the anonymous reviewers who examined a previous version of
this work for pointing out the extremely relevant techniques of~\cite{GooM:11}
and~\cite{Pri:11} (though the problems considered in those works
were somewhat different).%
}\end{remark}

\subsection{{Other properties of SHO-FA}}
\subsubsection{{SHO-FA v.s. {}``$2$-core'' of random hyper-graphs}}
\label{2-core}
We reprise some concepts pertaining to the analysis of random hypergraphs (from~\cite{Molloy:2005}), which are relevant to our work.  \\
\noindent \underline{$2$-cores of $d$-uniform hypergraphs:}
A {\it $d$-uniform hypergraph with $m$ nodes
and $k$ hyperedges} is defined over a set of $m$ nodes, with each $d$-uniform hyperedge corresponding to a subset of the nodes of size exactly $d$.
The {\it $2$-core} is defined as the largest {\it sub-hypergraph} (subset of nodes, and hyperedges defined only on this subset) such that each node in this subset is contained in at least $2$ hyperedges on this subset.

A standard {}``peeling process'' that computationally efficiently finds the $2$-core
is as follows: while there exists a node with degree $1$ (connected to just one hyperedge), delete
it and the hyperedges containing it.

\noindent \underline{The relationship between $2$-cores in $d$-uniform hypergraphs and SHO-FA:}
As in~\cite{GooM:11} and other works, there exists bijection between $d$-uniform hypergraphs and $d$-left-regular bipartite graphs, which can be constructed as follows:
\begin{enumerate}[\hspace{0.3cm}(a)]
\item Each hyperedge in the hypergraph is mapped to a left node in the bipartite
graph,
\item Each node in the hypergraph is mapped to a right node in the
bipartite graph,
\item The edges leaving a left-node in the bipartite graph correspond to the nodes contained in the corresponding hyperedge of the hypergraph.
\end{enumerate}

Suppose the $d$-uniform hypergraph does not contain a $2$-core.
This means that, in each iteration of {}``peeling process'', we
can find a vertex with degree $1$, delete it and the corresponding
hyperedges. and continue the iterations until all the hyperedges are
deleted. Correspondingly, in the bipartite graph, we can find a leaf node
in each iteration, delete it and the corresponding left node and continue
the iterations until all left nodes are deleted. We note that 
the SHO-FA algorithm follows essentially the same process.
This implies that SHO-FA succeeds if and only if the $d$-uniform hypergraph contains a $2$-core. 

\noindent \underline{Existence of 2-cores in d-uniform hypergraphs and SHO-FA:}
We now reprise a nice result due to~\cite{Molloy:2005} that helps us tighten the results of Theorem~\ref{thm:main}.

{
\begin{theorem} \label{2core}(\cite{Molloy:2005}) Let
$d+l>4$ and $G$ be a $d$-uniform random hypergraph with $k$ hyperedges
and $m$ nodes. Then there exists a number $c_{d,l}^{*}$ (independent of $k$ and $m$ that is the threshold for the appearance
of an $l$-core in $G$. That is, for constant $c$ and $m\rightarrow\infty$: If $k/m=c<c_{d,l}^{*}$,
then $G$ has an empty $l$-core with probability $1-o(1)$;If $k/m=c>c_{d,l}^{*}$, then $G$ has an $l$-core of linear size with probability $1-o(1)$.
\end{theorem}

Specifically, the results in~\cite[Theorem 1]{GooM:11} (which explicitly calculates some of the $c_{d,l}^{*}$) give that
for $l=2$ and $d=3$, $c_{3,2}^{*}=1/1.22$ with probability $1-\cO(1/\spar)$. This leads to the following theorem, that has tighter parameters than Theorem~\ref{thm:main}.

\begin{theorem} \label{thm:tighter} Let $\spar\leq\samp$. There exists a reconstruction algorithm SHO-FA for {$\bA\in\bbr^{\meas\times\samp}$ }with the following properties:
\begin{enumerate}
\item For every {$\bx\in\bbR^{\samp}$}, with probability $1-\cO(1/\spar)$ over the choice of $\bA$, SHO-FA produces a reconstruction $\hat{\bx}$ such that $||\bx-\hat{\bx}||_{1}/||\bx||_{1}\leq2^{-P}$
\item  The number of measurements  $\meas=2\cnst\spar$, $\forall$ $\cnst>1.22$.
\item The {number} of steps required by SHO-FA is $4\cnst\spar+14\spar$.
\item The {number} of bitwise arithmetic operations required by SHO-FA is $\cO(\spar(\log{\samp}+P))$.
\end{enumerate}
\end{theorem} 

\begin{remark} We note that by carefully choosing a {\it degree distribution} (for instance the ``enhanced harmonic degree distribution''~\cite{LubyMS:2006}) rather than constant degree $\dgr$  for left nodes in the bipartite graph, the constant $c$ can be made to approach to $1$, while still ensuring that $2$-cores do not occur with sufficiently high probability. This can further reduce the number of measurements $\meas$. However, this can come at a cost in terms of computational complexity, since the complexity of SHO-FA depends on the average degree of nodes on the left, and this is not constant for the enhanced harmonic degree distribution. \\
\end{remark}
\begin{remark}To get the parameters corresponding to $\meas$ of Theorem~\ref{thm:tighter} we have further reduced (compared to Theorem~\ref{thm:main}) the number of measurements by a factor of $2$ by combining the identification measurements and verification measurements. This follows directly from the observation that we can construct the phase of each non-zero matrix entry via first ``structured'' bits ($\log(\samp/\spar)$ corresponding to the bits we would have chosen corresponding in an identification measurement), followed by ``random''  bits ($\log(\spar)$ corresponding to the bits we would have chosen corresponding in an identification measurement). Hence a single measurement can serve the role of both identification and measurement. \\
\end{remark}
\begin{remark} The results in Theorem~\ref{2core} also indicate a ``phase transition'' on the emergence of $l$-cores. These results explain our simulation
results, presented in Appendix~\ref{app:sim}. \\
\end{remark}
}

\subsubsection{Database queries}
A useful property of our construction of the matrix $\bA$ is that
any desired signal component $x_{\xj}$ can be reconstructed in constant time with
a constant probability from measurement vector $\by=\bA\bx$. 
The following Lemma makes this precise. The proof
follows from a simple probabilistic argument and is included in Appendix~\ref{app:query}.
\begin{lemma}\label{lem:query}Let $\bx$ be $\spar$-sparse. Let
$\xj\in\{1,2,\ldots,\samp\}$ and let $\bA\in\mathbb{C}^{\cnst\spar\times\samp}$
be randomly drawn according to SHO-FA. Then, there exists an algorithm
$\mathcal{A}$ such that given inputs $(\xj,\by)$, $\mathcal{A}$
produces an output $\hat{x}_{\xj}$ with probability at least $(1-(\dgr/\cnst)^{\dgr})$
such that $\hat{x}_{\xj}=x_{\xj}$ with probability $(1-o(1/poly(\spar)))$.
\end{lemma}

\subsubsection{SHO-FA for sparse vectors in different bases}

\label{diff_bases}
In the setting of SHO-FA we consider $\spar$-sparse
input vectors $\bx$. In fact, we also can deal with the case that
$\bx$ is sparse in a certain basis that is known \textit{a priori}
to the decoder~%
\footnote{For example, {}``smooth'' signals are sparse in the Fourier basis
and {}``piecewise smooth'' signals are sparse in wavelet bases.%
}, say $\Psi$, which means that $\bx=\Psi\mathbf{w}$ where $\mathbf{w}$
is a $\spar$-sparse vector. Specifically, in this case we write the
measurement vector as $\by=B\bx$, where $B=A\Psi^{-1}$. Then, $\by=A\Psi^{-1}\Psi\mathbf{w}=\bA\mathbf{w}$,
where $\bA$ is chosen on the structure of the $\graph$ and $\mathbf{w}$
is a $\spar$-sparse vector. We can then apply SHO-FA to reconstruct
$\mathbf{w}$ and consequently $\bx=\Psi\mathbf{w}$. What has been
discussed here covers the case where $\bx$ is sparse itself, for
which we can simply take $\Psi=I$ and $\bx=\mathbf{w}$.

\subsubsection{Information-theoretically order-optimal encoding/update complexity}
The sparse structure of $\bA$ also ensures (``for free'') order-optimal encoding and update complexity of the measurement process.

We first note that for any measurement matrix $\bA$ that has a ``high probability'' (over $\bA$) of correctly reconstructing arbitrary $\bx$,
there is a lower bound of $\Omega(\samp)$ on the computational complexity of computing $\bA\bx$.
This is because if the matrix does not have at least $\Omega(\samp)$ non-zero entries, then with probability $\Omega(1)$ (over $\bA$) at least one non-zero entry of $\bx$ will ``not be measured'' by $\bA$, and hence cannot be reconstructed.
In the SHO-FA algorithm, since the degree of each left-node in $\graph$ is a constant ($\dgr$), 
the encoding complexity of our measurement process corresponds to $\dgr\samp$ multiplications and additions.

Analogously, the complexity of updating $\by$ if a single entry of $\bx$ changes is at most $\dgr$ for SHO-FA, which matches the natural lower bound of $1$ up to a constant ($\dgr$) factor.




\subsubsection{Information-theoretically optimal number of bits}
\label{subsec:bits} We recall that the reconstruction goal for SHO-FA
is to reconstruct $\bx$ up to relative error $2^{-\preci}$. That
is,
\[
||\bx-\hbx||_{1}/||\bx||_{1}\leq2^{-\preci}.
\]

We first present a sketch of an information-theoretic lower bound
of $\Omega\left(\spar\left(\preci+\log\samp\right)\right)$ bits holds
for any algorithm that outputs a $\spar$-sparse vector that achieves
this goal with high probability.

To see this is true, consider the case where the locations of $\spar$
non-zero entries in $\bx$ are chosen uniformly at random among all
the $\samp$, entries and the value of each non-zero entry is chosen
uniformly at random from the set $\{1,\ldots,2^{\preci}\}$. Then
recovering even the support requires at least $\log\left(2^{\spar P}\binom{\samp}{\spar})\right)$
bits, which is $\Omega(\spar\log(\samp/\spar))$.%
\footnote{Stirling's approximation(c.f.~\cite[Chapter~1]{Mackay:03}) is used
in bounding from below the combinatorial term ${\samp \choose \spar}$.%
} Also, at least a constant fraction of the $\spar$ non-zero entries
of $\bx$ must be be correctly estimated to guarantee the desired
relative error. Hence $\Omega\left(\spar\left(\preci+\log\samp\right)\right)$
is a lower bound on the measurement bit-complexity. 


The following arguments show that the total number of bits used in
our algorithm is information-theoretically order-optimal for any $\spar=\cO(\samp^{1-\Delta})$
(for any $\Delta>0$). First, to represent each non-zero entry of
$\bx$, we need to use arithmetic of $\Omega(\preci+\log(\spar))$
bit precision. Here the $\preci$ term is so as to attain the required
relative error of reconstruction, and the $\log(\spar)$ term is to
take into account the error induced by finite-precision arithmetic
(say, for instance, by floating point numbers) in $\cO(\spar)$ iterations
(each involving a constant number of finite-precision additions and
unit-magnitude multiplications). Second, for each identification
step, we need to use $\Omega(\log(\samp)+\log(\spar))$ bit-precision
arithmetic. Here the $\log(\samp)$ term is so that the identification
measurements can uniquely specify the locations of non-zero entries
of $\bx$. The $\log(\spar)$ term is again to take into account the
error induced in $\cO(\spar)$ iterations. Third, for each verification
step, the number of bits we use is (say) $3\log(\spar)$. Here, by the
Schwartz-Zippel Lemma~\cite{Sch:80,Zip:79}, $2\log(\spar)$ bit-precision
arithmetic guarantees that each verification step is valid with probability
at least $1-1/\spar^{2}$ -- a union bound over all $\cO(\spar)$
verification steps guarantees that all verification steps are correct
with probability at least $1-\cO(1/\spar)$ (this probability of success can be directly amplified by using higher precision arithmetic). Therefore, the total
number of bits needed by SHO-FA $\cO(\spar(\log(\samp)+\preci))$.
As claimed, this matches, up to a constant factor, the lower bound
sketched above.

\subsubsection{Universality}

\label{subsec:univ}
\begin{figure}[t]
\begin{centering}
\includegraphics[scale=0.45]{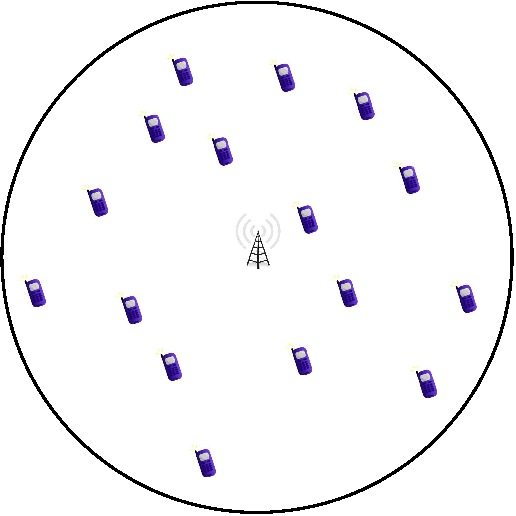}
\par\end{centering}

\caption{\label{fig:base_station} \underline{\it An example of a physical system that ``naturally'' generates ensembles of sparse $\bA$ that SHO-FA can use}:
Suppose there are $\spar$ cellphones (out of a set of $\samp$ possible
different cellphones in the whole world) in a certain neighbourhood
that has a base-station. The goal is for the $\xj$-th cellphone to
communicate its information ($x_{\xj}$) to the base-station at least
once per \textit{frame} of $\cnst\spar$ consecutive time-slots. The
challenge is to do so in a distributed manner, since multiple cellphones
transmitting at the same time $\yi$ would result in a linear combination
$y_{\yi}=\sum_{\xj}a_{\yi\xj}x_{\xj}$ of their transmissions reaching
the base-station, where $a_{\yi\xj}$ corresponds to the channel gain
from the cellphone $\xj$ to the base-station during time-slot $\yi$. With high probability, such $a_{\yi\xj}$ satisfying the properties we require for our algorithm to work -- ``sparsity'' (relatively few transmitters transmit at any given time) and ``distinctness'' (with high probability the channel gains from different transmitters are different).
Each cellphone transmits $x_{\xj}$ to the base-station a constant
($\dgr$) number of times in each frame -- the set of $\dgr$ time-slots
in each frame that cellphone $\xj$ transmits in is chosen by $\xj$
uniformly at random from the set of all ${{\cnst\spar} \choose {\dgr}}$
sets of slots. }
\end{figure}

While the ensemble of matrices $\{\bA\}$ we present above has carefully
chosen identification entries, and all the non-zero verification entries
have unit magnitude, as noted in Remark~\ref{rem:universality}, the implicit ideas underlying
SHO-FA work for significantly more general ensembles of matrices.
In particular, Property~\ref{prop:1} only requires that the graph
$\graph$ underling $\bA$ be {}``sparse'', with a constant number
of non-zero entries per column. Property~\ref{prop:2} only requires
that each non-zero entry in each row be distinct -- which is guaranteed
with high probability, for instance, if each entry is chosen $i.i.d$
from any distribution with sufficiently large support. An example
of such a scenario is shown in Figure~\ref{fig:base_station}. This
naturally motivates the application of SHO-FA to a variety of scenarios,
for $\it{e.g.}$, neighbor discovery in wireless communication~\cite{Guo:10}.

\section{Approximate reconstruction in the presence of noise}\label{sec:noisy}
A prominent aspect of the design presented in the previous section is that it relies on exact determination of all the phases as well as magnitudes of the measurement vector $\bA\bx$. In practice, however, we often desire that the measurement and {reconstruction} be robust to corruption both before and and during measurements. In this section, we show that our design may be modified slightly such that with a suitable decoding procedure, the reconstruction is robust to such "noise".

We consider the following setup. Let $\bx\in\bbr^\samp$ be a $\spar$-sparse signal with support $\cSx=\{\xj:x_{\xj}\neq 0\}$. Let $\bz\in\bbr^\samp$ have support $\{1,2,\ldots,\samp\}\setminus\cSx$ with each $z_\xj$ distributed according to a Gaussian distribution with mean $0$ and variance $\sigma_z^2$. Denote the measurement matrix by $\bA\in\bbc^{\meas\times\samp}$ and the measurement vector by $\by\in\bbc^\meas$. Let $\be\in\bbc^{\meas}$ be the measurement noise with $e_\yi$ distributed as a Complex Gaussian with mean $0$ and variance $\sigma_e^2$ along each axis. $\by$ is related to the signal as $$\by=\bA(\bx+\bz)+\be.$$

We first propose a design procedure for $\bA$ satisfying the following properties. 
\begin{theorem}\label{thm:approximate} Let $\spar=\cO(\samp^{1-\Delta})$ for some $\Delta>0$. There exists a reconstruction algorithm {SHO-FA} for $\bA\in\bbc^{\meas\times\samp}$ such that
\begin{enumerate}[($i$)]
\item $\meas=\cnst\spar$
\item {SHO-FA} consists of at most $4\spar$ iterations, each involving a constant number of arithmetic operations with a precision of $\cO(\log{\samp})$ bits.
\item With probability $1-o(1/\spar)$ over the design of $\bA$ and randomness in $\be$ and $\bz$,
$$||\hat{\bx}-\bx||_1\leq C\left(||\bz||_1+\sqrt{\log{\spar}}||\be||_1\right)$$ for some $C=C(\sigma_z,\sigma_e)>0$.
\end{enumerate}
\end{theorem}
We present a ``simple'' proof of the above theorem in Sections~\ref{sec:new ideas} to~\ref{app:recons}. In Theorem~\ref{thm:approximate2}, we outline an analysis (based on the work of~\cite{Pri:11}) that leads to a tighter characterization of the constant factors in the parameters of Theorem~\ref{thm:approximate}.

Recall that in the exactly $\spar$-sparse case, the decoding in $\iter$-th iteration relies on first finding {an $\cS\ofiter$-leaf node}, then decoding the corresponding signal coordinate and updating the undecoded measurements. In this procedure, it is critical that each iteration operates with low reconstruction errors as an error in an earlier iteration can propagate and cause potentially catastrophic errors. In general, one of the following events may result in any iteration ending with a decoded signal value that is far from the true signal value:

\begin{enumerate}[(a)]
\item The decoder picks an index outside the set $\{\yi:(\bA\bx)_{\yi}\neq 0\}$, but in the set $\{\yi:(\bA(\bx+\bz)+\be)_{\yi}\neq 0\}\}$.
\item The decoder picks an index within the set $\{\yi:(\bA\bx)_{\yi}\neq 0\}$ that is also a leaf for $\mathcal{S}$ with parent node $\xj$, but the presence of noise results in  the decoder identifying (and verifying) a node $\xj'\neq\xj$ as the parent and, subsequently, incorrectly decoding the signal at $\xj'$.
\item The decoder picks an index within the set $\{\yi:(\bA\bx)_{\yi}\neq 0\}$ that is not a leaf for $\mathcal{S}$, but the presence of noise results in  the decoder identifying (and verifying) a node $\xj$ as the parent and, subsequently, incorrectly decoding the signal at $\xj$.
\item The decoder picks an index within the set $\{\yi:(\bA\bx)_{\yi}\neq 0\}$ that is a leaf for $\mathcal{S}$ with parent node $\xj$, which it also identifies (and verifies) correctly, but the presence of noise introduces a small error in decoding the signal value. This error may also propagate to the next iteration and act as "noise" for the next iteration.
\end{enumerate}

To overcome these hurdles, our design takes the noise statistics into account to ensure that each iteration is resilient to noise with a high probability. This achieved through several new ideas that are presented in the following for ease of exposition. Next, we perform a careful analysis of the corresponding decoding algorithm and show that under certain regularity conditions, the overall failure probability can be made arbitrarily small to output a  reconstruction that is robust to noise. Key to this analysis is bounding the effect of propagation of estimation error as the decoder steps through the iterations. \footnote{For simplicity, the analysis presented here relies only on an upper bound on the length of the path through which the estimation error introduced in any iteration can propagate. This bound follows from known results on size of largest components in sparse hypergraphs~\cite{KarL:02}. We note, however, that a tighter analysis that relies on a finer characterization of the interaction between the size of these components and the contribution to total estimation error may lead to better bounds on the overall estimation error. Indeed, as shown in~\cite{Pri:11}, such an analysis enables us to achieve a tighter reconstruction guarantee of the form $||\bx-\hat{\bx}||_1=\cO(||\bz||_1+||\be||_1)$ }
\subsection{Key ideas}\label{sec:newideas}

\subsubsection{Truncated reconstruction}\label{sec:truncation} We observe that in the presence of noise, it is unlikely that signal values whose magnitudes are comparable to that of the noise values can be successfully recovered. Thus, it is futile for the decoder to try to reconstruct these values as long as the overall penalty in $l_1$-norm is not high. The following argument shows that this is indeed the case. Let
\begin{equation}\cS_{\delta}(\bx)=\{\xj:|x_{\xj}|<\delta/\spar\}.
\label{eq:delta_level}
\end{equation}
and let $\bx_{\cS_\delta}$ be the vector defined as $$(x_{\cS_\delta})_\xj=\left\{\begin{array}{ll}0&\ \xj\notin\cS_\delta(\bx)\\ x_{\xj}&\ \xj\in\cS_\delta(\bx).\end{array}\right.$$
Similarly, define $\bx_{\cS_\delta^c}$ which has non-zero entries only within the set $\cSx\setminus\cS_{\delta}(\bx)$.
The following sequence of inequalities shows that the total $l_1$ norm of $\bx_{S_\delta}$ is small:
\begin{eqnarray}
||\bx_{\cS_\delta}||_1&=&\sum_{\xj\in\cS_\delta(\bx)}|x_\xj|\nonumber\\
&\leq & |\cS_{\delta}(\bx)|\frac{\delta}{\spar}\nonumber\\
&\leq &|\cSx|\frac{\delta}{\spar}\nonumber\\
&=&\delta.\label{eq:deltabound}
\end{eqnarray}

Further, as an application of triangle inequality and the bound in~\eqref{eq:deltabound}, it follows that
\begin{eqnarray}
||\hat{\bx}-\bx||_1&=&||\hat{\bx}-\bx_{\cS_\delta^c}-\bx_{\cS_\delta}||_1\nonumber\\
&\leq & ||\hat{\bx}-\bx_{\cS^c_\delta}||_1+||\bx_{\cS_\delta}||_1\nonumber\\
&\leq &||\hat{\bx}-\bx_{\cS^c_\delta}||_1+\delta\label{eq:deltatriangle}
\end{eqnarray}

Keeping the above in mind, we rephrase our reconstruction objective to satisfy the following criterion with a high probability:
\begin{eqnarray}
||\hat{\bx}-\bx_{\cS^c_\delta}||_1&\leq& C_1(||\bz||_1+\spar^2||\be||_1),\label{eq:criterion1}
\end{eqnarray}
while simultaneously ensuring that our choice of parameter $\delta$ satisfies \begin{equation}\label{eq:deltactriteria}\delta<C_2||\bz||_1
\end{equation}
for some $C_2$, with a high probability.

\subsubsection{Phase quantization}

In the noisy setting, even when $\yi$ is a leaf node for $\cSx$,
the phase of $y_{\yi}$ may differ from the phase assigned by the
measurement. This is geometrically shown in Figure~\ref{fig:maxnoise}
for a measurement matrix $\bA'$. To overcome this, we modify our
decoding algorithm to work with \textquotedbl{}quantized\textquotedbl{}
phases, rather than the actual received phases. The idea behind this
is that if $\yi$ is a leaf node for $\cSx$, then quantizing the
phase to one of the values allowed by the measurement identifies the
correct phase with a high probability. The following lemma facilitates
this simplification.

\begin{lemma}[{\em Almost bounded phase noise}]\label{lem:phasenoise}
Let $\bx,\bz\in\bbr^{\samp}$ with $|x_{\xj}|>\delta/\spar$ for each
$\xj$. Let $\bA'\in\bbc^{\meas'\times\samp}$ be a complex valued
measurement matrix with the underlying graph $\graph$. Let $\yi$
be a leaf node for $\cSx$. Let $\Delta\theta_{\yi}=|\angle y_{\yi}-\angle(\bA'\bx)_{\yi}|$.
Then, for every $\alpha>0$,
\[
E_{\bz,\be}\large(\Delta\theta_{\yi}\large)\leq\sqrt{\frac{2\pi\spar^{2}(\dgr\samp\sigma_{z}^{2}/\cnst\spar+\sigma_{e}^{2})}{\delta^{2}}}
\]
 and
\[
\Pr_{\bz,\be}\bigg(\Delta\theta_{\yi}>\alpha E_{\bz,\be}\large(\Delta\theta_{\yi}\large)\bigg)<\frac{1}{2}e^{-(\alpha^{2}/2\pi)}.
\]
 \end{lemma} 
 \begin{proof}See Appendix~\ref{app:phasenoise}.
 \end{proof}
 For a desired error probability $\epsilon'$, the above lemma stipulates
that it suffices to let $\alpha=\sqrt{2\pi\log(1/2\epsilon')}$. We examine
the effect of phase noise in more detail in {Appendix}~\ref{sec:noisyerror}.
\subsubsection{Repeated measurements}Our algorithm works by performing a series of $\steps\geq 1$ identification and verification measurements in each iteration instead of a single measurement of each type as done in the exactly $\spar$-sparse case. The idea behind this is that, in the presence of noise, even though a single set of identification and verification measurements cannot exactly identify the coordinate $\xj$ from the observed $y_\yi$, it helps us narrows down the set of coordinates $\xj$ that can possibly contribute to give the observed phase. Performing measurements repeatedly, each time with a different measurement matrix, helps us identify a single $\xj$ with a high probability.

We implement the above idea by first mapping each $\xj\in\{1,2,\ldots,\samp\}$ to its $\steps$-digit representation in base $\bbg=\{0,1,\ldots\lceil \samp^{1/\steps}-1\rceil\}$. For each $\xj\in\{1,2,\ldots,\samp\}$, let $\digit(\xj)=(\digit_1(\xj),\digit_2(\xj),\ldots,\digit_\steps(\xj))$ be the $\steps$-digit representation of $\xj$. Next, perform one pair of identification and verification measurements (and corresponding phase reconstructions), each of which is intended to distinguish exactly one of the digits. In our construction, we only need a constant number of such phase measurements per iteration. See Fig~\ref{fig:repeated} for an illustrating example.

\begin{figure}
\begin{center}\includegraphics[height=120pt]{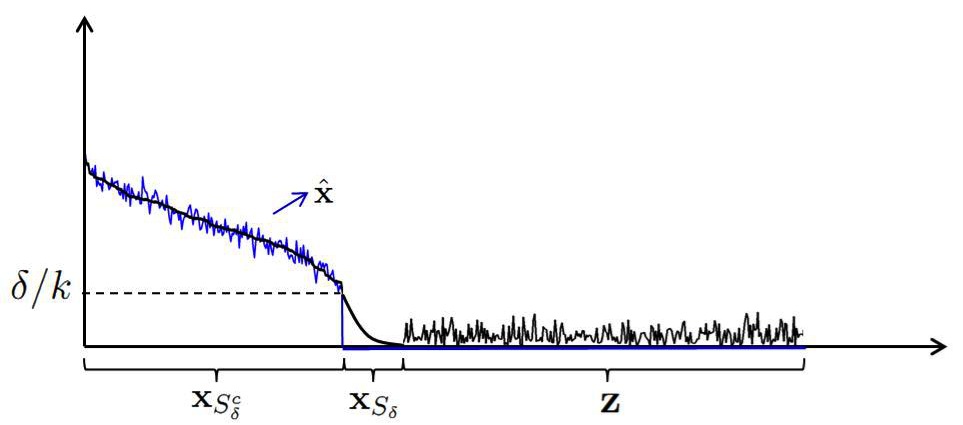}
\end{center}
\caption{The black curve corresponds to the magnitudes of $\bx+\bz$ (for ease of visual presentation, the components of $\bx$ have been sorted in decreasing order of magnitude and placed in the first $\spar$ components of the signal, but the components of $\bz$ are unsorted. The blue curve corresponds to our reconstruction $\hat{\bx}$ of $\bx$. Note that we only attempt to reconstruct components of $\bx$ that are ``sufficiently large" (that is, we make no guarantees about correct reconstruction of components of $\bx$ in ${\cS}_\delta(\bx)$, {\it, i.e}, those components of $\bx$ that are smaller than some ``threshold" $\delta/\spar$. Here $\delta$ is a parameter of code-design to be specified later. As shown in Section~\ref{sec:truncation}, as long as $\delta$ is not ``too large", this relaxation does not violate our relaxed reconstruction criteria~\eqref{eq:criterion1}.}
\end{figure}

\begin{figure}
  \centering
  \subfloat[Maximum phase displacement occurs when the contribution due to noise, i.e., $(\bA\bz)_\yi+e_\yi$ is orthogonal to the measurement $\by_\yi$]{\label{fig:maxnoise}\includegraphics[width=0.3\textwidth]{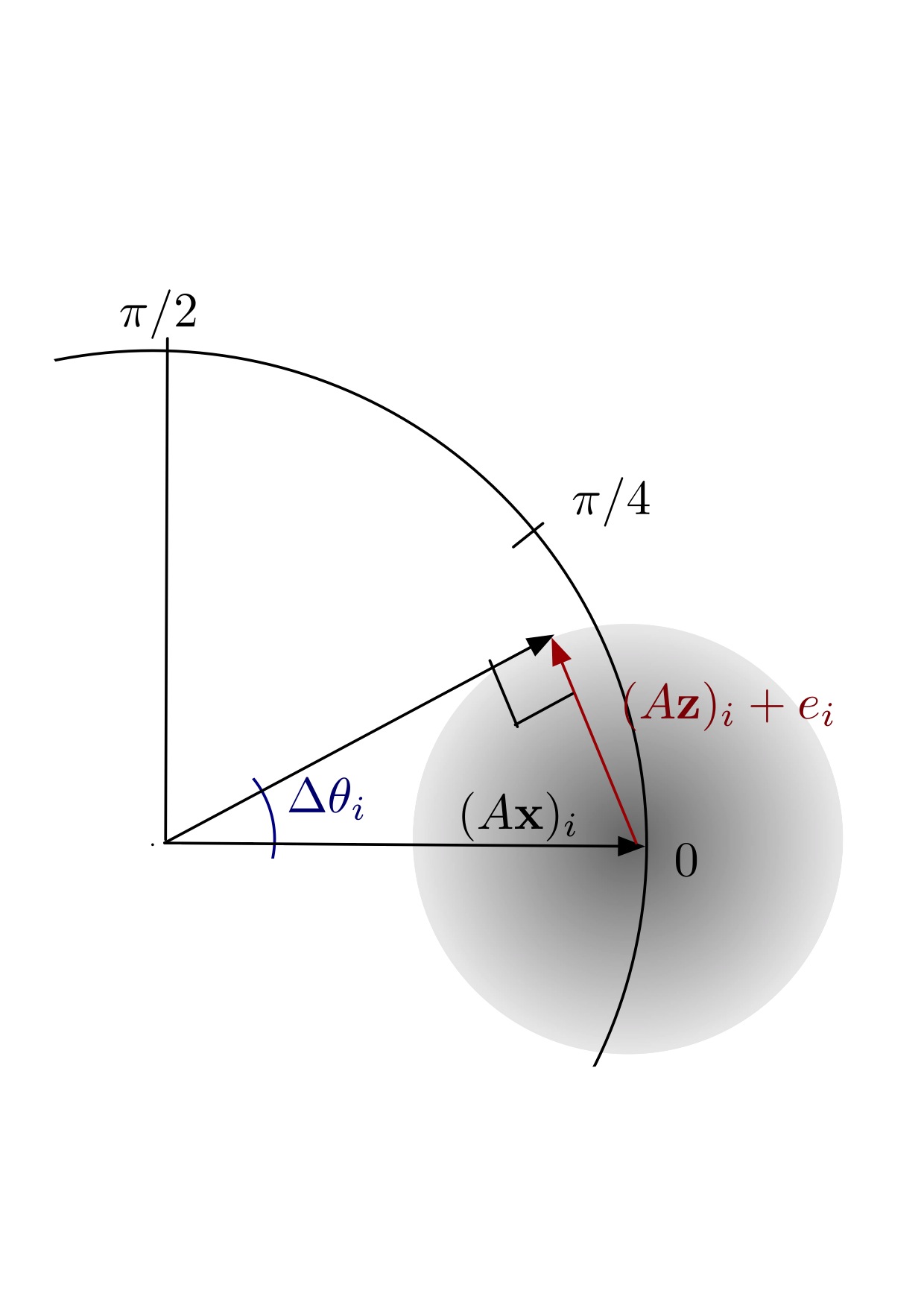}} \qquad\subfloat[Maximum magnitude displacement takes place when the contribution due to noise is aligned with $(\bA\bx)_\yi$ ]{\label{fig:magnitudebound}\includegraphics[width=0.3\textwidth]{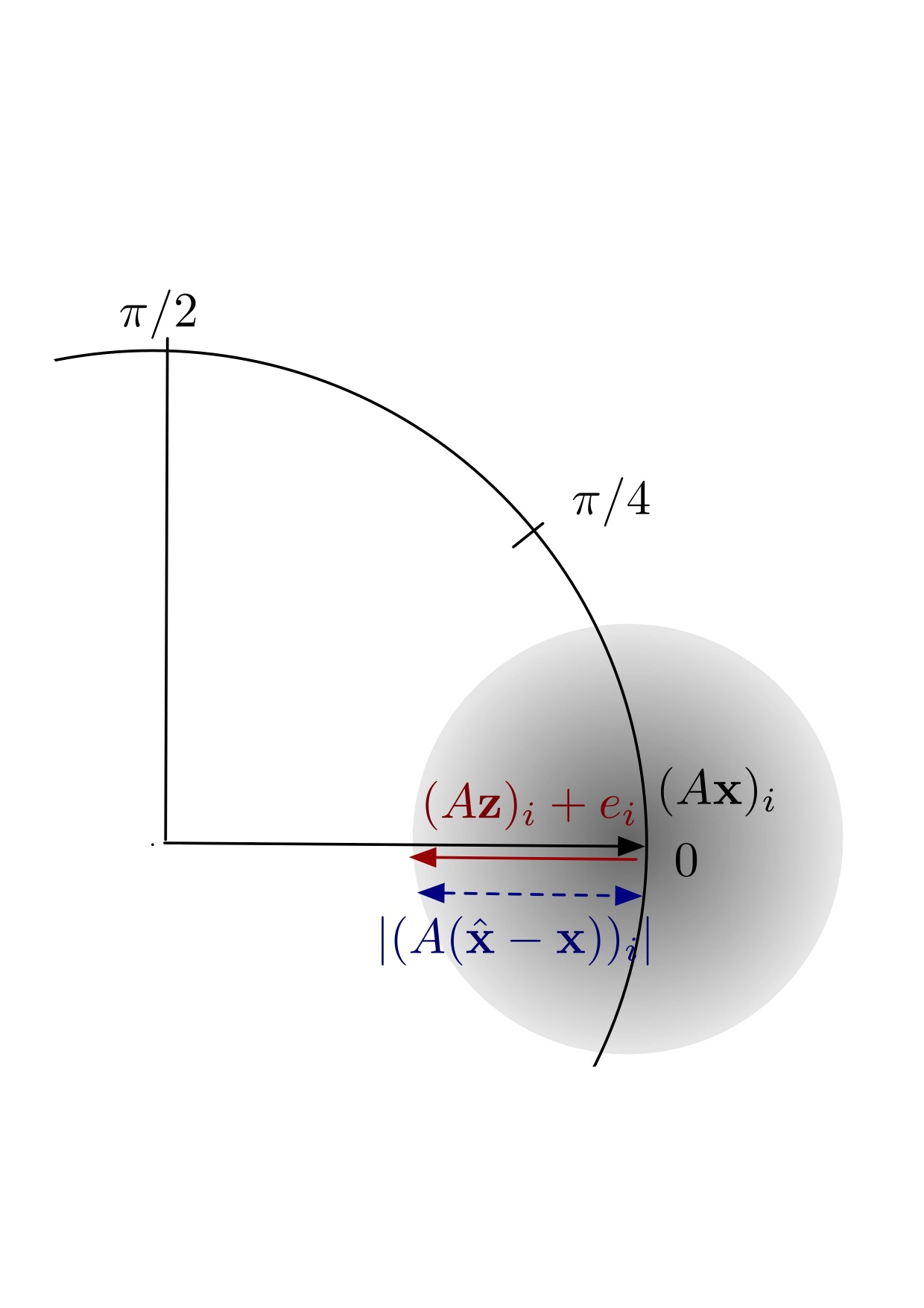}}
  \caption{The effect of noise on a measurement output}
\end{figure}

\begin{figure}
  \centering
  \subfloat[ The decoder \textquotedbl{}randomly\textquotedbl{} picks
$y_{1}$. Since the phase of $y_{1}^{(I,1)}$ is between $-\pi/4$
and $\pi/4$, the decoder can distinguish that the first bit of non-zero
location is $0$ since the decoder can tolerate at most $\pi/4$ phase
displacement for $y_{1}$. So, the non-zero entry is one of $x_{1}$,
$x_{2}$, $x_{3}$, $x_{4}$.]{\label{fig:firstdigit}\includegraphics[width=0.3\textwidth]{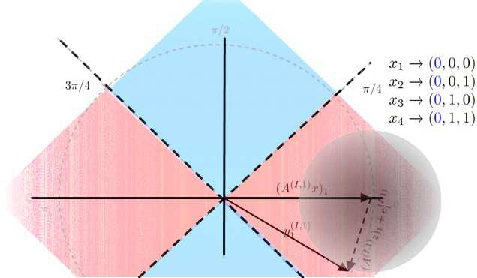}} \qquad\subfloat[The decoder \textquotedbl{}randomly\textquotedbl{} picks
$y_{1}$ again. Since the phase of $y_{1}^{(I,2)}$ is between $3\pi/4$
and $5\pi/4$, the decoder can distinguish that the second bit of
non-zero location is $1$ since the decoder can tolerate at most $\pi/4$
phase displacement for $y_{1}$. So, the non-zero entry is one of
$x_{3}$, $x_{4}$, $x_{7}$, $x_{8}$. Combing the output in the
first phase measurement, we conclude that the non-zero entry is one
of $x_{3}$ and $x_{4}$.]{\label{fig:seconddigit}\includegraphics[width=0.3\textwidth]{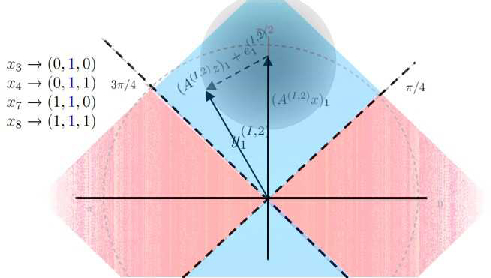}}\qquad
  \subfloat[The decoder \textquotedbl{}randomly\textquotedbl{} picks
$y_{1}$ again. Since the phase of $y_{1}^{(I,3)}$ is between $-\pi/4$
and $\pi/4$, the decoder can distinguish that the third bit of non-zero
location is $0$ since the decoder can tolerate at most $\pi/4$ phase
displacement for $y_{1}$. So, the non-zero entry is one of $x_{1}$,
$x_{3}$, $x_{5}$, $x_{7}$. Combing the outputs in the first and
second phase measurement, we conclude that the non-zero entry is $x_{3}$.]{\label{fig:thirdddigit}\includegraphics[width=0.3\textwidth]{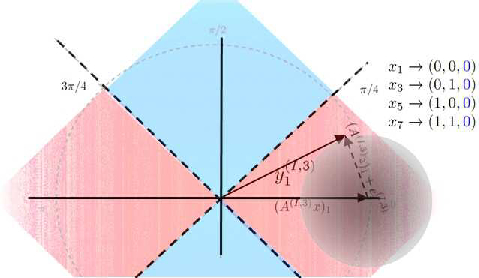}}
  \caption{If we were to distinguish each $\xj$ from $1$ to $8$ by a different phase, the decoder can tolerate at most $\pi/14$ phase displacement for any output $y_\yi$. Instead, we first represent each $\xj=1,2,\ldots,8$ by a three-length binary vector. Next, we perform three sets of phase assignments -- one for each digit. It is easily seen that by allowing multiple measurements, the noise tolerance for the decoder increases.}\label{fig:repeated}
\end{figure}

\subsection{Measurement Design}
As in the exactly $\spar$-sparse case, we start with a randomly drawn left regular bipartite graph $\graph$ with $\samp$ nodes on the left and $\meas'$  nodes on the right.

{\underline{Measurement matrix}:} The measurement matrix $\bA\in\bbc^{2\meas'\steps\times\samp}$ is chosen based on the graph $\graph$. The rows of $\bA$ are partitioned into $\meas'$ groups, with each group consisting of  $2\steps$ consecutive rows. The $\xj$-th entries of the rows $2(\yi-1)\steps+1,(\yi-1)\steps+2,\ldots,2\yi\steps$ are denoted by $a^{(I,1)}_{\yi,\xj},a^{(I,2)}_{\yi,\xj},\ldots,a^{(I,\steps)}_{\yi,\xj},a^{(V,1)}_{\yi,\xj},a^{(V,2)}_{\yi,\xj},\ldots,a^{(V,\steps)}_{\yi,\xj}$ respectively. In the above notation, $I$ and $V$ are used to refer to identification and verification measurements.

For ease of notation, for each $\step=1,2,\ldots,\steps$, we use $\bA^{(I,\step)}$ (resp. $\bA^{(V,\step)}$) to denote the sub-matrix of $\bA$ whose $(\yi,\xj)$-th entry is $a^{(I,\step)}_{\yi,\xj}$ (resp. $a^{(V,\step)}_{\yi,\xj}$).

We define the $\step$-th identification matrix $\bA^{(I,\step)}$ as follows. For each $(\yi,\xj)$, if the graph $\graph$ does not have an edge connecting  $\yi$ on the right to $\xj$ on the left, then $a^{(I,\step)}_{\yi,\xj}=0$. Otherwise, we set $a^{(I,\step)}_{\yi,\xj}$ to be the unit-norm complex number $$a^{(I,\step)}_{\yi,\xj}=e^{\iota\digit_\step(\xj)\pi/2(|\bbg|-1|)}.$$
Note here that the construction for the exactly $\spar$-sparse case can be recovered by setting $\steps=1$, which results in $\bbg=\{1,2,\ldots,\samp\}$ and $\digit_\step(\xj)=\xj$.

Next, we define the $\step$-th verification matrix $\bA^{(V,\step)}$ in a way similar to how we defined the verification entries in the exactly $\spar$-sparse case. For each $(\yi,\xj)$, if the graph $\graph$ does not have an edge connecting  $\yi$ on the right to $\xj$ on the left, then $a^{(V,\step)}_{\yi,\xj}=0$. Otherwise, we set $$a^{(V,\step)}_{\yi,\xj}=e^{\iota\theta^{(V,\step)}_{\yi\xj}},$$ where $\theta^{(V,\step)}_{\yi,\xj}$ is drawn uniformly at random from $\{0,\pi/2(|\bbg|-1),\pi/(|\bbg|-1),3\pi/2(|\bbg|-1)\ldots,\pi/2\}$.

Given an signal vector $\bx$, signal noise $\bz$, and measurement noise $\be$, the measurement operation produces a measurement vector $\by=\bA(\bx+\be)$. Since $\bA$ can be partitioned into $\steps$ identification and $\steps$ verification rows, we think of the measurement vector $\by$ as a collection of outcomes from $\steps$ successive measurement operations such that $$\by^{(I,\step)}=A^{(I,\step)}(\bx+\bz)+\be^{(I,\step)}$$ and $$\by^{(V,\step)}=A^{(V,\step)}(\bx+\bz)+\be^{(V,\step)}$$  are the outcomes from the $\step$-th measurement and $\by=((\by^{(I,\step)},\by^{(V,\step)}):1\leq\step\leq\steps)$.

\subsection{Reconstruction}\label{app:recons}
The decoding algorithm for this case extends the decoding algorithm presented earlier for the exactly $\spar$-sparse case by including the ideas presented in Section~\ref{sec:newideas}. The total number of iterations for our algorithm is upper bounded by $4\spar$.
\begin{enumerate}
\item We initialize by setting the {\it signal estimate vector} $\hat{\bx}{(1)}$ to the all-zeros vector $0^{\samp}$, and for each $\step=1,2,\ldots,\steps$, we set the {\it residual measurement identification/verification vectors} $\tilde{\by}^{(I,\step)}(1)$ and $\tilde{\by}^{(V,\step)}(1)$ to the decoder's observations $\by^{(I,\step)}$ and $\by^{(V,\step)}$.

Let $\cD{(1)}$, the initial neighborly set, be the set of indices  $\yi$ for which, at which the magnitude corresponding to all verification and identification vectors is greater than $\delta/\spar$, {\em i.e.},

$$\cD{(1)}=\bigcap_{\step=1}^\steps\left\{\yi:|y^{(I,\step)}_{\yi}|>\frac{\delta}{\spar},\ |y^{(V,\step)}_{\yi}|>\frac{\delta}{\spar}\right\}$$vector $\by^{(V)}$, {\it i.e.}, the set $ \{\yi\leq \meas:\tilde{y}_{\yi}^{(V)}{(1)}\neq 0\}$. This step takes $\cO(\spar)$ steps, since merely reading $\by$ to check for the zero locations of $\by^{(V)}$ takes that long.

\item The $\iter^{th}$ decoding iteration accepts as its input the $\iter^{th}$ signal estimate vector $\hat{\bx}^\ofiter$, the $\iter^{th}$ neighbourly set $\cD{\ofiter}$, and the $\iter^{th}$ residual measurement identification/verification vectors $\left((\tilde{\by}^{(I,\step)}\ofiter,\tilde{\by}^{(V,\step)}\ofiter):\step=1,2,\ldots,\steps\right)$. In $\cO(1)$ steps it outputs the $(\iter+1)^{th}$ signal estimate vector $\hat{\bx}^\ofiterplusone$, the $(\iter+1)^{th}$ neighbourly set $\cD{\ofiterplusone}$, and the $(\iter+1)^{th}$ residual measurement identification/verification vectors$\left((\tilde{\by}^{(I,\step)}\ofiterplusone,\tilde{\by}^{(V,\step)}\ofiterplusone):\step=1,2,\ldots,\steps\right)$ after the performing the following steps sequentially (each of which takes at most a constant number of atomic steps).

\item {\em Pick a random $\yi\ofiter$:} The decoder picks $\yi\ofiter$ uniformly at random from $\cD\ofiter$
\item {\em Compute quantized phases:} For each $\step=1,2,\ldots,\steps$, compute the {\em current identification angles}, $\hat{\theta}^{(I,\step)}_\iter$, and {\em current identification angles}, $\hat{\theta}^{(I,\step)}_\iter$ defined as follows:
\begin{eqnarray*}
\hat{\theta}^{(I,\step)}_\iter &=& \left[\frac{2(|\bbg|-1|)\left(\angle y^{(I,\gamma)}_{\yi\ofiter}(\mbox{mod }\pi)\right)}{\pi}\right]\frac{\pi}{2(|\bbg|-1|)},\\
\hat{\theta}^{(V,\step)}_\iter &=& \left[\frac{2(|\bbg|-1|)\left(\angle y^{(V,\gamma)}_{\yi\ofiter}(\mbox{mod }\pi)\right)}{\pi}\right]\frac{\pi}{2(|\bbg|-1|)}.\end{eqnarray*}
In the above, $[\cdot]$ denotes the closest integer function. Since there are $\Theta(\samp)$ different phase vectors, to perform this computation, $\cO(\log\samp)$ precision and $\cO(1)$ steps suffice.

For each $\step=1,2,\ldots,\steps$, let $\hat{\digit}_{\step}^\ofiter=2(|\bbg|-1|)\hat{\theta}^{(I,\step)}/\pi$ be the {\em current estimate of $\step$-th digit} and let $\xj\ofiter$ be the number whose representation in $\bbg$ is $(\hat{\digit}_{1}^\ofiter,\hat{\digit}_{2}^\ofiter,\ldots,\hat{\digit}_{\steps}^\ofiter)$.
\item{\em Check if the current identification and verification angles correspond to a valid and unique \xj:} This step determines if $\yi\ofiter$ is a leaf node for $\cS_\delta(\bx-\hat{\bx}\ofiter)$. This operation is similar to the corresponding exact-$\spar$ case. The main difference here is that we perform the verification operation on each of the $\steps$ measurements separately and declare $\yi\ofiter$ as a leaf node only if it passes all the verification tests. The verification step for the $\step$-th measurement is given by the test:
$$\hat{\theta}^{(V,\step)}_\iter\overset{?}{=}\theta_{\yi\ofiter,\xj\ofiter}^{(V,\step)}.$$ If the above test succeeds for every $\step=1,2,\ldots,\steps$, we set $\Delta x\ofiter$ to $|\tilde{y}^{(I,\gamma)}_{\yi\ofiter}(\iter)|$ if $\angle y^{(I,\gamma)}_{\yi\ofiter}\in(-\pi/4,3\pi/4]$, and $-|\tilde{y}^{(I,\gamma)}_{\yi\ofiter}(\iter)|$ if  $\angle y^{(I,\gamma)}_{\yi\ofiter}\in(3\pi/4,7\pi/4]$
Otherwise, we set $\Delta x\ofiter=0$. This step requires at most $\steps$ verification steps and therefore, can be completed in $\cO(1)$ steps.
\item{\em Update $\hat{\bx}\ofiterplusone$, $\tilde{\by}\ofiterplusone$, and $\cD\ofiterplusone$:} If the verification tests in the previous steps failed, there are no updates to be done, i.e., set $\hat{\bx}\ofiterplusone=\hat{\bx}\ofiter$, $\tilde{\by}\ofiterplusone=\tilde{\by}\ofiter$, and $\cD\ofiterplusone=\cD\ofiter$.

Otherwise, we first update the current signal estimate to $\hat{\bx}\ofiterplusone$ by setting the $\xj\ofiter$-th coordinate to $\Delta x\ofiter$. Next, let $\yi_1,\yi_2,\yi_3$ be the possible neighbours of $\xj\ofiter$. We compute the {\em residual identification/verification vectors $\tilde{\by}\ofiterplusone$} at $\yi_1,\yi_2,\yi_3$ by subtracting the weight due to $\Delta x\ofiter$ at each of them. Finally, we update the neighbourly set by removing $\yi_1,\yi_2$, and $\yi_3$ from $\cD\ofiter$ to obtain $\cD\ofiterplusone$.
\end{enumerate}
The decoding algorithm terminates after the $T$-th iteration, where  $T=\min\{4\spar,\{\iter:\cD{(\iter+1)}=\phi\}\}$.
\subsection{Improving performance guarantees of SHO-FA via Set-Query Algorithm of~\cite{Pri:11}}
In \cite{Pri:11}, Price considers a related problem called the {\em Set-Query problem}. In this setup, we are given an unknown signal vector $\bx$ and the objective is to design a measurement matrix $\bA$ such that given $\by=\bA\bx+\be$ (here, $\be$ is an arbitrary ``noise'' vector), and a desired {\em query set} $\cS\subseteq\{1,2,\ldots,\samp\}$, the decoder outputs a reconstruction $\hat{\bx}$ with having support $\cS$ such that $\hat{\bx}-\bx_\cS$ is ``small''. The following Theorem from~\cite{Pri:11} states the performance guarantees for a randomised construction of $\bA$. 
\begin{theorem}[Theorem 3.1 of~\cite{Pri:11}]\label{thm:setquery}
  For every $\epsilon>0$, there is a randomized sparse binary sketch matrix $A$ and recovery
  algorithm $\mathscr{A}$, such that for any $\bx \in \bbR^\samp$, $\cS
  \subseteq\{1,2,\ldots,\samp\}$ with $|\cS| = \spar$, $\hat{\bx} = \mathscr{A}(\bA\bx + \be, \cS)
  \in \bbR^\samp$ has support $\cS(\hat{\bx}) \subseteq \cS$ and
  \[
  ||\hat{\bx} - \bx_\cS||_l \leq (1+\epsilon)(||\bx - \bx_\cS||_l + ||\be||_l)
  \]
for each $l\in\{1,2\}$  with probability at least $1 - 1/\spar$.  $\bA$ has
  $\cO(\spar)$ rows and $\mathscr{A}$ runs in $\cO(\spar)$ time.
\end{theorem}

We argue that the above design may be used in conjunction with our SHO-FA algorithm from Theorem~~\ref{thm:approximate} to give stronger reconstruction guarantee than Theorem~\ref{thm:approximate}. In fact, this allows us to even prove a stronger reconstruction guarantee of $\ell_2<\ell_2$ form. The following theorem makes this precise.
\begin{theorem}\label{thm:approximate2} Let $\spar=\cO(\samp^{1-\Delta})$ for some $\Delta>0$ and let $\epsilon>0$. There exists a reconstruction algorithm {SHO-FA-NO} for $\bA\in\bbc^{\meas\times\samp}$ such that
\begin{enumerate}[($i$)]
\item $\meas=\cnst\spar$
\item {SHO-FA-NO} consists of at most $5\spar$ iterations, each involving a constant number of arithmetic operations with a precision of $\cO(\log{\samp})$ bits.
\item With probability $1-o(1/\spar)$ over the design of $\bA$ and randomness in $\be$ and $\bz$, and for each $l\in\{1,2\}$,
$$||\hat{\bx}-\bx||_l\leq (1+\epsilon)\left(||\bz||_l+||\be||_l\right).$$\
\end{enumerate}
\end{theorem}
\begin{proof} We first note that the measurement matrix proposed in~\cite{Pri:11} is independent of the query set $\cS$ and depends only on the size $\spar$ of the set $\cS$. We design our measurement matrix $\bA\in\bbR^{\meas\times\samp}$ by combining the measurement matrices from Theorem~\ref{thm:approximate} and Theorem~\ref{thm:setquery}  as follows. Let $\bA_1\in\bbR^{\meas_1\times\samp}$ be drawn according to Theorem~\ref{thm:approximate} and $\bA_2\in\bbR^{\meas_2\times\samp}$ be drawn according to Theorem~\ref{thm:setquery} for some $\meas_1$ and $\meas_2$ scaling as $\cO(\spar)$ so as to achieve an error probability $\cO(1/\spar)$. Let $\bA\in\bbR^{\meas\times\samp}$ with $\meas=\meas_1+\meas_2$ with the upper $\meas_1$ rows consisting of all rows of $\bA_1$ and the lower $\meas_2$ rows consisting of all the row of $\bA_2$.

To perform the decoding, the decoder first produces a coarse reconstruction $\hat{\tilde{\bx}}$ by passing the first $\meas_1$ rows of the measurement output $\by$ through the SHO-FA decoding algorithm. Let $\delta$ be the truncation threshold for the decoder. Next, the decoder computes $\cS_{\delta}(\bx)$ to be the support of $\hat{\tilde{\bx}}$. Finally, the decoder we apply the set query algorithm from~\ref{thm:setquery} with inputs $\left([y_{\meas_1+2},\ldots,y_{\meas_2}]^T,\cS_\delta(\bx)\right)$ to obtain a reconstruction $\hat{\bx}$ that satisfies the desired reconstruction criteria.

\end{proof}

\section{SHO-FA with Integer-valued measurement matrices (SHO-FA-INT)}
\label{sec:shofaint}
In the measurement designs presented in Sections~\ref{sec:noiseless} and~\ref{sec:noisy}, a key requirement is that the entries of the measurement matrix may be chosen to be arbitrary real or complex numbers of magnitude (upto $\cO(\log\samp)$ bits of precision). However, in several scenarios of interest, the entries of the measurement matrix are constrained by the measurement process.\begin{itemize}

\item In network tomography~\cite{XuMT:11}, one attempts to infer individual link delays by sending test packets along pre-assigned paths. In this case, the overall end-to-end path delay is the sum of the individual path delays along that path, corresponding to a measurement matrix with only $0$s and $1$s (or in general, with positive integers, if loops are allowed). 

\item If transmitters in a wireless setting are constrained to transmit symbols from a fixed constellation, then the entries of the measurement matrix can only be chosen from a finite ensemble. 

\end{itemize}
In both these examples, the entries of the measurement matrix are tightly constrained. 

In this section, we discuss how key ideas from Sections~\ref{sec:noiseless} and~\ref{sec:noisy}  may be applied in compressive sensing problems where the entries of measurement matrix is constrained to take values form a discrete set. For simplicity, we assume that the entries of the matrix $\bA$ can take values in the set $\set{1,2,\ldots,M}$ for some integer $M\in\bbn^+$. For simplicity, we consider only the exact $\spar$-sparse problem, noting that extensions to the approximate $\spar$-sparse case follow from techniques similar to those used in Section~\ref{sec:noisy}.

\subsection{Measurement Design}
As in the exactly $\spar$-sparse and approximately sparse cases,
the ${\it {measurement}}$ ${\it {matrix}}$ $A$ is chosen based
on the structure of $\graph$ with $\samp$ nodes on the left and
$\meas'$ nodes on the right. Each left node has degree equal $\dgr$.
We denote the edge set by $\cE$.

We give a combinatorial construction of measurement matrices that ensure the equivalent of Property~\ref{prop:2} in this setting.  Let $\zeta(\cdot)$ denote the Riemann-zeta function and let $\eL\leq\log{(\dgr\samp/2)}/\log M$
be an integer. We design a $2\eL\meas'\times\samp$ measurement matrix
$A$ as follows. First, we partition the rows of $\bA$ into $\meas$'
groups of rows, each consisting of $2\eL$ consecutive rows as follows.
%
$$A= \left[\begin{array}{c:c:c:c}
a^{(1)}_{11}&a^{(1)}_{12}&\dots & a^{(1)}_{1\samp}\\
\vdots &\vdots &\ddots&\vdots\\
a^{(\eL)}_{11}&a^{(\eL)}_{12}&\dots & a^{(\eL)}_{2\samp}\\
\hdashline
a^{(1)}_{21}&a^{(1)}_{22}&\dots & a^{(1)}_{2\samp}\\
\vdots &\vdots &\ddots&\vdots\\
a^{(\eL)}_{21}&a^{(\eL)}_{22}&\dots & a^{(\eL)}_{2\samp}\\
\hdashline
\vdots &\vdots &\ddots&\vdots
\end{array}\right]$$
Let the $\el$-th and $\el+\eL$-th rows in the $\xj$-th column of
the $\yi$-th group of $\bA$ are respectively denoted $a_{\yi\xj}^{(I,\el)}$
and $a_{\yi\xj}^{(V,\el)}$ . Let $a_{\yi\xj}=[a_{\yi\xj}^{(I)}a_{\yi\xj}^{(V)}]=[a_{\yi\xj}^{(I,1)}a_{\yi\xj}^{(I,2)}\ldots a_{\yi\xj}^{(V,\eL-1)}a_{\yi\xj}^{(V,\eL)}]^{T}$.
First, for each $(\yi,\xj)\notin\cE$, we set $a_{\yi\xj}=0$. Next,
set the non-zero entries of $\bA$ by picking the vectors $(a_{\yi\xj}^{(V)}:(\yi,\xj)\in\cE)$
by uniformly sampling without replacement from the set
\[
\mathcal{C}\triangleq\left\{ [c_{1},c_{2},\ldots,c_{\eL}]^{T}\in[M]^{\eL}:\mbox{gcd}(c_{1},c_{2},\ldots,c_{\eL})=1\right\} .
\]
 Lemma~\ref{lem:coprimes} shows that, $M^{\eL}/2\leq|\mathcal{C}|\leq M^{\eL}$
(here, $\zeta(\cdot)$ is the Reimann zeta function). Therefore, it
suffices that $|\cE|\leq M^{\eL}/\zeta(\eL)$ for such a sampling
to be possible. Further, noting that $|\cE|=\dgr\samp$, $\zeta(R)\leq2$ (via standard bounds),
and $R\leq\log{(dn/2)}/\log{M}$ (by assumption), this condition is satisfied.

\begin{lemma}\label{lem:coprimes} For $M$ large enough, $M^{\eL}/2\leq|\mathcal{C}|\leq M^{\eL}$.
\end{lemma} \begin{proof} The upper bound on $|\mathcal{C}|$ is
trivial, since each element of ${\cal C}$ is an $\eL$ length vector
whose each coordinate takes values from the set $[M]$. To prove the
lower bound, note that a sufficient condition for $\mbox{gcd}(c_{1},c_{2},\ldots,c_{\eL})$
to be true is that for each prime number $p\in\bbn$, there exists
at least one index $\el\in[\eL]$ such that $p$ does not divide $c_{\el}$.
Since the number of vectors $(c_{1},c_{2},\ldots,c_{\eL})\in[M]^{\eL}$
such that for each $\el$, $c_{\el}$ is divisible by $p$ is at most
$(M/p)^{\eL}$, the number of vectors in $[M]^{\eL}$ such that at
least one component is not divisible by $p$ is at least $M^{\eL}(1-p^{-\eL})$.
Denoting the set of prime numbers by $\mathbb{P}$ and extending the
above argument to exclude all vectors that are divisible by some prime
number greater than or equal to two, we obtain, for $M$ large enough,
\[
|{\cal C}|\geq M^{\eL}\prod_{p:p\in\mathbb{P}}(1-1/p^{-\eL})=M^{\eL}/\zeta(\eL).
\]
In the above, the second equality follows from Euler's product formula
for the Reimann zeta function. \end{proof}

Pick the vectors $(a_{\yi\xj}^{(I)}:(\yi,\xj)\in\cE)$ in
the following way. First, calculate the {}``normalized'' version
of every vector $c$ in ${\cal C}$ by set $nor(c)=c/c_{1}$. Then,
rearrange the {}``normalized'' vectors in {}``lexicographical''
order. Finally, assign the original value of $j$-th {}``normalized'' vector to $a_{\yi\xj}^{(I)}$
.

The output of the measurement is a $2\eL\meas'$-length vector
$\by=\bA\bx$. Again, we partition $\by$ into $\meas'$ groups
of $2\eL$ consecutive rows each, and denote the $\yi$-th sub-vector
as $y_{\yi}=[y_{\yi}^{(I)T},y_{\yi}^{(V)T}]^{T}$. 
\subsection{Reconstruction}
The decoding algorithm is conceptually similar to the decoding algorithm presented in Section~\ref{sec:noiselessreconstruction}. The decoder first generates a list of leaf nodes. Next, it proceeds iteratively by decoding the input value corresponding to one leaf node, and updating the list of leaf nodes by subtracting the contribution of the last decoded signal coordinate from the measurement vector.

Our algorithm proceeds iteratively,
and has at most $\spar$ overall number of iterations, with $\iter$
being the variable indexing the iteration number.
\begin{enumerate}
\item \underline{\it Initialization}: We initialize by setting the \textit{signal
estimate vector} $\hat{\bx}{(1)}$ to the all-zeros vector $0^{\samp}$,
and the \textit{residual measurement identification/verification vectors}
$\tilde{\by}^{(I)}(1)$ and $\tilde{\by}^{(V)}(1)$ to the decoder's
observations $\by^{(I)}$ and $\by^{(V)}$.
\item \underline{\it Leaf-Node List}: Let $\mathcal{L}(1)$, the initial
$\cSx$-leaf node set, be the set of indices $\yi$ which are $\cSx$-leaf
nodes. We make the list in the following steps:

\begin{enumerate}
\item {\em Compute {}``normalized'' vectors $nor^{(I)}(\yi)$ and $nor^{(V)}(\yi)$:}\label{step:computephase-int}
Let the\textit{ {}``normalized'' identification and verification
vectors} be defined respectively for index $i$ (starting from 1),
as follows:
\begin{eqnarray*}
nor^{(I)}(\yi) & \triangleq & \tilde{y}_{\yi}^{(I)}(1)/\tilde{y}_{\yi}^{(I,1)}(1),\\
nor^{(V)}(\yi) & \triangleq & \tilde{y}_{\yi}^{(V)}(1)/\tilde{y}_{\yi}^{(V,1)}(1).
\end{eqnarray*}

\item {\em Check if the current {}``normalized'' identification and
verification vectors correspond to a valid and unique $x_{\xj}$:}
For this, we check at most two things.

\begin{enumerate}
\item First, we check if $nor^{(I)}(\yi)$ corresponds $\xj$-th {}``normalized''
$c$, and the corresponding $\xj^{th}$ column of the $\yi^{th}$
group is non-zero. If so, we have {}``tentatively identified{''}
that the $\yi^{th}$ component of $\tilde{\by}$ is a leaf-node of
the currently unidentified non-zero components of $\bx$, and in particular
is connected to the $\xj^{th}$ node on the left, and the algorithm
proceeds to the next step below. If not, we simply increment $\yi$
by $1$ and return to Step (\ref{step:computephase-int}).
\item Next, we verify our estimate from the previous step. If $nor(a_{\yi,\xj}^{(V)})=nor(\yi)$,
the verification test passes, and include $\yi$ in $\mathcal{L}(1)$.
If not, we simply increment $\yi$ by $1$ and return to Step (\ref{step:computephase-int}).
\end{enumerate}
\end{enumerate}
\item \underline{\it Operations in $\iter^{th}$ iteration:}{The $\iter^{th}$
decoding iteration accepts as its input the $\iter^{th}$ signal estimate
vector {$\hat{\bx}\ofiter$}, the $\iter^{th}$ leaf node set $\mathcal{L}(\iter)$,
and the $\iter^{th}$ residual measurement identification/verification
vectors $(\tilde{\by}^{(I)}\ofiter,\tilde{\by}^{(V)}\ofiter)$. In
$\cO(1)$ steps it outputs the $(\iter+1)^{th}$ signal estimate vector
{$\hat{\bx}\ofiterplusone$}, the $(\iter+1^{th}$ leaf node set
$\mathcal{L}(\iter+1)$, and the $(\iter+1)^{th}$ residual measurement
identification/verification vectors $(\tilde{\by}^{(I)}\ofiterplusone,\tilde{\by}^{(V)}\ofiterplusone)$
after the performing the following steps sequentially (each of which
takes at most a constant number of atomic steps):}

\begin{enumerate}
\item {\em Pick a random $\yi\ofiter\in\mathcal{L}(\iter)$:} The decoder
picks an element $\yi\ofiter$ uniformly at random from the $\iter^{th}$
leaf-node list $\mathcal{L}(\iter)$.
\item {\em Compute {}``normalized'' vectors $nor^{(I)}({\iter})$ and
$nor^{(V)}({\iter})$:}\label{step:computephase-2} Let the \textit{current
{}``normalized'' identification and verification vectors} be defined
respectively as the phases of the residual identification and verification
entries being considered in that step, as follows:
\begin{eqnarray*}
nor^{(I)}(\iter) & \triangleq & \tilde{y}_{i\ofiter}^{(I)}\ofiter/\tilde{y}_{i\ofiter}^{(I,1)}\ofiter,\\
nor^{(V)}(\iter) & \triangleq & \tilde{y}_{i\ofiter}^{(V)}\ofiter/\tilde{y}_{i\ofiter}^{(V,1)}\ofiter.
\end{eqnarray*}

\item {\em Locate non-zero entry $\xj$ and derive the value of $\hat{x}_{\xj\ofiter}\ofiter$:}
For this, we do at most two things.

\begin{enumerate}
\item First, we check if $nor^{(I)}(\iter)$ corresponds $\xj\ofiter$-th
{}``normalized'' $c$. We have identified that the $\yi^{th}$ component
of $\tilde{\by}$ is a leaf-node of the currently unidentified non-zero
components of $\bx$, and in particular is connected to the $\xj\ofiter^{th}$
node on the left, and the algorithm proceeds to the next step below.
\item Next, we assign the value, $\tilde{y}_{i\ofiter}^{(I)}\ofiter/a_{\yi\ofiter,\xj\ofiter}^{(I)}=\tilde{y}_{i\ofiter}^{(V)}\ofiter/a_{\yi\ofiter,\xj\ofiter}^{(V)}$,
to $\hat{x}_{\xj\ofiter}\ofiter$ and proceeds the algorithm to the
next step below.
\end{enumerate}
\item \noindent {\em Update $\hat{\bx}\ofiterplusone$, $\mathcal{L}(t+1)$
,$\tilde{\by}^{(I)}\ofiterplusone$, and $\tilde{\by}^{(V)}\ofiterplusone$:}
In particular, at most $7$ components of each of these vectors need
to be updated. Specifically, $\hat{x}_{\xj\ofiter}\ofiterplusone$
equals $\tilde{y}_{i\ofiter}^{(I)}\ofiter/a_{\yi\ofiter,\xj\ofiter}^{(I)}$.
$\yi\ofiter$ is removed from the leaf node set $\mathcal{L}(t)$
and check whether the neighbours of $\hat{x}_{\xj\ofiter}\ofiter$
become leaf node to get the leaf-node list $\mathcal{L}(t+1)$. And
finally values each of $\tilde{\by}^{(I)}\ofiterplusone$
and $\tilde{\by}^{(V)}\ofiterplusone$ are updated from those of $\tilde{\by}^{(I)}\ofiter$
and $\tilde{\by}^{(V)}\ofiter$ (those corresponding to the neighbours
of $\hat{x}_{\xj\ofiter}\ofiter$) by subtracting out $\hat{x}_{\xj\ofiter}\ofiter$
multiplied by the appropriate coefficients of $\bA$.
\end{enumerate}
\item {\underline{Termination:}} The algorithm stops when the leaf node
set is empty, and outputs the last $\hat{\bx}(\iter)$.
\end{enumerate}

\subsection{Correctness}
We show that the algorithm presented above correctly reconstructs the vector $\hbx$ with a high probability over the random selection of graph $\graph$ and the random choice of $\bA$ provided that the number of measurements is $\Omega(\spar\log{\samp}/\log{M})$
\begin{theorem}\label{thm:shofa-int}Let $\graph$ and $\bA$ be determined as above. Then, given $\by=\bA\bx$, SHO-FA-INT outputs a reconstruction $\hbx$ such that $\Pr_{\graph,\bA}(\hbx\neq\bx)=\cO(1/\spar)$. Further, SHO-FA-INT stops in at most $\spar$ iterations and the overall complexity of decoding  is $\cO(\spar\log{\samp})$.
\end{theorem}
\begin{proof}
By definition, the probability that $S(\bx)$ and all its subsets expand by at least $(1-\epsilon)$ is at least $1-\alpha$. Therefore, to show that the overall error probability is upper bounded by $2\alpha$, it suffices to show that conditioned on the even that $S(\bx)$ and all its subsets expand, the error probability is at most $\alpha$. In particular, noting that  each iteration that results in a successful leaf identification decreases the number of undecoded non-zero values by $1$, it suffices to show the following:\\
\noindent\phantom{a} 1) In each iteration $\iter$, $\yi(\iter)$ is correctly identified as a leaf or non-leaf with probability $1-o(1/\spar)$.\\
\noindent\phantom{a} 2) The probability of picking a leaf in each iteration is lower bounded by a constant.\\
Since we choose the measurement weights $a_{\yi e}$'s for each $e\in E$ and $i\in[\meas]$ to be distinct elements from $\cal C$ by sampling uniformly without replacement, when a node $\yi(\iter)$ is a leaf node, its parent node $e$ is correctly identified since the measurement vector at node $\yi(\iter)$  equals $x_ea_{\yi(\iter)e}$. On the other hand, when $\yi(\iter)$ is not a leaf node, an error might occur if for the specific vector $\bx$, the measurement output at $y_\yi$ is proportional to the measurement weight for some $e''$ connected to $\yi$, $\it{i.e.}$, $$\sum_{e'\in N(\yi)}x_{e'}a_{\yi(\iter)e'}=x'' a_{\yi(\iter)e''}$$ for some $d''$. Since all the measurement weights  are chosen randomly, by the Schwartz-Zippel lemma~\cite{Sch:80,Zip:79}, the probability of this event is $\cO(1/\samp)$, which is $o(1/\spar)$.

\subsection{Decoding complexity}
Since there are at most $\spar$ non-zero $x_e$'s, the algorithm terminates in $\spar$ steps.
Finally, to compute the decoding complexity, note that all arithmetic operations are over vectors in $[M]^\eL$. Therefore, each such operation can be done in $\cO(\eL\log M)$ time. Note that $ M^\eL=\cO(\samp)$ by our design choice. Further, the total number of update operations is upper bounded by the number of edges in $\cE(S(\bx))$, which is $\cO(\meas)$ by the expansion property.  Thus, the complexity required by the decoder is $\cO(\spar\log \samp)$ or $\cO(\spar(\log \samp+\eL P))$ if $\hat \bx$ equal $\bx$ upto $P$ bits of precision.
\end{proof}

\section{Conclusion}
In this work we present a suite of algorithms (that we call SHO-FA) for compressive sensing that require an information-theoretically order-optimal number of measurements, bits over all measurements, and encoding, update, and decoding time-complexities. As a bonus, with non-zero probability it can also handle ``data-base queries''. The algorithms are robust to noisy signal tails and noisy measurements. The algorithms are ``practical'' (all constant factors involved are ``small''), as validated by both our analysis, and simulations. Our algorithms can reconstruct signals that are sparse in any basis that is known {\it a priori} to both the encoder and decoder, and work for ``many'' ensembles of ``sparse'' measurement matrices.

\bibliographystyle{IEEEtran}
\bibliography{IEEEabrv}

\appendix

\subsection{Proof of Lemma~\ref{lem:expansion}}\label{apx:proof_lem_expansion}

\begin{proof} It suffices to prove the desired property for all $\cSx$ of size exactly $\spar$. Let $\cSpx \subseteq \cSx$. Let $\{(s_1,t_1),(s_2,t_2),\ldots,(s_{\dgr|\cSpx|},t_{\dgr|\cSpx|})\}$ be the set of outgoing edges from $\cSpx$. Without loss of generality, we assume these edges are drawn in the following manner.

In the initialization stage, we ``split'' every node on the right of $\graph$ to $\dgr \samp /\cnst \spar$ ``virtual'' nodes\footnote{We assume $\dgr \samp /\cnst \spar$ is an integer, with the understanding that in practice one can always increase $\cnst$ to make $\dgr \samp /\cnst \spar$ integer while the ``fail-to-expand'' probability is still bounded by the desired target $\epsilon$.}. Each virtual node represents a ``true'' node on the right. We maintain a set of ``remaining'' virtual nodes, which we will select and remove virtual nodes from.

To draw the edges, we visit the nodes in $\cSpx$ (on the left of $\graph$) sequently. For each node, we select uniformly at random a set of $\dgr$ distinct virtual nodes from the remaining virtual node set. We form $\dgr$ edges by connecting this node in $\cSpx$ and the true nodes on the right that those $\dgr$ selected virtual nodes represent. After the $\dgr$ edges are formed, we remove the $\dgr$ selected virtual nodes from the remaining virtual node set, and proceed to the next node in $\cSpx$.

In this way, we generate a bipartite graph that is both $\dgr$ left-regular and $\dgr \samp /\cnst \spar$ right-regular; that is, each node on the left has a degree of $\dgr$ and each node on the right has a degree of $\dgr \samp /\cnst \spar$. By using standard arguments of
sequential implementation of random experiments, one can verify that the graph generated in this way is chosen uniformly at random from all bipartite graphs that are both $\dgr$ left-regular and $\dgr \samp /\cnst \spar$ right-regular.

For each $i=1,2,\ldots,\dgr|\cSpx|$, the probability that the edge $(s_i, t_i)$ reaches an ``old'' true node (on the right) that is already reached by those edges generated ahead of $(s_i, t_i)$ is upper bounded as
\begin{eqnarray*}
\Pr_{\graph}(t_i\in\{t_1,\ldots t_{i-1}\})&\leq & \frac{(i-1)}{\cnst \spar}\\
&\leq &\frac{\dgr|\cSpx|}{\cnst \spar}.
\end{eqnarray*}

Let $N(\cSpx)$ be the set of all neighboring nodes of the nodes in $\cSpx$. The size of $N(\cSpx)$ is no more than $2\dgr|\cSpx|/3$ if and only if out of $\dgr |\cSpx|$ edges, there exists a set of at least $\dgr|\cSpx|/3$ edges fail to reach ``new'' nodes (on the right). Exploiting this observation, we have
\begin{eqnarray*}
\lefteqn{\Pr_{\graph}\left(|N(\cSpx)|\leq 2\dgr|\cSpx|/3\right)}\\
&=&\Pr_{\graph}\left(\bigcup_{\substack{\sigma\subseteq\{1,\ldots,\dgr|\cSpx|\}\\|\sigma|\geq |\dgr\cSpx/3|}}\bigcap_{i\in\sigma}\{t_i\in\{t_1,\ldots t_{i-1}\}\}\right)\\
&=&\Pr_{\graph}\left(\bigcup_{\substack{\sigma\subseteq\{1,\ldots,\dgr|\cSpx|\}\\|\sigma|= \dgr|\cSpx|/3}}\bigcup_{\sigma'\supseteq\sigma}\bigcap_{i\in\sigma'}\left\{t_i\in\{t_1,\ldots t_{i-1}\right\}\}\right)\\
&\leq & {\dgr|\cSpx|\choose \dgr|\cSpx|/3}\Big(\frac{\dgr |\cSpx|}{\cnst\spar}\Big)^{\dgr|\cSpx|/3}.
\end{eqnarray*}
Consequently, the probability that there exists one $\cSpx\subseteq \cSx$ so that $|N(\cSpx)|\leq 2\dgr|\cSpx|/3$ can be bounded by
\begin{eqnarray}
\lefteqn{\Pr_{\graph}\large(\cup_{\cSpx\subseteq \cSx}\{|N(\cSpx)|\leq 2|\cSpx|\}\large)}\nonumber\\
&\leq& \sum_{\cSpx\subseteq \cSx} {\dgr|\cSpx|\choose \dgr|\cSpx|/3}\Big(\frac{\dgr |\cSpx|}{\cnst\spar}\Big)^{\dgr|\cSpx|/3}\nonumber\\
&=&\sum_{j=1}^{\spar}{\spar \choose j} {\dgr j\choose \dgr j/3}\Big(\frac{\dgr j}{\cnst\spar}\Big)^{\dgr j/3}\nonumber\\
&\leq & \sum_{j=1}^{\spar}\Big(\frac{\spar e}{j}\Big)^j\big( 3e\big)^{\dgr j/3}\Big(\frac{\dgr j}{\cnst\spar}\Big)^{\dgr j/3}\label{eq:stirling}\\
&=& \sum_{j=1}^\spar \Bigg(\frac{\spar e}{j}\left(\frac{3\dgr j e}{\cnst\spar}\right)^{\dgr/3}\Bigg)^j\nonumber \\
&\leq & \sum_{j=1}^{\lceil\sqrt{\spar}\rceil} \Bigg(\frac{\spar e}{j}\left(\frac{3\dgr j e}{\cnst\spar}\right)^{\dgr/3}\Bigg)^j + \sum_{j=\lfloor\sqrt{\spar}\rfloor}^\spar \Bigg(\frac{\spar e}{j}\left(\frac{3\dgr je}{\cnst\spar}\right)^{\dgr/3}\Bigg)^j\nonumber\\
&\leq & \sqrt{\spar} \Bigg(\frac{\spar e}{\sqrt{\spar}}\left(\frac{3\dgr \sqrt{\spar} e}{\cnst\spar}\right)^{\dgr/3}\Bigg) + \sum_{j=\lfloor\sqrt{\spar}\rfloor}^\infty \Bigg(e\left(\frac{3\dgr e}{\cnst}\right)^{\dgr/3}\Bigg)^j\label{eq:upperboundindividual}\\
&\leq & \left(\frac{3\dgr e}{\cnst}\right)^{\dgr/3}\spar^{-\dgr/6}e+\exp(-\theta(\sqrt{\spar}))\label{eq:order}\\
&&=\cO(\spar^{-\dgr/6}).
\end{eqnarray}
In the above, the inequality in~\eqref{eq:stirling} follows from Stirling's approximation; the upper bound in~\eqref{eq:upperboundindividual} is derived by noting that the first term in the sum takes its maximum when $j=\lfloor\sqrt{\spar}\rfloor$ and the second term is maximum when $j=\spar$;~\eqref{eq:order} is obtained by  noting that the second term is a geometric progression.

Finally, we plug in the choice of $\dgr=7$ to complete the proof.\end{proof}

\subsection{Proof of Lemma~\ref{lem:lowerbound}}\label{apx:proof_lowerbound}
Suppose each set of  of size $\spar$ of $\cSx$ nodes on the left of $\graph$ has strictly more than $d/2$ times as many nodes neighbouring those in $\cSx$, as there are in $\cSx$. Then by standard arguments in the construction of expander codes~\cite{Spi:95}, this implies the existence of a linear code of rate at least $1-\meas/\samp$, and with relative minimum distance at least $\spar/\samp$.\footnote{For the sake of completeness we sketch such an argument here. Given such an expander graph $\graph$, one can construct a $\samp \times \spar$ binary matrix $\bA$ with $1$s in precisely those $(i,j)$ locations where the $i$th node on the left is connected with the $j$th node on the right. Treating this matrix $\bA$ as the parity check matrix of a code over a block-length $\samp$ implies that the rate of the code is at least $\spar/\samp$, since the parity-check matrix imposes at most $\spar$ constraints on the $\samp$ bits of the codewords. Also, the minimum distance is at least $\spar$. Suppose not, {\it i.e.} there exists a codeword in this linear code of weight less than $\spar$. Let the support of this codeword be denoted $\cSx$. Then by the expansion property of $\graph$, there are strictly more than $|\cSx|d/2$ neighbours of $\cSx$. But this implies that there is at least one node, say $v$, neighboring $\cSx$ which has exactly one neighbor in $\cSx$. But then the constraint corresponding to $v$ cannot be satisfied, leading to a contradiction.} But by the {Hamming bound~\cite{Roth}}, it is known that codes of minimum distance $\delta$ can have rate at most $1-H(\delta)$, where $H(.)$ denotes the binary entropy function. Since $\spar = \o(\samp)$, $\delta = \spar/\samp \rightarrow 0$. But in this regime $1-H(\delta) \rightarrow 1-\delta\log(1/\delta)$. Comparing $(\spar/\samp)\log(\samp/\spar)$ with $\meas/\samp$ gives the required result.
\endproof

\subsection{Proof of Lemma~\ref{lem:manyleafs}}\label{apx:proof_lem_manyleafs}
For any set of nodes $S$ in the graph $\graph$, we define $N(S)$ as the set of neighboring nodes of the nodes in $S$. For any set $\cSpx \subseteq \cSx$, we define $\bet$ as the portion of the nodes in $N(\cSpx)$ that are $\cSpx$-leaf nodes.

First, each node $v\in N(\cSpx)$ is of one of the following two types:
\begin{enumerate}
\item It has only one neighboring node in $\cSpx$, on the left of $\graph$. By the definition of $\bet$, the number of nodes in $N(\cSpx)$ of this type is $\bet|N(\cSpx)|$.
\item It has at least two neighboring nodes in $\cSpx$, on the left of $\graph$. The number of nodes in $N(\cSpx)$ of this type is $(1-\bet)|N(\cSpx)|$.
\end{enumerate}

We have two observations. First, since the degree of each node in $\cSpx$ is $\dgr$, the total number of edges from $\cSpx$ to $N(\cSpx)$ is at most $\dgr |\cSpx|$ and the number of nodes in $N(\cSpx)$ is at most $\dgr|\cSpx|$.

Second, the total number of edges entering $N(\cSpx)$ from $\cSpx$ is at least $$\bet|N(\cSpx)|+2(1-\bet)|N(\cSpx)|=(2-\bet)|N(\cSpx)|,$$ as the number of neighboring nodes for the nodes of Type 1 is one and of Type 2 is at least two.

Combining the above two observations, we can get the following inequality:
$$
    (2-\bet)\dgr|N(\cSpx)|/3\leq \dgr|\cSpx|.
$$
According to the setting of the Lemma, we also have
$\left|N(\cSpx)\right| \geq 2\dgr/|\cSpx|3$. Therefore, it follows that $$2(2-\bet)\dgr|\cSpx|/3\leq \dgr|\cSpx|,$$ and consequently $\bet\geq 1/2$.
\endproof
\subsection{Proof of Lemma~\ref{lem:query}}\label{app:query}
Consider the algorithm $\mathcal{A}$ that proceeds as follows. First, among the set of all right nodes that neighbour $\xj$, check if there exists a node $\yi$ such that $y_\yi^{(I)}=y_\yi^{(V)}=0$. If there exists such a node, then output $\hat{x}_\xj=0$. Otherwise, check if there exists a $\cSx$-leaf node among the neighbours of $\xj$. This check can be performed by using verification and identification observations as described for the SHO-FA reconstruction algorithm. If there exists a leaf node, say $\yi$, then output $\hat{x}_\xj=|y_\yi|$. Else, the algorithm terminates without producing any output.

Two see that the above algorithm satisfies the claimed properties, consider the following two cases. \\
\underline{Case 1: $x_\xj=0$.} In this case, $\hat{x}_\xj=0$ is output if at least one neighbour of $\xj$ lies outside $N(\cSx)$. Since $N(\cSx)$ has at most $\dgr\spar$ elements, the probability that a neighbour of $\xj$ lies inside $N(\cSx)$ is at most ${\dgr\spar}/{\cnst\spar}=\dgr/\cnst$. Thus, the probability that none of the neighbours of $\xj$ lie outside $N(\cSx)$ is at least $(1-(\dgr/\cnst)^\dgr)$. The algorithm incorrectly reconstructs $x_\xj$ if all neighbours of $\xj$ lie within $N(\cSx)$ and SHO-FA incoorectly identifies one of these nodes as a leaf node. By the analysis of SHO-FA, this event occurs with probability $o(1/\spar)$.\\
\underline{Case 2: $x_\xj\neq 0$.} For $\mathcal{A}$ to produce the correct output, it has to identify one of the neighbours of $\xj$ as a leaf. The probability that there exists a leaf among the neighbours of $\xj$ is at least $(1-(\dgr/\cnst)^\dgr)$ by an argument similar to the previous case. Similarly, the proabability of erroneous identification is $o(1/\spar)$.\endproof

\subsection{Phase noise}\label{app:phasenoise}
\proofof{Lemma~\ref{lem:phasenoise}}}
 First, we find an upper bound on the maximum possible
phase displacement in $y_{\yi}$ due to fixed noise vectors $\bz$
and $\be$. Let $\Delta\theta_{\yi}$ be the difference in phase between
the \textquotedbl{}noiseless\textquotedbl{} output $(\bA'\bx)_{\yi}$
and the actual output $y_{\yi}=(\bA'(\bx+\bz)+\be)_{\yi}$. Figure~\ref{fig:maxnoise}
shows this geometrically. By a {straightforward} geometric argument,
for fixed $\bz$ and $\be$, the phase displacement $\Delta\theta_{\yi}$
is upper bounded by $\pi|(\bA'\bz)_{\yi}+e_{\yi}|/|(\bA'\bx)_{\yi}|$.
Since $\yi$ is a leaf node for $\cSx$, $|(\bA'\bx)_{\yi}|\geq|\delta/\spar|$.
Therefore,
\[
\Delta\theta_{\yi}\leq\pi|(\bA'\bz)_{\yi}+e_{\yi}|\spar/\delta.
\]

 Since each $z_{\xj}$ is a Gaussian with zero mean and variance $\sigma_{z}^{2}$,
$(\bA'\bz)_{\yi}$ is a Complex Gaussian with zero mean and variance
at most $\samp\sigma_{z}^{2}$. Further, each row of $\bA'$ has at
most $\dgr\samp/\cnst\spar$ non-zero entries. Therefore, $(\bA'\bz)_{\yi}+e_{\yi}$
is a zero mean complex Gaussian with variance at most $(\dgr\samp/\cnst\spar)\sigma_{z}^{2}+\sigma_{e}^{2}$.

The expected value of $\Delta\theta_{\yi}$ is bounded as follows:
\begin{eqnarray*}
\lefteqn{E_{\bz,\be}\large(\Delta\theta_{\yi}\large)}\\
 & \leq & E_{\bz,\be}\large(\pi|(\bA'\bz)_{\yi}+e_{\yi}|\spar/\delta\large)\\
 & \leq & \frac{\pi\spar}{\delta}\int_{0}^{\infty}\sqrt{\frac{2}{\pi(\dgr\samp\sigma_{z}^{2}/\cnst\spar)+\sigma_{e}^{2}}}le^{-l^{2}/2(\dgr\samp\sigma_{z}^{2}/\cnst\spar+\sigma_{e}^{2})}dl\\
 & = & \sqrt{\frac{2\pi\spar^{2}(\dgr\samp\sigma_{z}^{2}/\cnst\spar+\sigma_{e}^{2})}{\delta^{2}}}.
\end{eqnarray*}

Next, note that
\begin{eqnarray*}
{\Pr_{\bz,\be}\big(\Delta\theta_{\yi}>\alpha E_{\bz,\be}(\Delta\theta_i)\big)} & \leq & \Pr_{\bz,\be}\bigg(|(\bA'\bz)_{\yi}+e_{\yi}|\spar/\delta>\alpha E_{\bz,\be}(\Delta\theta_i)\bigg)\\
 & = & \Pr_{\bz,\be}\bigg(|(\bA'\bz)_{\yi}+e_{\yi}|>\alpha E_{\bz,\be}(\Delta\theta_i)\delta/\pi\spar\bigg)\\
 &=&Pr_{\bz,\be}\bigg(|(\bA'\bz)_{\yi}+e_{\yi}|>\alpha\sqrt{\frac{2(\dgr\samp\sigma_{z}^{2}/\cnst\spar+\sigma_{e}^{2})}{\pi}}\bigg).
\end{eqnarray*}

Finally, applying standard bounds on the tail probabilities of Gaussian
random variables, the required probability is upper bounded by $e^{-(\alpha^{2}/2\pi)}/2.$
 \endproof 
 
\subsection{Probability of error}

\label{sec:noisyerror} An error occurs only if one of the following
take place:
\begin{enumerate}
\item The underlying graph $\graph$ is not an $\cSx$-expander. This probability
can be made $o(1/\spar)$ by choosing $\meas=\cnst\spar$, where the
constant $\cnst$ is determined by Lemma~\ref{lem:expansion}
\item The phase noise in $\tilde{y}_{\yi\ofiter}\ofiter$ leads to an incorrect
decoding of $\hat{\theta}_{\iter}^{(I,\step)}$ or $\hat{\theta}_{\iter}^{(V,\step)}$
for some $\step$ and $\iter$.

Note that the phase noise in $\tilde{y}_{\yi\ofiter}\ofiter$ consists:
\begin{enumerate}[(a)]
\item The contribution due to noise vectors $\bz$ and $\be$, and
\item The contribution due to the noise propagated while computing each
$\tilde{y}_{\yi\ofiter}(\tau)$ from $\tilde{y}_{\yi\ofiter}(\tau-1)$
for $\tau\leq\iter$.
\end{enumerate}

The contribution due to the first term is bounded by Lemma~\ref{lem:phasenoise}.
Thus, for a target error probability $\epsilon'$, we choose $\alpha=\sqrt{2\pi\log{1/2\epsilon'}}$,
giving a contribution to the phase noise of at most
\[
2\pi\sqrt{\frac{\log{(1/2\epsilon')}\spar^{2}(\dgr\samp\sigma_{z}^{2}/\cnst\spar+\sigma_{e}^{2})}{\delta^{2}}}.
\]
 To bound the contribution due to the second term, we note a few facts about the random graph $\graph$. Let $\graph_{\bx}$ be the restriction of $\graph$ to $\cSx$ and its neighbours. Denote the smallest disjoint components  of $\graph_{\bx}$ by ${\cal C}_{\bx}(1),{\cal C}_{\bx}(2),\ldots,{\cal C}_{\bx}(M)$ and let the number of right nodes in component ${\cal C}_{\bx}(p)$ be $D_{\bx}(p)$. The following properties of the random sparse graph $\graph_{\bx}$ and its components follow from~\cite{KarL:02,Pri:11}.
\begin{lemma}[\cite{KarL:02,Pri:11}]\label{lem:components}The random graph $\graph_\bx$ satisfies the following properties:
\begin{enumerate}[A.]
\item\label{prop:unicycle} For a large enough choice of $\cnst$,  with probability
$1-o(1/\spar)$, $\graph_{\bx}$ consists almost entirely of hypertrees and unicyclic components.
\item $\max_{p}D_\bx(p)=\cO(\log{\spar})$ with probability $1-o(1/\spar)$.
\item $E_{\graph}\left((D_{\bx}(p))^2\right)=O(1)$. \end{enumerate}
\end{lemma}
 
Now, we observe that at each
iteration $\iter$, any error in reconstruction of $\hat{x}_{\xj\ofiter}$
potentially adds to reconstruction error in all future iterations
$\iter'$ for which there is a path from $\xj\ofiter$ to $\xj(\iter')$. Thus, if $\xj(\iter)$ lies in the component ${\cal C}_{\bx}(p_\iter)$, then from Property~A above, the magnitude error in reconstruction of $\hat{x}_{\xj\ofiter}$ due to noisy reconstructions in previous iterations is upper bounded by 
\begin{equation}
(D_{\bx}(p_\iter))^2\sqrt{{2\pi\log(1/2\epsilon')(\samp\sigma_{z}^{2}/\spar+\sigma_{e}^{2})}}\label{eq:xnoise} 
\end{equation}
with probability at least $1-D_{\bx}(p_\iter)\epsilon'$. Thus, the phase displacement in each $y_{\yi}^{(I,\step)}$ and $y_{\yi}^{(V,\step)}$
is at most 
\[
2\pi(D_{\bx}(p_\iter))^2\sqrt{\frac{\log{(1/2\epsilon')}\spar^{2}(\samp\sigma_{z}^{2}/\spar+\sigma_{e}^{2})}{\delta^{2}}}.
\]
Next, applying Property~B, as long as
\begin{equation}
(\log\spar)^{2}\sqrt{\frac{2\pi\log{(1/2\epsilon')}\spar^{2}(\samp\sigma_{z}^{2}/\spar+\sigma_{e}^{2})}{\delta^{2}}}=o\left({\samp^{-1/\steps}}\right),
\end{equation}
 the probability of any single phase being incorrectly detected is
upper bounded by $\epsilon'$. Since we there are a total of $8\steps\spar$
possible phase measurements, we choose $\epsilon'=1/\steps\spar^{2}$
to achieve an overall target error probability $1/\spar$.

\item The verification step passes for each measurement in the $\iter$-th
measurement, even though $\yi\ofiter$ is not a leaf node for $\cS_{\delta}^{c}(\bx)$.
\item $\cD{(T)}\neq\bA'$, i.e., the algorithm terminates without recovering
all $x_{\xj}$'s. Note that similar to the exact $\spar$-sparse case,
in each iteration $\iter$, by Lemma~\ref{lem:manyleafs}, the probability
that $\yi\ofiter$ is a leaf node for $\cS_{\delta}(\bx-\hat{\bx}\ofiter)$
at least $1/2$. However, due to noise, there is a non-zero probability
that even when $\yi\ofiter$ is a leaf node, it does not pass the
verification tests. We know from the analysis for the previous case
that this probability is $\cO(1/\spar)$ for each $\yi\ofiter$. Therefore,
the probability that a randomly picked $\yi\ofiter$ passes the verification
test is $1/2-\cO(1/\spar)$. Thus, in expectation, the number of iterations
required by the algorithm is $2\spar/(1-\cO(1/\spar))$. By concentration
arguments, it follows that the probability that the algorithm does
not terminate in $4\spar$ iterations is $o(1/\spar)$ as $\spar$
grows without bound.
\end{enumerate}

\subsection{Estimation error}

Next, we bound the error in estimating $\hat{\bx}$. We first find
an upper bound on $||\hat{\bx}-\bx_{\cS_{\delta}^{c}}||_{1}$ that
holds with a high probability. Applying the bound in~\eqref{eq:xnoise},
for each $\iter=1,2,\ldots,T$,
\[
|x_{\xj\ofiter}-\hat{x}_{\xj\ofiter}|=\cO\left((D_{\bx}(p_\iter))^2\sqrt{{2\pi\log(1/2\epsilon')(\samp\sigma_{z}^{2}/\spar+\sigma_{e}^{2})}}\right)
\]
 with probability $1-\cO(1/\spar)$. Therefore, with probability $1-\cO(1/\spar)$,
\begin{eqnarray}
||\hat{\bx}-\bx_{\cS_{\delta}^{c}}||_{1} & = & \sum_{\substack{1\leq\iter\leq T\\
t:\xj\ofiter\notin\cS_{\delta}
}
}|\hat{x}_{\xj}-x_{\xj}|+\sum_{\substack{1\leq\iter\leq T\\
t:\xj\ofiter\in\cS_{\delta}
}
}|\hat{x}_{\xj}|\nonumber \\
 & \leq & \sum_{\xj\notin\cS_{\delta}}|\hat{x}_{\xj}-x_{\xj}|+\sum_{\xj\in\cS_{\delta}}|\hat{x}_{\xj}-x_{\xj}|+\sum_{\xj\in\cS_{\delta}}|x_{\xj}|\nonumber \\
 & = &\cO\left(\sum_{p=1}^P\sum_{\xj\in\mathcal{C}(p)}(D_{\bx}(p_\iter))^2\sqrt{{2\pi\log(1/2\epsilon')(\samp\sigma_{z}^{2}/\spar+\sigma_{e}^{2})}}\right)+\delta.\label{eq:errorbound1}
\end{eqnarray}
 Next, note that $||\bz||_{1}=\sum_{\xj=1}^{\samp}|z_{\xj}|$ and
$||\be||_{1}=\sum_{\yi=1}^{\meas}|e_{\yi}|$. Since each $z_{\xj}$
is a Gaussian random variable with variance $\sigma_{z}^{2}$, The
expected value of $|z_{\xj}|$ is $\sigma_{z}\sqrt{2/\pi}$. Therefore,
for every $\epsilon'>0$, for $\samp$ large enough,
\begin{equation}
\Pr(||\bz||_{1}<(1/2)\samp\sigma_{z}\sqrt{2/\pi})<\epsilon'.\label{eq:zbound}
\end{equation}
 Similarly, for $\meas$ large enough,
\begin{equation}
\Pr(||\be||_{1}<(1/2)\cnst\spar\sigma_{e}\sqrt{2/\pi})<\epsilon'.\label{eq:ebound}
\end{equation}
Combining inequalities~\eqref{eq:errorbound1}-~\eqref{eq:ebound} and Property~C of Lemma~\ref{lem:components},
we have, with a high probability,
\begin{eqnarray}
E\left(||\hat{\bx}-\bx_{\cS_{\delta}^{c}}||_{1}\right) & = & \cO\left(\spar\sqrt{\log(1/\epsilon')}\left(\frac{||\bz||_{1}}{\sqrt{\samp\spar}}+\frac{||\be||_{1}}{\spar}\right)\right)+\delta\nonumber \\
 & = & \cO\left(\sqrt{\frac{\spar}{\samp}}\sqrt{\log{(1/\epsilon')}}||\bz||_{1}+\sqrt{\log{(1/\epsilon')}}||\be||_{1}\right)+\delta.
\end{eqnarray}

Next, applying the bound in~\eqref{eq:deltabound}, we obtain
\begin{eqnarray}
E\left(||\hat{\bx}-\bx||_{1}\right)&=&\cO\left(\sqrt{\frac{\spar}{\samp}}\sqrt{\log{(1/\epsilon')}}||\bz||_{1}+\sqrt{\log{(1/\epsilon')}}||\be||_{1}\right)+2\delta\nonumber\\
&=& \cO(\sqrt{\frac{\spar\log\spar}{\samp}}||\bz||_1+\sqrt{\log\spar}||\be||_1) \label{eq:errorbound}
\end{eqnarray}
 with a high probability.

\subsection{Proof of Theorem~\ref{thm:approximate}}
Finally, to complete the proof of Theorem~\ref{thm:approximate}, we let $\delta=\min\{\cO(n\sigma_z),o(1)\}$. By~\eqref{eq:zbound} with a high probability, $\delta=\cO(||\bz||)$. Finally, recall the assumption that $\spar=\cO(\samp^{1-\Delta})$. Applying these to the bound obtained in  \eqref{eq:errorbound}, we get $$||\hat{\bx}-\bx||_1\leq C\left(||\bz||_1+\sqrt{\log{\spar}}||\be||_1\right)$$ for an appropriate constant $C=C(\sigma_z,\sigma+e)$.

\subsection{Simulation Results}\label{app:sim}
This section describes simulations that use synthetic data. The $\spar$-sparse signals used here are generated by randomly choosing $\spar$
  locations for non-zero values and setting the non-zero values to 1. The contours in each plot show the probability of successful reconstruction (the lighter the color,
the higher the probability of reconstruction). The probability of error at each data point in the plots was obtained by running multiple simulations ($400$ in Fig~\ref{fig:sim1} and Fig~\ref{fig:sim2}, and $200$ in Fig~\ref{fig:sim3}) and noting the fraction of simulations which resulted in successful reconstruction.

\begin{figure*}[h]
\begin{minipage}[b]{0.47\linewidth}\centering\includegraphics[width=\linewidth]{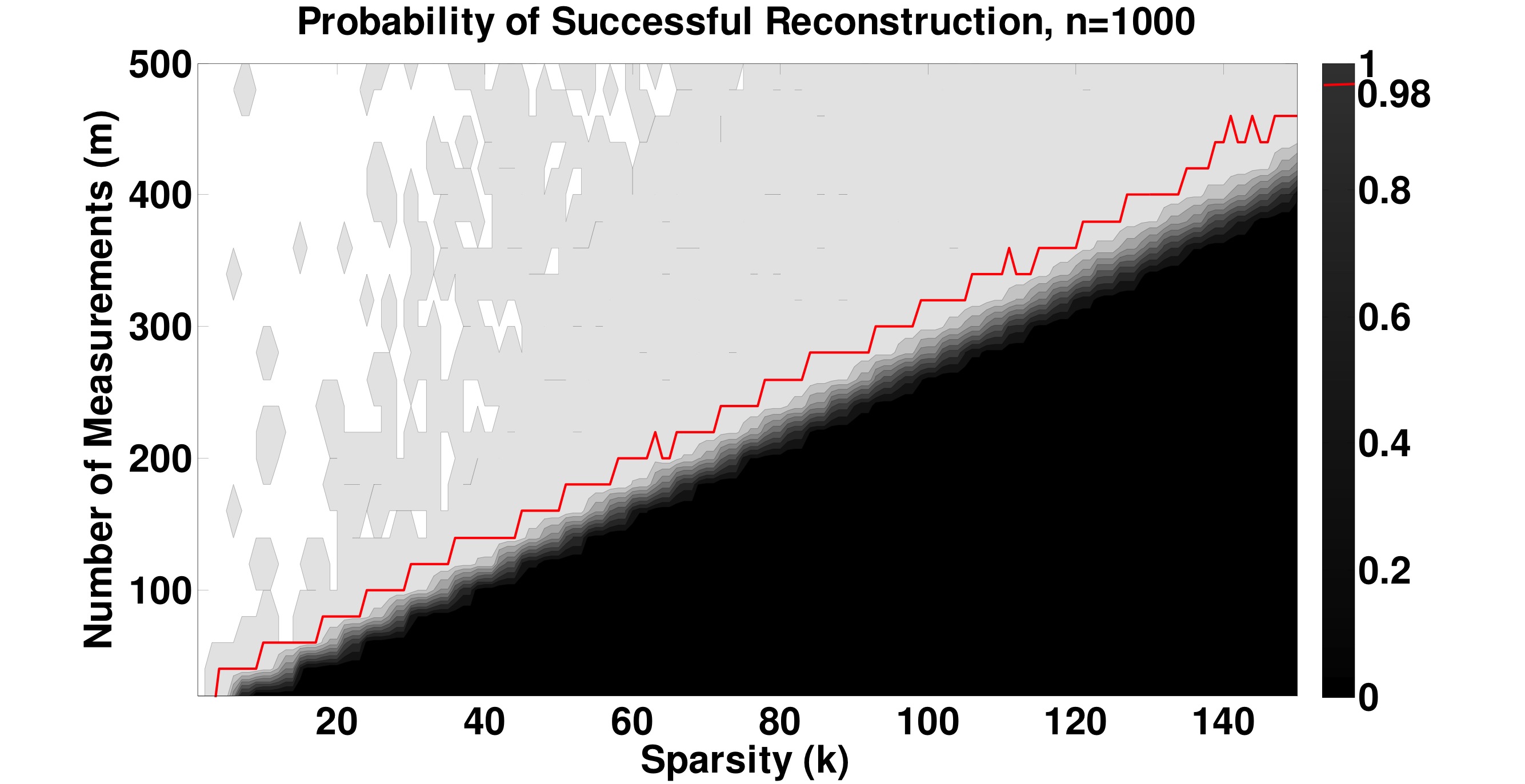}
\caption{
{\bf{Exactly sparse signal and noiseless measurements -- reconstruction performance for fixed signal length $\samp$:}} The $y$-axis denotes the number of measurements $\meas$, and the $x$-axis denotes the sparsity $\spar$, for fixed signal length $\samp=1000$. The simulation results show that the number of measurements
$\meas$ grows roughly proportional to the sparsity $\spar$ for a fixed probability of reconstruction error. Also note
that there is a sharp transition in reconstruction performance once the number
of measurements exceeds a linear multiple of $\spar$. The red line denotes
the curve where the probability of successful reconstruction equals
$0.98$. For $\spar=150$, the probability of success equals $0.98$ when $\meas=450$
and $c=\meas/\spar=3$.
}
\label{fig:sim1}
\end{minipage}
\hspace{0.05\linewidth}\begin{minipage}[b]{0.47\linewidth}\centering\includegraphics[width=\linewidth]{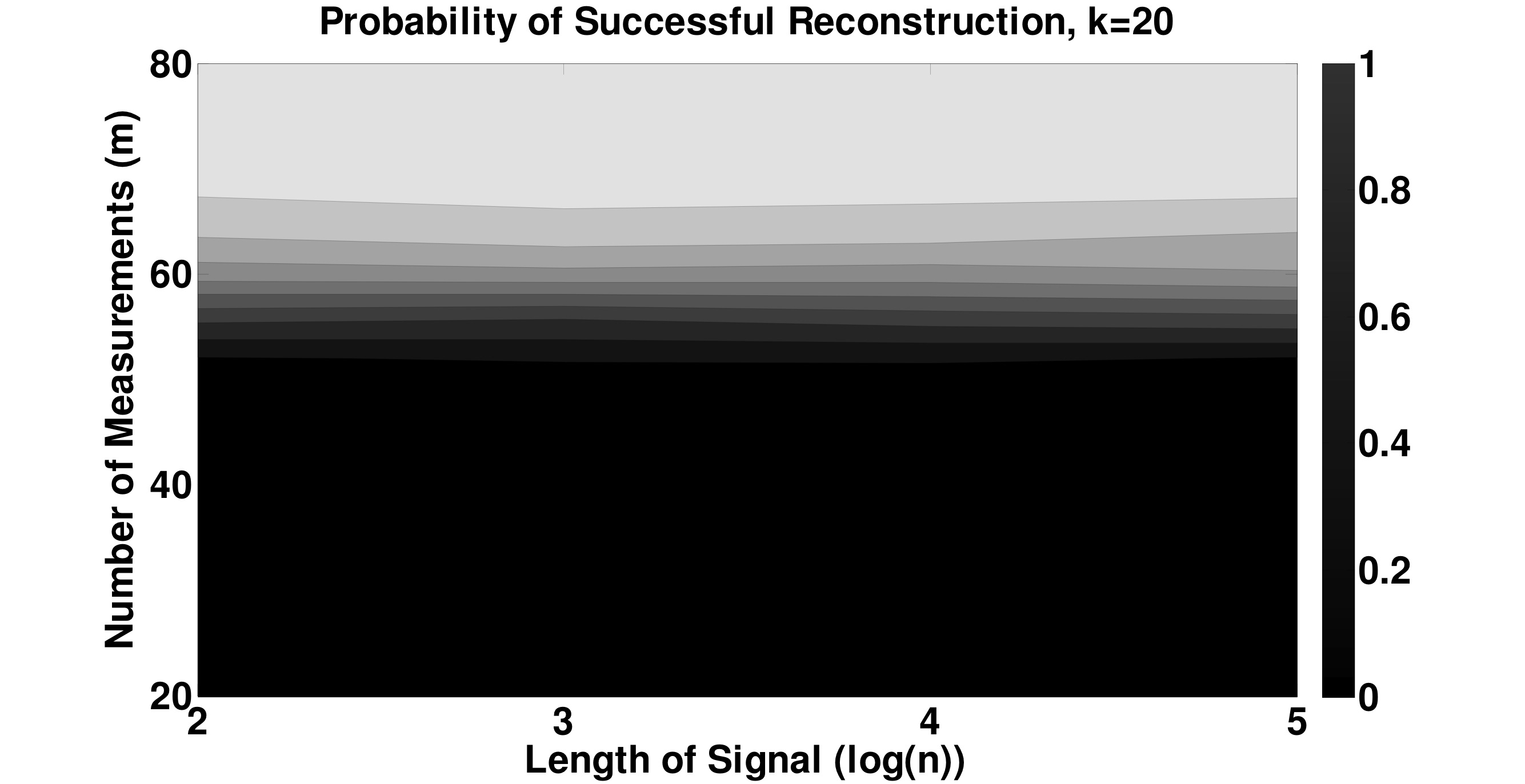}
\caption{
{\bf{Exactly sparse signal and noiseless measurements -- reconstruction performance for fixed sparsity $\spar$:} } The number of measurements $\meas$ are plotted on the $y$-axis, plotted against $\log(\samp)$ on the $x$-axis -- the
sparsity $\spar$ is fixed to be $20$. Note that there is ${\it no}$ scaling
of $\meas$ with $\samp$, as guaranteed by our theoretical bounds.
}
\label{fig:sim2}\end{minipage}
\end{figure*}

\begin{figure*}[h]
\centering\includegraphics[scale=.35]{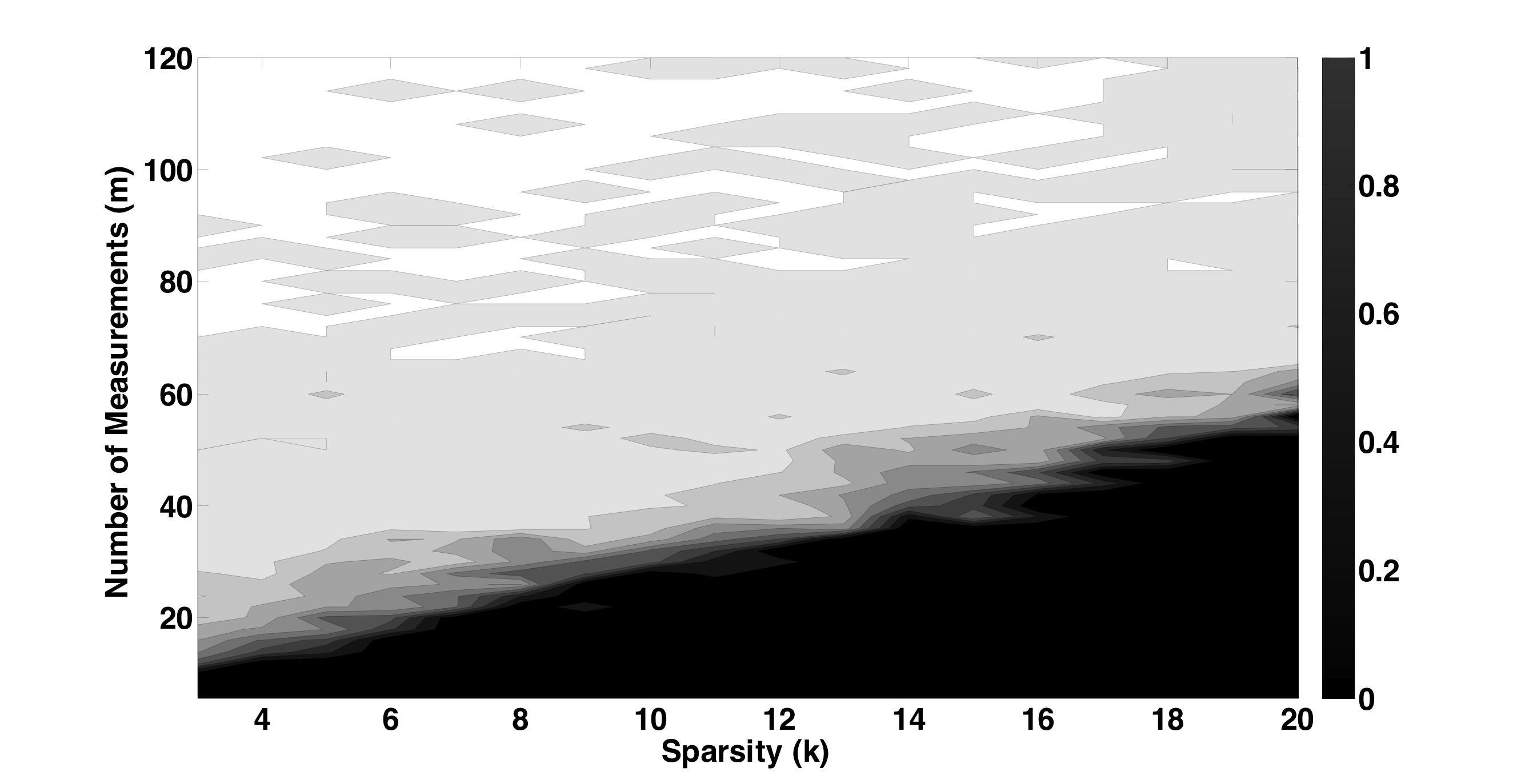}
\caption{{\bf Approximately sparse signal and noisy measurements -- reconstruction performance for fixed signal-length $\samp$:} As in Fig~\ref{fig:sim1}, the $y$-axis denotes the number of measurements $\meas$, and the $x$-axis denotes the sparsity $\spar$, for fixed signal length $\samp=1000$. In this case, we set $\sigma_z = 0.03$, and allowed relative reconstruction error of at most $0.3$. }
\label{fig:sim3}
\end{figure*}

\end{document}